\documentclass[a4paper,10pt]{article}
\usepackage{amsthm,amssymb, amsmath}
\usepackage{graphicx}
\usepackage{graphics}
\usepackage{times}
\usepackage{amsmath}
\usepackage{amsfonts}
\usepackage{color}

\newtheorem{Def}{Definition}
\newtheorem{Thm}{Theorem}
\newtheorem{Lem}{Lemma}
\newtheorem{Cor}{Corollary}
\newcommand{\mb}{\mathbf}
\newcommand{\mat}[2][ccccccccccccccccccccccccccc]{\left( \begin{array}
{#1}#2 \\ \end{array} \right)}

\begin{document}

\title{\textbf{Rate Region of the Vector Gaussian One-Helper Source-Coding Problem}}

\date{}

\author{Md. Saifur Rahman and Aaron B. Wagner$^{\footnote{Both authors are with the School of Electrical and Computer Engineering, Cornell University, Ithaca, NY 14853 USA. (Email: mr534@cornell.edu, wagner@ece.cornell.edu.)}}$}

\maketitle

\begin{abstract}
We determine the rate region of the vector Gaussian one-helper source-coding problem under a covariance matrix distortion constraint.  The rate region is achieved by a simple scheme that separates the lossy vector quantization from the lossless spatial compression. The converse is established by extending and combining three analysis techniques that have been employed in the past to obtain partial results for the problem.
\end{abstract}

\textbf{Keywords:} multiterminal source coding, one-helper problem, covariance matrix distortion constraint, vector Gaussian sources, vector quantization, distortion projection, source enhancement.

\section{Introduction} \label{sec:Intro}
We study the vector Gaussian one-helper source-coding problem\footnote{The material in this paper was presented in part at the 49th Annual Allerton Conference on Communications, Control, and Computing, University of Illinois, Urbana-Champaign, Sept. 2011.}, depicted in Fig. \ref{fig:Fig1}. Here $\mb{X}$ and $\mb{Y}$ are two jointly vector Gaussian sources. Encoders 1 and 2 observe two i.i.d. strings distributed according to $\mb{X}$ and $\mb{Y}$, respectively, and separately send messages to the decoder at rates $R_1$ and $R_2$ bits per observation, respectively, using noiseless channels. The decoder uses both messages to estimate $\mb{X}$ such that a given distortion constraint on the average error covariance matrix is satisfied. The goal is to determine the rate region of the problem, which is the set of all rate pairs $(R_1,R_2)$ that allow us to satisfy the distortion constraint for some design of the encoders and the decoder.
\begin{figure}[htp]
\centering
 \includegraphics[width=3.0in]{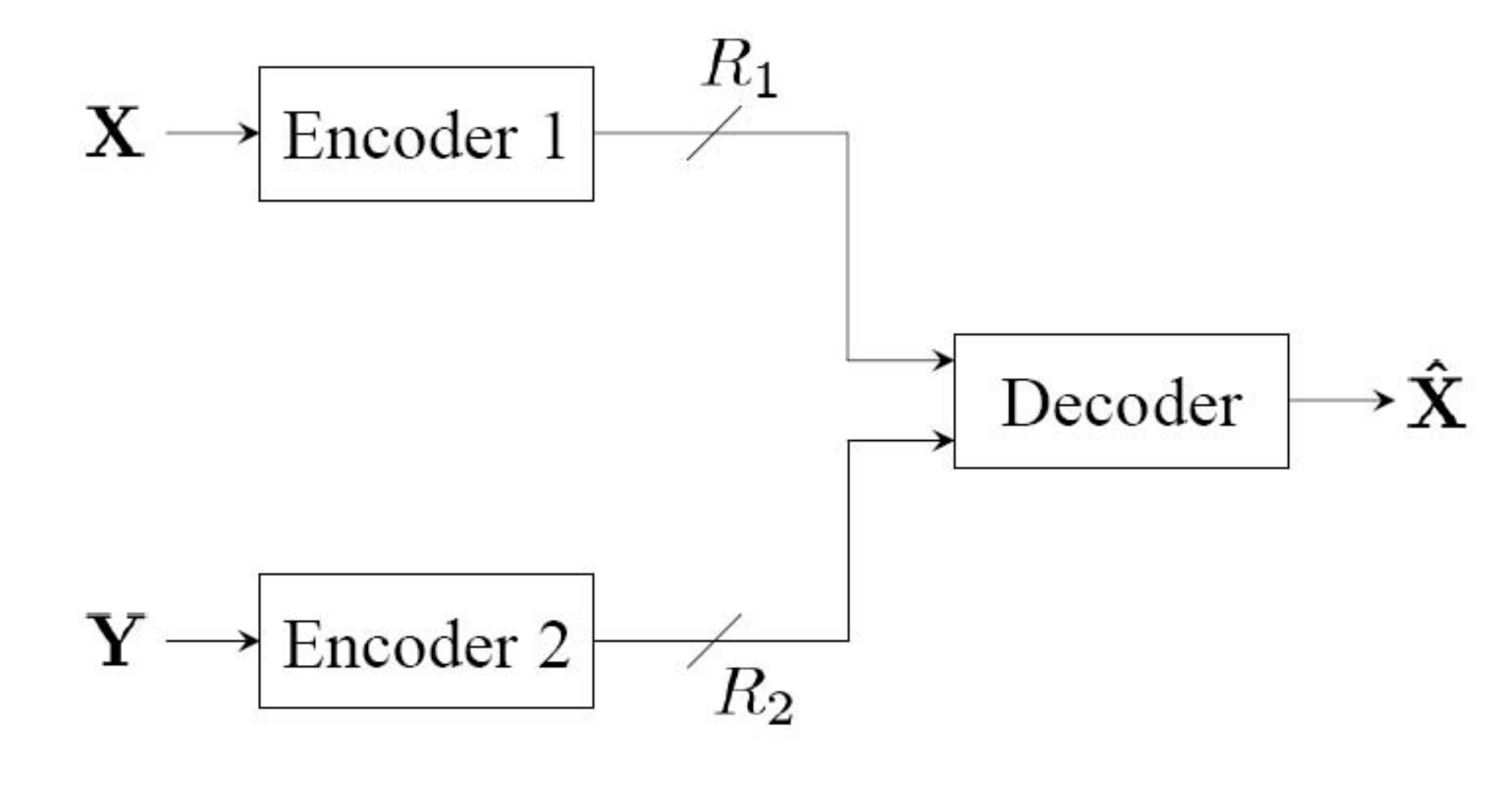}
\caption{Vector Gaussian one-helper source-coding problem.}\label{fig:Fig1}
\end{figure}

Oohama~\cite{Oohama} gave a complete characterization of the rate region
for the case in which both sources are scalar. His
achievability proof is a Gaussian scheme that is described in
more detail below. The converse argument uses the entropy-maximizing
property of the
Gaussian distribution and the entropy power inequality (EPI), and it
bears a certain resemblance
to Bergmans' earlier converse for the scalar Gaussian broadcast
channel~\cite{Bergmans}. As such, one might hope that the
\emph{channel enhancement} technique introduced by Weingarten \emph{et al.} \cite{Wein}
to solve the MIMO Gaussian broadcast channel would be sufficient
to solve the problem considered here. This turns out not to be
the case, however. Among other contributions, Liu and
Viswanath~\cite{Liu} showed that \emph{channel enhancement} yields an
outer bound for the vector one-helper problem that is not tight
in general. This was later improved slightly by the present
authors to show that the Gaussian scheme achieves a portion
of the boundary of the rate region \cite{Rahman}.
Liu and Viswanath's approach was later subsumed by
Zhang~\cite{Guo}, who applied enhacement in a different
way and called it \emph{source enhancement},
but this also yielded an outer bound that is not always tight.

\begin{figure}[htp]
\centering
\includegraphics[width=3.0in]{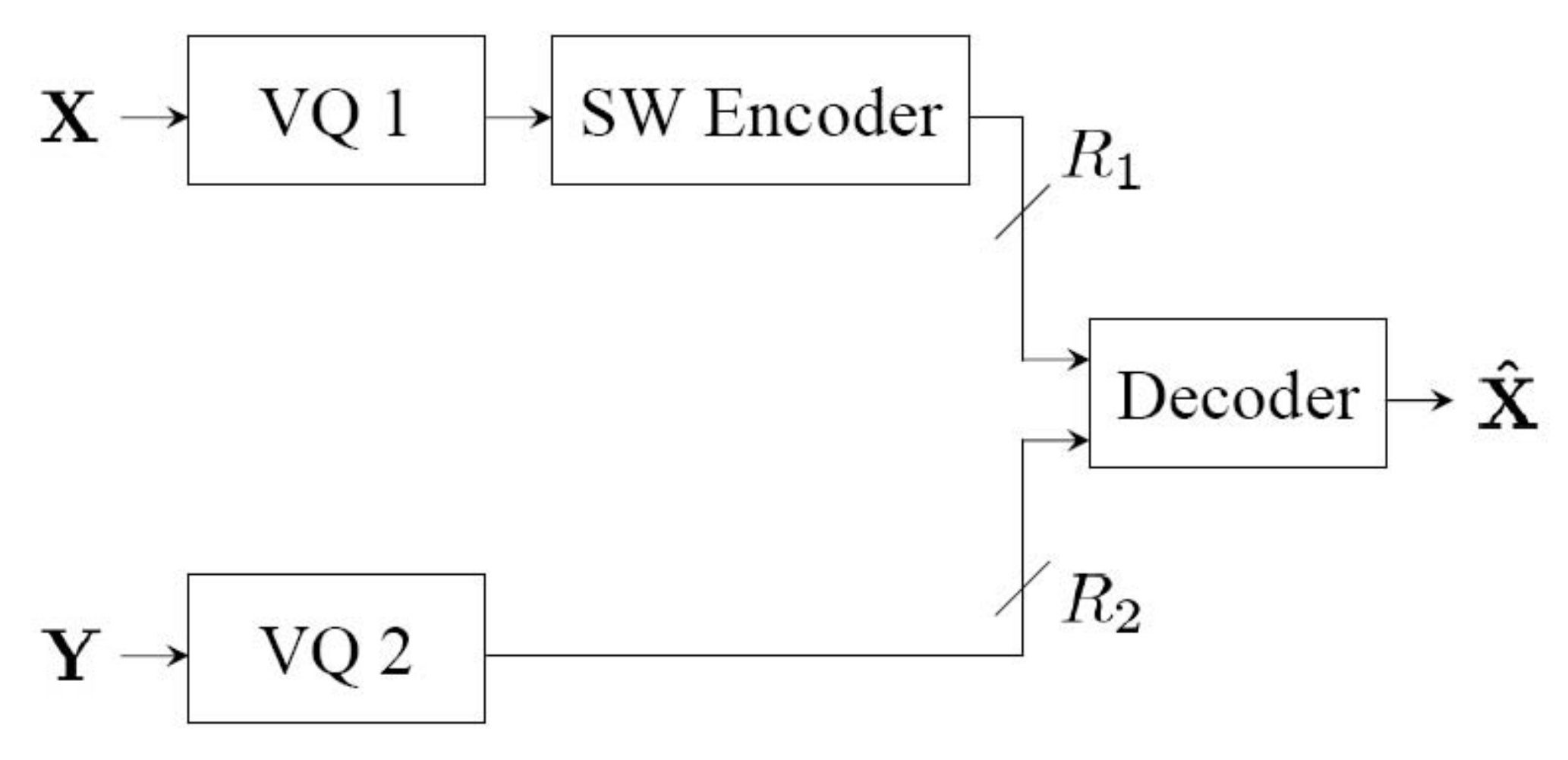}
\caption{A Gaussian achievable scheme.}\label{fig:Fig2}
\end{figure}

The case in which $Y$ is a scalar and $\mb{X}$ is a vector
was recently solved by the authors~\cite{Rahman1}. The proof did not use
enhancement, but it did require a novel technique that
we call \emph{distortion projection.} Here we shall show that
\emph{distortion projection}, \emph{source enhancement}, and
Oohama's converse technique together are sufficient to solve the
general problem in which both $\mb{X}$ and $\mb{Y}$ are vectors.
In particular, we shall determine the rate region exactly
and show that a vector extension of the Gaussian scheme
used by Oohama is optimal. In this scheme, as depicted in Fig. \ref{fig:Fig2}, encoder 1 vector quantizes (VQ) its observations using a Gaussian test channel as in point-to-point rate-distortion theory. It then compresses the quantized values using Slepian-Wolf (SW) encoding \cite{Slepian}. Encoder 2 just vector quantizes its observations using another Gaussian test channel. The decoder decodes the quantized values and estimates the observations of encoder 1 using a minimum mean-squared error (MMSE) estimator.

The rest of the paper is organized as follows. Section \ref{sec:Notation} explains the notation used in the paper. In Section \ref{sec:ProblemFormulationResults}, we present the mathematical formulation of the problem, a description of the scheme, and the statement of our main result. Section \ref{sec:ConverseOverview} gives an outline of the converse argument. Since the proof of the converse is somewhat involved, it is divided into Sections \ref{sec:PropOptGaussSol}
through \ref{sec:Extension}.
%, \ref{sec:ConvIngre}, and \ref{sec:ConvProof}. Finally in Section \ref{sec:Extension}, we complete the proof for the most general case of the problem.

\section{Notation}\label{sec:Notation}
We use uppercase to denote random variables and vectors. Boldface is used to distinguish vectors from scalars. Arbitrary realizations of random variables and vectors are denoted in lowercase. For a random vector $\mb{X}$, $\mb{X}^n$ denotes an i.i.d. vector of length $n$, $\mb{X}^n(i)$ denotes its \emph{i}th component, and $\mb{X}^n(i : j)$ denotes the \emph{i}th through \emph{j}th components. The superscript $T$ denotes matrix transpose. The covariance matrix of $\mb{X}$ is denoted by $\mb{K_X}$. The conditional covariance matrix of $\mb{X}$ given $\mb{Y}$ is denoted by $\mb{K}_{\mb{X}|\mb{Y}}$ and is defined as
\[
\mb{K}_{\mb{X}|\mb{Y}} \triangleq E \left [ \left (\mb{X} - E(\mb{X}|\mb{Y}) \right ) \left (\mb{X} - E(\mb{X}|\mb{Y}) \right )^T \right ].
\]
All vectors are column vectors and are $m$-dimensional, unless otherwise stated. We use $\mb{I}_m$ to denote an $m \times m$ identity matrix. With a little abuse of notation, $\mb{0}$ is used to denote both zero vectors and zero matrices of appropriate dimensions. We use $\textrm{Diag}  (d_1,d_2,\dots,d_p)$ to denote a diagonal matrix with diagonal entries $d_1,d_2,\dots,d_p$. The trace of a matrix $\mb{A}$ is denoted by $\textrm{Tr}(\mb{A})$. For two real symmetric matrices $\mb{A}$ and $\mb{B}$, $\mb{A} \succcurlyeq \mb{B}$ ($\mb{A} \succ \mb{B}$) means that $\mb{A-B}$ is positive semidefinite (definite). Similarly, $\mb{A} \preccurlyeq \mb{B}$ ($\mb{A} \prec \mb{B}$) means that $\mb{B-A}$ is positive semidefinite (definite). All logarithms in this paper are to the base 2. The determinant of a matrix $\mb{K}$ is denoted by $|\mb{K}|$. The notation $X \leftrightarrow Y \leftrightarrow Z$ means that $X, Y,$ and $Z$ form a Markov chain in this order. We use $\textrm{span}\{\mb{c}_i\}_{i=1}^l$ to denote the subspace spanned by $\{\mb{c}_i\}_{i=1}^l$. %$[\hspace{0.05in}]$ denotes the null matrix.

\section{Problem Formulation and Main Results}\label{sec:ProblemFormulationResults}

Let $\mb{X}$ and $\mb{Y}$ be two generic zero-mean jointly Gaussian random vectors with covariance matrices $\mb{K_X}$ and $\mb{K_Y}$, respectively.
Initially, we shall assume that $\mb{X}$ is $m$-dimensional and
$\mb{Y}$ is $k$-dimensional.
Let $\left \{\left(\mb{X}^n(i),\mb{Y}^n(i)\right) \right  \}_{i=1}^{n}$ be a sequence of i.i.d. random vectors with the distribution at a single stage being the same as that of the generic pair $(\mb{X},\mb{Y})$. As depicted in Fig. \ref{fig:Fig1}, encoder 1 observes $\mb{X}^n$ and sends a message to the decoder using an encoding function
\begin{align*}
f_1^{(n)} : \mathbb{R}^{mn} \mapsto \left \{1,\dots,M_1^{(n)} \right \}.
\end{align*}
Analogously, encoder 2 observes $\mb{Y}^n$ and sends a message to the decoder using another encoding function
\begin{align*}
f_2^{(n)}  : \mathbb{R}^{kn} \mapsto \left \{1,\dots,M_2^{(n)} \right \}.
\end{align*}
The decoder uses both received messages to estimate $\mb{X}^n$ using a decoding function
\begin{align*}
g^{(n)}  : \left \{1,\dots,M_1^{(n)} \right \} \times \left \{1,\dots,M_2^{(n)} \right \} \mapsto \mathbb{R}^{mn}.
\end{align*}
\begin{Def}
A rate-distortion vector $\left (R_1,R_2,\mb{D} \right )$ is \emph{achievable} for the vector Gaussian one-helper source-coding problem if there exist a block length $n$, encoding functions $f_1^{(n)} $ and $f_2^{(n)} $, and a decoding function $g^{(n)} $ such that
\begin{align*}
R_i &\ge \frac{1}{n} \log M_i^{(n)}  \hspace {0.15 cm} \textrm{for all} \hspace {0.15 cm} i \in \{1,2\}, \hspace {0.15 cm} \textrm{and}\nonumber\\
\mb{D} &\succcurlyeq \frac{1}{n} \sum_{i=1}^n E \left [ \left (\mb{X}^n(i) - \hat {\mb{X}}^n(i) \right ) \left (\mb{X}^n(i) - \hat {\mb{X}}^n(i) \right )^T \right ],
\end{align*}
where
\begin{align*}
\hat {\mb{X}}^n \triangleq g^{(n)} \left (f_1^{(n)} \left (\mb{X}^n \right ), f_2^{(n)}\left (\mb{Y}^n \right ) \right ).
\end{align*}
Let ${\mathcal{RD}}$ be the set of all achievable rate-distortion vectors and $\overline {\mathcal{RD}}$ be its closure. Define
\[
\mathcal{R}\left (\mb{D} \right) \triangleq \left \{(R_1,R_2): (R_1,R_2,\mb{D}) \in \overline {\mathcal{RD}} \right \}.
\]
We call $\mathcal{R} (\mb{D})$ the \emph{rate region} for the vector Gaussian one-helper source-coding problem.
\end{Def}
Our goal is to characterize the rate region $\mathcal{R} (\mb{D})$. Note that the matrix distortion constraint is more general in the sense that it subsumes other natural distortion constraints such as a finite number of upper bounds on the mean square error of reproductions of linear functions of the source. In particular, it subsumes the case in which the distortion constraint is on the mean square error of reproductions of the components of $\mb{X}$.

Since we are interested in a quadratic distortion constraint, without loss of generality we can restrict the decoding function to be the MMSE estimate of $\mb{X}^n$ based on the received messages. Therefore, $\hat {\mb{X}}^n$ can be written as
\begin{align*}
\hat {\mb{X}}^n = E \left [\mb{X}^n \bigr| f_1^{(n)} \left (\mb{X}^n \right ), f_2^{(n)}\left (\mb{Y}^n \right )\right ].
\end{align*}
We can assume without loss of generality\footnote{Since $\mb{X}$ and $\mb{Y}$ are jointly Gaussian, we can write $$\mb{X} = \mb{A}\mb{Y}+\mb{N},$$ where $\mb{A}$ is an $m \times k$ matrix and $\mb{N}$ is an $m$-dimensional zero-mean Gaussian random vector that is independent of $\mb{Y}$. Since there is no distortion constraint on $\mb{Y},$ and $\mb{AY}$ is a sufficient statistic for $\mb{X}$ given $\mb{Y}$ (i.e., $\mb{X} \leftrightarrow \mb{Y} \leftrightarrow \mb{AY}$ and $\mb{X} \leftrightarrow \mb{AY} \leftrightarrow \mb{Y}$), we can relabel $\mb{AY}$ as $\mb{Y}$ and write $$\mb{X} = \mb{Y}+\mb{N}.$$} that
\begin{align*}
\mb{X}= \mb{Y}+\mb{N},
\end{align*}
where $\mb{N}$ is a zero-mean Gaussian random vector with the covariance matrix $\mb{K_N}$ and is independent of $\mb{Y}$. The case in which $\mb{K_X} \preccurlyeq \mb{D}$ has a trivial solution. In this case, the rate region is the entire nonnegative quadrant. So, we assume that $\mb{K_X} \preccurlyeq \mb{D}$ does not hold in the rest of the paper. This means that there exists a direction $\mb{z} \neq \mb{0}$ such that
\begin{equation}
\label{eq:assumption}
\mb{z}^T \mb{K_Xz} > \mb{z}^T \mb{Dz}.
\end{equation}
For now, we assume that $\mb{K_X}, \mb{K_Y},$ and $\mb{D}$ are positive definite. The general case of the problem will be addressed in Section \ref{sec:Extension}.

\subsection{Rate Region} \label{subsec:rateregion}
The rate region $\mathcal{R}(\mb{D})$ is a closed convex set in the nonnegative quadrant. It is closed by definition and is convex because any convex combination of two points in the rate region is in the rate region as it can be achieved by time-sharing between the encoding and decoding strategies of the two points. Therefore, we can characterize it completely by its supporting hyperplanes, which can be expressed as the following optimization problem
\[
\mathcal{R}(\mb{D}, \mu) \triangleq \inf_{(R_1,R_2) \in \mathcal{R}(\mb{D})} \mu R_1+R_2,
\]
where $\mu$ is a nonnegative real number. Let us define
\[
\mathcal{R}^{*}(\mb{D}, \mu) \triangleq \left\{
\begin{array}{l l}
  v\left (P_{pt-pt} \right) & \quad \mbox{if $0 \le \mu \le 1$}\\
  v\left (P_{G1} \right) & \quad \mbox{if $\mu > 1$,}\\ \end{array} \right.
\]
where $v\left (P_{pt-pt} \right)$ and $v\left (P_{G1} \right)$ are the optimal values of the optimization problems $\left (P_{pt-pt} \right)$ and $\left (P_{G1} \right)$, respectively, which are defined as
\begin{align*}
\left (P_{pt-pt} \right)\hspace {0.2in} \triangleq \hspace {0.2in} \min_{\mb{K}_{\mb{X}|\mb{U}}} \hspace {0.3in} &\frac{\mu}{2} \log \frac{ \left |\mb{K_X} \right |}{ \left|\mb{K}_{\mb{X}|\mb{U}}\right|} \nonumber\\
\textrm{subject to} \hspace{0.2in} &\mb{K}_{\mb{X}} \succcurlyeq \mb{K}_{\mb{X}|\mb{U}} \succcurlyeq \mb{0} \hspace{0.1in} \textrm{and}\nonumber\\
&\mb{D} \succcurlyeq \mb{K}_{\mb{X}|\mb{U}},\nonumber
\end{align*}
and %$v\left (P_{G1} \right)$ is the optimal value of the optimization problem $\left (P_{G1} \right)$ defined as
\begin{align*}
\left (P_{G1} \right)\hspace {0.2in} \triangleq \hspace {0.2in} \min_{\mb{K}_{\mb{Y}|\mb{V}},\mb{K}_{\mb{X}|\mb{U,V}}} \hspace {0.2in} &\frac{\mu}{2} \log \frac{ \left |\mb{K}_{\mb{Y}|\mb{V}}+\mb{K_N} \right |}{ \left|\mb{K}_{\mb{X}|\mb{U,V}}\right|} + \frac{1}{2} \log \frac{ \left |\mb{K_Y} \right |}{ \left|\mb{K}_{\mb{Y}|\mb{V}}\right|}\nonumber\\
\textrm{subject to} \hspace{0.3in} &\mb{K}_{\mb{Y}} \succcurlyeq \mb{K}_{\mb{Y}|\mb{V}}  \succcurlyeq \mb{0},\\
&\mb{K}_{\mb{Y}|\mb{V}}+\mb{K_N} \succcurlyeq \mb{K}_{\mb{X}|\mb{U,V}} \succcurlyeq \mb{0}, \hspace{0.1in} \textrm{and}\nonumber\\
&\mb{D} \succcurlyeq \mb{K}_{\mb{X}|\mb{U,V}}.\nonumber
\end{align*}
We use similar notation to denote other optimization problems and their optimal values throughout the paper. The main result of this paper is the following theorem.
\begin{Thm} \label{thm:MainThm}
The minimum weighted sum rate for the vector Gaussian one-helper source
coding problem is given by the solution to the above matrix optimization
problem
\begin{align*}
\mathcal{R}(\mb{D}, \mu) = \mathcal{R}^{*}(\mb{D},\mu).
\end{align*}
\end{Thm}

\subsection{A Gaussian Achievable Scheme} \label{subsec:scheme}In this subsection, we present a {Gaussian achievable scheme} (Fig. \ref{fig:Fig2}). The scheme is well-known and is sometimes referred to as the Berger-Tung scheme \cite{Berger,Tung}. This scheme is known to be optimal for several problems in Gaussian multiterminal source-coding literature \cite{Oohama, Zhang, Oohama2005, Vinod, Viswanathan, Wagner, Wagner1}. However, it is not optimal in some cases. For instance, a lattice-based scheme can outperform it if the goal is to reconstruct a hidden random vector that is jointly Gaussian with $\mb{X}$ and $\mb{Y}$~\cite{WagnerLinear, Pradhan}, and the discrete memoryless version of the scheme can be suboptimal if the sources have common components~\cite{WagnerCommon}. For the problem under consideration however, we shall prove that the Berger-Tung scheme is indeed optimal. We present an overview of the scheme here. The details for similar problem setups can be found in \cite{Oohama,Zhang}.

Let $\mathcal{S}$ be the set of zero-mean jointly Gaussian random vectors $\mb{U}$ and $\mb{V}$ such that
\begin{enumerate}
\item[(C1)] $\mb{U}, \mb{X}$, $\mb{Y}$, and $\mb{V}$ form a Markov chain $\mb{U} \leftrightarrow \mb{X} \leftrightarrow \mb{Y} \leftrightarrow \mb{V}$, and
\item[(C2)] $\mb{K}_{\mb{X|U,V}} \preccurlyeq \mb{D}$.
\end{enumerate}
Consider any $(\mb{U,V}) \in \mathcal{S}$ and a large block length $n$. Let $R_1^{'} \triangleq I(\mb{X};\mb{U})+\epsilon$, where $\epsilon > 0$. To construct the codebook for encoder 1, first generate $2^{n R_1^{'}}$ independent codewords $\mb{U}^n$ randomly according to the marginal distribution of $\mb{U}$, and then uniformly distribute them into $2^{n R_1}$ bins. Encoder 2's codebook is constructed by generating $2^{n R_2}$ independent codewords $\mb{V}^n$ randomly according to the marginal distribution of $\mb{V}$.

Given a source sequence $\mb{X}^n$, encoder 1 looks for a codeword $\mb{U}^n$ that is jointly typical with $\mb{X}^n$, and sends the index $b$ of the bin to which $\mb{U}^n$ belongs. Encoder 2, upon observing $\mb{Y}^n$, sends the index of the codeword $\mb{V}^n$ that is jointly typical with $\mb{Y}^n$. The decoder receives the two indices, then looks into the bin $b$ for a codeword $\mb{U}^n$ that is jointly typical with $\mb{V}^n$. The decoder can recover $\mb{U}^n$ and $\mb{V}^n$ with high probability as long as
\begin{align*}
R_1 &\ge I(\mb{X};\mb{U}|\mb{V}) \hspace{0.1in} \textrm{and}\\
R_2 &\ge I(\mb{Y};\mb{V}).
\end{align*}
The decoder then computes the MMSE estimate of the source $\mb{X}^n$ given the messages $\mb{U}^n$ and $\mb{V}^n$, and (C2) above guarantees that this estimate will satisfy the covariance matrix distortion constraint. Let
\begin{align*}
\mathcal{R}_G(\mb{D}) \triangleq \bigr \{(R_1,R_2) &: \hspace{0.05in} \textrm{there exists}\hspace{0.05in} (\mb{U,V}) \in \mathcal{S} \hspace{0.05in}\textrm{such that} \\
R_1 &\ge I(\mb{X};\mb{U}|\mb{V}) \hspace{0.1in} \textrm{and} \\
R_2 &\ge I(\mb{Y};\mb{V}) \bigr\}.
\end{align*}
Furthermore, define
\[
\mathcal{R}_G(\mb{D}, \mu) \triangleq \min_{(R_1,R_2) \in \mathcal{R}_G(\mb{D})} \mu R_1+R_2.
\]
The following lemma gives the weighted sum-rate achieved by this scheme.
\begin{Lem} \label{lem:BTRegion}
The {Gaussian achievable scheme} achieves $\mathcal{R}_G(\mb{D},\mu)$ and
\[
\mathcal{R}_G(\mb{D},\mu) = \mathcal{R}^{*}(\mb{D},\mu).
\]
\end{Lem}
\begin{proof}
It follows immediately that the {Gaussian achievable scheme} achieves $\mathcal{R}_G(\mb{D},\mu)$. The equality in Lemma \ref{lem:BTRegion} is proved in Appendix A.
\end{proof}
Lemma \ref{lem:BTRegion} implies that $$\mathcal{R}(\mb{D},\mu) \le \mathcal{R}^{*}\left (\mb{D},\mu \right).$$ We prove the reverse inequality (converse) next. Since the proof is rather long, we divide it into sections. The
next section gives a nonrigorous overview of the argument.
In the following section, we study the optimization problem $\left (P_{G1} \right)$ in the definition of $\mathcal{R}^{*}(\mb{D},\mu)$ and establish several properties that its optimal solution satisfies. We use these properties in Section \ref{sec:ConvIngre} to prove the main result needed for the converse. We finally complete the proof of Theorem \ref{thm:MainThm} in Section \ref{sec:ConvProof}.

\section{Overview of the Converse Argument}\label{sec:ConverseOverview}

The starting point of our proof is Oohama's converse
for the scalar case, which proceeds as follows. Let $f_1^{(n)}$ and
$f_2^{(n)}$ be encoding functions and $g^{(n)}$ be a decoding
function that achieve the rate-distortion vector $(R_1,R_2,D)$. Let $C_1 \triangleq f_1^{(n)}(X^n)$ and $C_2 \triangleq f_2^{(n)}(Y^n).$
By standard steps, we have
\begin{align*}
nR_2 & \ge \log M_2^{(n)} \\
     & \ge H(C_2) \\
     & = I(Y^n;C_2).
\end{align*}
Likewise, we have
\begin{align*}
nR_1 & \ge \log M_1^{(n)} \\
     & \ge H(C_1) \\
     & \ge H(C_1|C_2) \\
     & = I(X^n;C_1|C_2) \\
     & = I(X^n;C_1,C_2) -
                 I(X^n;C_2).
\end{align*}
It follows that
\begin{align}
\nonumber
  nR_1  \ge \hspace{.25in} \inf_{C_1,C_2} & \hspace{.125in} I(X^n;C_1,C_2) -
                 I(X^n;C_2) \\
                 \label{eq:mainopt}
     \text{subject to} & \hspace{.125in}
         \sum_{i = 1}^n E \left[(X^n(i) - E[X^n(i)|C_1,C_2])^2 \right] \le nD, \\
\nonumber
     & \hspace{.125in}
         I(Y^n;C_2) \le n R_2, \hspace{0.1in} \textrm{and}  \\
\nonumber
          & \hspace{.125in}
         X^n \leftrightarrow Y^n \leftrightarrow C_2.
\end{align}
Now this infimum can be lower bounded by separately optimizing each term
\begin{align}
\label{eq:sep}
\nonumber
nR_1  \ge \hspace{.25in} \inf_{C_1,C_2}   & \hspace{.125in}         I(X^n;C_1,C_2) \hspace{2in} - \hspace{0.5in}
            \sup_{C_2} \hspace{.125in} I(X^n;C_2) \\
            \text{subject to} & \hspace{.125in}
         \sum_{i = 1}^n E \left[(X^n(i) - E[X^n(i)|C_1,C_2])^2 \right] \le nD \hspace{0.62in}\text{subject to} \hspace{.125in}
         I(Y^n;C_2) \le n R_2 \hspace{0.1in} \textrm{and}  \\
         \nonumber
         & \hspace{4in}
         X^n \leftrightarrow Y^n \leftrightarrow C_2.
\end{align}
%\begin{align}
%\label{eq:sep}
%nR_1 & \ge \inf_{(C_1,C_2):\sum_{i = 1}^n E[(X^n(i) - E[X^n(i)|C_1,C_2])^2] \le nD}
%              I(X^n;C_1,C_2) -
%            \sup_{C_2:X^n \leftrightarrow Y^n \leftrightarrow C_2,
%                    I(Y^n;C_2) \le n R_2}  I(X^n;C_2).
%\end{align}
The first optimization problem,
\begin{align*}
%\label{eq:opt1}
\nonumber
\inf_{C_1,C_2} & \hspace{.125in} I(X^n;C_1,C_2) \\
\text{subject to} & \hspace{.125in}
\sum_{i = 1}^n E \left[(X^n(i) - E[X^n(i)|C_1,C_2])^2 \right] \le nD,
\end{align*}
which we call the \emph{distortion problem},
can be solved using the entropy-maximizing property of the
Gaussian distribution and the concavity of the logarithm.
The second problem,
\begin{align}
\label{eq:opt2}
\nonumber
\sup_{C_2} & \hspace{.125in} I(X^n;C_2) \\
\text{subject to} & \hspace{.125in}  I(Y^n;C_2) \le n R_2 \hspace{0.1in} \textrm{and}\\
\nonumber
& \hspace{.125in} X^n \leftrightarrow Y^n \leftrightarrow C_2,
\end{align}
which we call the \emph{helper problem},
can be solved via the conditional version of the
entropy power inequality~\cite{Bergmans}.
Substituting these solutions into~(\ref{eq:sep}) yields
exactly the $R_1$ achieved by the scheme from the
previous section for the given $R_2$ and $D$.
This completes Oohama's converse proof for
the scalar case.

The key to Oohama's proof is that separately minimizing
the two terms in~(\ref{eq:mainopt}) does not decrease the
objective. More precisely, for any pair $(C_1^*,
C_2^*)$ that achieves the infimum in~(\ref{eq:mainopt})
we have
\begin{align}
\label{eq:distortionloss}
\nonumber
I(X^n;C_1^*,C_2^*) = \hspace{.25in} \inf_{C_1,C_2} & \hspace{.125in} I(X^n;C_1,C_2) \\
\text{subject to} & \hspace{.125in}
\sum_{i = 1}^n E \left[(X^n(i) - E[X^n(i)|C_1,C_2])^2 \right] \le nD,
\end{align}
and
\begin{align}
\label{eq:helperloss}
\nonumber
I(X^n;C_2^*) = \hspace{.35in} \sup_{C_2} & \hspace{.125in} I(X^n;C_2) \\
\text{subject to} & \hspace{.125in}  I(Y^n;C_2) \le n R_2 \hspace{0.1in} \textrm{and}\\
\nonumber
& \hspace{.125in} X^n \leftrightarrow Y^n \leftrightarrow C_2,
\end{align}
Whenever (\ref{eq:distortionloss}) occurs, we shall say that the
\emph{distortion problem incurs no loss.}
Whenever (\ref{eq:helperloss}) occurs, we shall say that the
\emph{helper problem incurs no loss.}

It is not difficult to verify that this proof also works when
$X$ is a scalar and $\mb{Y}$ is a vector. In particular,
both the distortion and helper problems incur no loss in
this case.
%Of course,
%this vector-help-scalar case can also be solved by arguing
% that the projection of $Y$ in the direction of $X$ is a sufficient
% statistic for the helper, thereby reducing the problem back to
% the scalar-help-scalar case.
When both $\mb{X}$ and $\mb{Y}$ are vectors, the proof
breaks down in three places:
\begin{enumerate}
\item The distortion problem incurs a loss in general.
For instance, if $\mb{D} \preccurlyeq \mb{K_X}$, then the distortion problem is solved by choosing $C_1$ and
$C_2$ so that
$$
\sum_{i = 1}^n E\left[\Bigr(\mb{X}^n(i) - E[\mb{X}^n(i)|C_1,C_2]\Bigr)
         \Bigr(\mb{X}^n(i) - E[\mb{X}^n(i)|C_1,C_2]\Bigr)^T\right] = n \mb{D}.
$$
That is, the constraint is
met with equality. For the original problem in (\ref{eq:mainopt}), on the other
hand, even if $\mb{D} \preccurlyeq \mb{K_X}$ we can only guarantee that
$$
\sum_{i = 1}^n E\left[\Bigr(\mb{X}^n(i) - E[\mb{X}^n(i)|C_1^*,C_2^*]\Bigr)
         \Bigr(\mb{X}^n(i) - E[\mb{X}^n(i)|C_1^*,C_2^*]\Bigr)^T\right] \preccurlyeq n \mb{D},
$$
and equality does not hold in general.
The lack of equality
is easiest to see when $\mb{K_Y}$ is poorly conditioned. If
$\mb{K_Y}$ has essentially one nonzero eigenvalue, then the helper
will allocate all of its rate in the direction of the associated
eigenvector. If $R_2$ is large, this could result in ``overshooting''
the distortion constraint in that direction.

%This is because the helper cannot choose which
%directions of $\mb{X}$ to favor with $C_2^*$;
%the choice is dictated by the covariance
%between $\mb{X}$ and $\mb{Y}$.
\item The helper problem also incurs a loss in general.
One way of seeing this is to note that if the
goal is only to maximize the mutual information in
(\ref{eq:opt2}), then one might choose $C_2$ to favor
a direction along which the distortion constraint
$\mb{D}$ is not active over one for which it is.
This would necessarily deviate from the optimizer
$C_2^*$ of the original problem.
\item The vector EPI does not solve the helper problem
       in general.
\end{enumerate}
To address the first issue, observe that the distortion
problem incurs no loss if the optimizers $C_1^*$ and
$C_2^*$ for the original problem happen to meet the
distortion constraint with equality, i.e., it holds
that
$$
\sum_{i = 1}^n E\left[\Bigr(\mb{X}^n(i) - E[\mb{X}^n(i)|C_1^*,C_2^*]\Bigr)
         \Bigr(\mb{X}^n(i) - E[\mb{X}^n(i)|C_1^*,C_2^*]\Bigr)^T\right] = n \mb{D}.
         $$
In prior work~\cite{Rahman1}, we showed that it is possible
to reduce the general case to this one by projecting
the source and the distortion constraint in the
directions in which the distortion constraint is
met with equality for the candidate optimal
scheme. We call this process \emph{distortion
projection.} This addresses the first issue.
One can verify that if $\mb{X}$ is
a vector and $Y$ is a scalar, then the second
and third issues do not arise, and hence \emph{distortion projection} together with Oohama's converse arguments is sufficient to solve the
problem~\cite{Rahman1}.

Liu and Viswanath~\cite{Liu} showed that the
\emph{channel enhancement} technique of Weingarten
\emph{et al.}~\cite{Wein} is sufficient to solve the helper
problem in the vector case, thereby addressing the
third issue. Their solution, however, is not sufficient to handle the second issue. Recently, Zhang~\cite{Guo} introduced a
variation on the enhancement idea called \emph{source
enhancement} that subsumes Liu and Viswanath's approach.
\emph{Source enhancement} effectively replaces
the original problem with a relaxation for which the
helper problem incurs no loss and the vector EPI solves the helper problem,
although Zhang does not
describe it in this way. This addresses the second and third
issues. Thus it appears that \emph{distortion projection}, \emph{source enhancement}, and Oohama's converse technique together should be sufficient
to solve the case in which
both $\mb{X}$ and $\mb{Y}$ are vectors.
We shall show that this is indeed true. \emph{Source enhancement} and Oohama's converse technique are
lifted directly from~\cite{Oohama, Guo}. The \emph{distortion projection}, on the other hand, requires an extension
beyond what was needed in the scalar helper case~\cite{Rahman1}.
This extension requires us to first establish several properties of
the optimal Gaussian solution to the problem, to which we turn next.

\section{Properties of the Optimal Gaussian Solution}\label{sec:PropOptGaussSol}
In this section, we study the optimization problem $\left (P_{G1} \right)$ defined in Section \ref{subsec:rateregion}. Note first that the constraints
\begin{align*}
\mb{K}_{\mb{Y}|\mb{V}} &\succcurlyeq \mb{0} \hspace{0.1in} \textrm{and}\\
\mb{K}_{\mb{X}|\mb{U,V}} &\succcurlyeq \mb{0}
\end{align*}
are never active because otherwise the objective value is infinite. We therefore ignore these constraints in the study of the problem. Now, instead of studying $\left (P_{G1} \right)$ directly as it is, we study an equivalent formulation. This formulation is implicit in \cite{Guo}. Note that if $\mb{K}_{\mb{Y}|\mb{V}}$ and $\mb{K}_{\mb{X}|\mb{U,V}}$ are feasible for $\left (P_{G1} \right)$, then there exist two positive semidefinite matrices $\mb{B}_1$ and $\mb{B}_2$ such that
\begin{align*}
\mb{K}_{\mb{Y}|\mb{V}} &= \mb{K_Y}-\mb{B}_2,\\
\mb{K}_{\mb{X}|\mb{U,V}} &= \mb{K}_{\mb{Y}|\mb{V}}+\mb{K_N} - \mb{B}_1 \\
&= \mb{K_Y}-\mb{B}_2+\mb{K_N} - \mb{B}_1\\
&=\mb{K_X} - \mb{B}_1 - \mb{B}_2, \hspace{0.1in} \textrm{and}\\
\mb{K_X} - \mb{B}_1 - \mb{B}_2 &\preccurlyeq \mb{D}.
\end{align*}
Therefore, $\left (P_{G1} \right)$ is equivalent to the following problem
\begin{align*}
\left (P_{G2} \right)\hspace {0.2in} \triangleq \hspace {0.2in}\min_{\mb{B}_1,\mb{B}_2} \hspace {0.2in} &\frac{\mu}{2} \log \frac{ \left |\mb{K_X}- \mb{B}_2 \right |}{ \left|\mb{K_X}- \mb{B_1}- \mb{B}_2\right|} + \frac{1}{2} \log \frac{ \left |\mb{K_Y} \right |}{ \left|\mb{K_Y}-\mb{B}_2\right|}\\
\textrm{subject to} \hspace{0.1in} &\mb{B}_i \succcurlyeq \mb{0} \hspace {0.15 cm} \textrm{for all} \hspace {0.15 cm} i \in \{1,2\}, \hspace {0.15 cm} \textrm{and}\\
&\mb{D} \succcurlyeq \mb{K_X-\mb{B}_1-\mb{B}_2}.
\end{align*}
We next establish several properties that the optimal solution to $\left (P_{G2} \right)$ satisfies.

Since $\left (P_{G2} \right)$ has continuous objective and a compact feasible set, there exists an optimal solution $(\mb{B}^{*}_1,\mb{B}^{*}_2)$ to it. The Lagrangian of the problem is~\cite[Sec.~5.9.1]{Boyd}
\begin{align*}
%\min_{\mb{B}_1,\mb{B}_2} \hspace {0.2in}
&\frac{\mu}{2} \log \frac{ \left |\mb{K_X}- \mb{B}_2 \right |}{ \left|\mb{K_X}- \mb{B_1}- \mb{B}_2\right|} + \frac{1}{2} \log \frac{ \left |\mb{K_Y} \right |}{ \left|\mb{K_Y}-\mb{B}_2\right|}- \textrm{Tr} \Bigr( \mb{B}_1 \mb{M}_1 + \mb{B}_2 \mb{M}_2 - (\mb{K_X}-\mb{B}_1-\mb{B}_2 - \mb{D}) \mb{\Lambda}\Bigr),
\end{align*}
where $\mb{M}_1,\mb{M}_2,$ and $\mb{\Lambda}$ are positive semidefinite Lagrange multiplier matrices corresponding to the constraints $\mb{B}_1 \succcurlyeq \mb{0}, \mb{B}_2 \succcurlyeq \mb{0}$, and $\mb{D} \succcurlyeq \mb{K_X-\mb{B}_1-\mb{B}_2}$, respectively.
%Then the optimal solution $(\mb{B}^{*}_1,\mb{B}^{*}_2)$ must satisfy the following necessary KKT conditions \cite[pp. 241-248]{Boyd}:
The KKT conditions for this problem are~\cite[Sec.~5.9.2]{Boyd}
\begin{align}
\frac{\mu}{2}(\mb{K_X}-\mb{B}^{*}_1-\mb{B}^{*}_2)^{-1} - \mb{\Lambda}^{*} - \mb{M}_1^{*} &= \mb{0}, \label{eq:KKT1}\\
\frac{\mu}{2}(\mb{K_X}-\mb{B}^{*}_1-\mb{B}^{*}_2)^{-1} -\frac{\mu}{2}(\mb{K_X}-\mb{B}^{*}_2)^{-1}+\frac{1}{2}(\mb{K_Y}-\mb{B}^{*}_2)^{-1}- \mb{\Lambda}^{*} - \mb{M}_2^{*} &= \mb{0}, \label{eq:KKT2}\\
\mb{B}^{*}_i\mb{M}_i^{*} &= \mb{0}, \hspace{0.1in} \textrm{for all} \hspace{0.1in} i \in \{1,2\} \label{eq:KKT3}\\
(\mb{K_X}-\mb{B}^{*}_1-\mb{B}^{*}_2 - \mb{D}) \mb{\Lambda}^{*}&= \mb{0}, \hspace {0.1in} \textrm{and} \label{eq:KKT4}\\
\mb{M}_1^{*}, \mb{M}_2^{*},\mb{\Lambda}^{*}  &\succcurlyeq \mb{0}, \label{eq:KKT5}
\end{align}
where $\mb{M}_1^{*}, \mb{M}_2^{*},$ and $\mb{\Lambda}^{*}$ are optimal Lagrange multiplier matrices. Conditions (\ref{eq:KKT1}) and (\ref{eq:KKT2}) respectively are obtained by setting gradients of the objective with respect to $\mb{B}_1$ and $\mb{B}_2$ to zero. Conditions (\ref{eq:KKT3}) through (\ref{eq:KKT4}) are slackness conditions on the Lagrange multiplier matrices. We next establish that
these KKT conditions must hold at $(\mb{B}_1^*,\mb{B}_2^*)$.
%Since the problem is non-convex, it is not clear that there exist Lagrange multiplier matrices that satisfy the KKT conditions. We have the following lemma.
\begin{Lem} \label{lem:ExistenceKKT}
There exist matrices $\mb{M}_1^{*}, \mb{M}_2^{*},$ and $\mb{\Lambda}^{*} $ that satisfy the KKT conditions (\ref{eq:KKT1}) -- (\ref{eq:KKT5}).
\end{Lem}
\begin{proof}
See Appendix B.
\end{proof}

Let us define
\[
\mb{\Delta}^{*} \triangleq \mb{\Lambda}^{*} - \frac{\mu}{2} \left [(\mb{K_X}-\mb{B}^{*}_1-\mb{B}^{*}_2)^{-1} - (\mb{K_X}-\mb{B}^{*}_2)^{-1}\right].
\]
It follows from conditions (\ref{eq:KKT1}) and (\ref{eq:KKT2}) that
\begin{align}
\mb{\Delta}^{*} = \frac{\mu}{2}(\mb{K_X}-\mb{B}^{*}_2)^{-1}- \mb{M}_1^{*}= \frac{1}{2}(\mb{K_Y}-\mb{B}^{*}_2)^{-1} - \mb{M}_2^{*}. \label{eq:Delta}
\end{align}
We have the following lemma.
\begin{Lem} \label{lem:PosSemiDefDelta}
$\mb{\Delta}^{*}$ is a nonzero positive semidefinite matrix.
\end{Lem}
\begin{proof}
See Appendix C.
\end{proof}

If $\mb{\Delta}^{*}$ happens to be positive definite, then \emph{distortion projection} turns out to be unnecessary.
%, and Theorem \ref{thm:MainThm} follows from the KKT conditions and the work of Zhang~\cite{Guo}.
%\begin{Prop} \label{prop:Zhang}
%If $\mb{\Delta}^{*}$ is positive definite, then the Gaussian achievable scheme
%is optimal, i.e., the conclusion of Theorem \ref{thm:MainThm} holds.
%\end{Prop}
%\begin{proof}
%If $\mb{\Delta}^{*}$ is positive definite, then so is $\mb{\Lambda}^{*}$, which implies that the distortion constraint is met with equality, i.e.,
%\[
%\mb{K_X}-\mb{B}^{*}_1-\mb{B}^{*}_2 = \mb{D}.
%\]
%Hence, the KKT conditions (\ref{eq:KKT1}) through (\ref{eq:KKT5}) in this case are equivalent to the following conditions:
%\begin{align*}
%\frac{\mu}{2}(\mb{K_X}-\mb{B}^{*}_2)^{-1}- \mb{M}_1^{*}&= \frac{1}{2}(\mb{K_Y}-\mb{B}^{*}_2)^{-1} - \mb{M}_2^{*} \succ \mb{0},\\
%\mb{B}^{*}_2\mb{M}_2^{*} &= \mb{0}, \\
%[\mb{B}^{*}_2 - ( \mb{K_X}-\mb{D}) ] \mb{M}_1^{*}&= \mb{0},\hspace {0.1in} \textrm{and}\\
%\mb{M}_1^{*}, \mb{M}_2^{*}  &\succcurlyeq \mb{0},
%\end{align*}
%which, as shown in \cite[Theorem 3.1]{Guo}, are sufficient for the optimality of the Gaussian achievable scheme.
%\end{proof}
To handle the case in which $\mb{\Delta}^{*}$ is singular, we shall use \emph{distortion projection}.
%We shall show that the Gaussian achievable scheme is optimal even if $\mb{\Delta}^{*}$ is singular.
Since $\mb{\Delta}^{*}, \mb{M}_1^{*},$ and $\mb{M}_2^{*}$ are positive semidefinite, we can write their spectral decompositions as
\begin{align}
\mb{\Delta}^{*} &= \sum_{i=1}^{r} \lambda_i \mb{s}_i \mb{s}_i^{T}, \label{eq:EigDecom1}\\
\mb{M}_1^{*} &= \sum_{i=1}^{p} \alpha_i {\mb{a}}_i {\mb{a}}_i^{T}, \hspace {0.1in} \textrm{and}  \label{eq:EigDecom2}\\
\mb{M}_2^{*} &= \sum_{i=1}^{q} \beta_i \mb{b}_i \mb{b}_i^{T},  \label{eq:EigDecom3}
\end{align}
where
\begin{enumerate}
\item[(i)] $0 < r \le m$,
\item[(ii)] $0 \le p, q \le m$,
\item[(iii)] $\lambda_i > 0,$ for all $i \in \{1,\dots,r\}$,
\item[(iv)] $\alpha_i > 0,$ for all $i \in \{1,\dots,p\}$
\item[(v)] $\beta_i > 0,$ for all $i \in \{1,\dots,q\},$ and
\item[(vi)] $\{\mb{s}_i\}_{i=1}^{r}, \{{\mb{a}}_i\}_{i=1}^{p},$ and $\{\mb{b}_i\}_{i=1}^{q}$ are sets of orthonormal vectors.
\end{enumerate}
Note that we allow $p$ and $q$ to be zero because $\mb{M}_1^{*}$ and $\mb{M}_2^{*}$ can be zero. Since (\ref{eq:Delta}) implies
\begin{align*}
\mb{\Delta}^{*} + \mb{M}_1^{*} &= \frac{\mu}{2}(\mb{K_X}-\mb{B}^{*}_2)^{-1} \succ \mb{0} \hspace {0.1in} \textrm{and} \\
\mb{\Delta}^{*} + \mb{M}_2^{*}&= \frac{1}{2}(\mb{K_Y}-\mb{B}^{*}_2)^{-1}\succ \mb{0},
\end{align*}
we must have
\begin{align*}
r+p &\ge m \hspace {0.1in} \textrm{and}\\
r+q &\ge m.
\end{align*}
This means that if $r+p=m$, then $\mb{s}_1,\mb{s}_2,\dots,\mb{s}_r,\mb{a}_1,\mb{a}_2,\dots,\mb{a}_p$ must be linearly independent. Similarly, if $r+q=m$, then $\mb{s}_1,\mb{s}_2,\dots,\mb{s}_r,\mb{b}_1,\mb{b}_2,\dots,\mb{b}_q$ must be linearly independent. %This means that if $p$ (similarly $q$) is zero, then $r$ must equal $m$, i.e., $\mb{\Delta}^{*}$ is strictly positive definite. Moreover, if $\mb{\Delta}^{*}$ is rank deficient, i.e., $r < m$, then both $p$ and $q$ are nonzero.

Define the matrix
\begin{align*}
\mb{S} &\triangleq \left [\sqrt{\lambda_1}\mb{s}_1,\sqrt{\lambda_2}\mb{s}_2,\dots,\sqrt{\lambda_r}\mb{s}_r \right ].
\end{align*}
It now follows from the definition of $\mb{\Delta}^{*}$ that
\begin{align*}
\mb{\Lambda}^{*} &\succcurlyeq \mb{\Delta}^{*}= \mb{SS}^T
\end{align*}
because $$(\mb{K_X}-\mb{B}^{*}_1-\mb{B}^{*}_2)^{-1} \succcurlyeq (\mb{K_X}-\mb{B}^{*}_2)^{-1}.$$ This and (\ref{eq:KKT4}) imply that
\begin{align}
(\mb{K_X}-\mb{B}^{*}_1-\mb{B}^{*}_2 - \mb{D}) \mb{S}= \mb{0}.  \label{eq:ActiveDistConst}
\end{align}

Let $\mb{C}$ be an $m \times m$ positive definite matrix and $\{\mb{C}_1,\mb{C}_2, \dots, \mb{C}_t\}$ be a set of $m \times m$ positive definite matrices.
\begin{Def}
A non-zero $m \times p$ matrix $\mb{E}$ is $\mb{C}$\emph{-orthogonal} if $\mb{E}^T \mb{C} \mb{E}$ is a diagonal matrix.
\end{Def}
\begin{Def}
A non-zero $m \times p$ matrix $\mb{E}$ is $\{\mb{C}_1,\mb{C}_2, \dots, \mb{C}_t\}$\emph{-orthogonal} if it is $\mb{C}_i$\emph{-orthogonal} for all $i \in \{1,2,\dots,t\}$.
\end{Def}
\begin{Def}
A non-zero $m \times p$ matrix $\mb{E}$ and a non-zero $m \times q$ matrix $\mb{F}$ are \emph{cross} $\mb{C}$\emph{-orthogonal} if $\mb{E}^T \mb{C} \mb{F} = \mb{0}.$
\end{Def}
\begin{Def}
A non-zero $m \times p$ matrix $\mb{E}$ and a non-zero $m \times q$ matrix $\mb{F}$ are \emph{cross} $\{\mb{C}_1,\mb{C}_2, \dots, \mb{C}_t\}$\emph{-orthogonal} if they are \emph{cross} $\mb{C}_i$\emph{-orthogonal} for all $i \in \{1,2,\dots,t\}$.
\end{Def}
\begin{Def}
A non-zero vector $\mb{w}$ is in \emph{span}$\{\mb{c}_i\}_{i=1}^l$ if there exist real numbers $\{\gamma_i\}_{i=1}^l$ such that
\begin{align*}
\mb{w} = \sum_{i=1}^l \gamma_i \mb{c}_i.
\end{align*}
We denote this as
\[
\mb{w} \in \textrm{\emph{span}}\{\mb{c}_i\}_{i=1}^l.
\]
\end{Def}
We have the following theorem about the optimal solution to the optimization problem $\left (P_{G2} \right)$.
\begin{Thm}  \label{thm:PropOptGaussSol}
There exist two matrices
\begin{align*}
\mb{T} \triangleq  [\mb{t}_1,\mb{t}_2,\dots,\mb{t}_{m-r}]
\end{align*}
and
\[
\mb{W} \triangleq  [\mb{w}_1,\mb{w}_2,\dots,\mb{w}_{m-r}]
\]
such that $[\mb{S},\mb{T}]$ and $[\mb{S},\mb{W}]$ are invertible and if $r < m$ then
\begin{enumerate}
\item[(a)] ${\mb{t}}_{1}, {\mb{t}}_{2}, \dots, {\mb{t}}_{m-r} \in \textrm{\emph{span}}\{{\mb{a}}_i\}_{i=1}^p$,
\item[(b)] ${\mb{T}}$ is $\bigr\{(\mb{K_X}-\mb{B}_2^{*}), (\mb{K_X}-\mb{B}_1^{*}-\mb{B}_2^{*})\bigr\}$\emph{-orthogonal} with
\[
{\mb{T}}^T (\mb{K_X}-\mb{B}_2^{*}) {\mb{T}} = {\mb{T}}^T (\mb{K_X}-\mb{B}_1^{*}-\mb{B}_2^{*}) {\mb{T}},
\]
\item[(c)] $\mb{S}$ and ${\mb{T}}$ are \emph{cross} $\bigr\{\mb{D}, (\mb{K_X}-\mb{B}_2^{*}), (\mb{K_X}-\mb{B}_1^{*}-\mb{B}_2^{*})\bigr\}$\emph{-orthogonal},
\item[(d)] ${\mb{w}}_{1}, {\mb{w}}_{2}, \dots, {\mb{w}}_{m-r} \in \textrm{\emph{span}}\{\mb{b}_i\}_{i=1}^q$,
\item[(e)] ${\mb{W}}$ is $\bigr\{\mb{K_Y}, (\mb{K_Y}-\mb{B}_2^{*})\bigr\}$\emph{-orthogonal} with
\[
{\mb{W}}^T \mb{K_Y} {\mb{W}} = {\mb{W}}^T (\mb{K_Y}-\mb{B}_2^{*}) {\mb{W}},  \hspace{0.1in} \textrm{and}
\]
\item[(f)] $\mb{S}$ and ${\mb{W}}$ are \emph{cross} $\bigr\{\mb{K_Y}, (\mb{K_Y}-\mb{B}_2^{*})\bigr\}$\emph{-orthogonal}.
\end{enumerate}
\end{Thm}
\begin{proof}
It suffices to consider $r < m$ case. Since $\mb{\Delta}^{*}=\mb{S}\mb{S}^T$ is rank deficient in this case, there exists $\mb{z}_1 \neq \mb{0}$ such that $$\mb{S}^T \mb{z}_1=\mb{0}.$$ Let us define
\[
{\mb{t}}_1 \triangleq (\mb{K_X}-\mb{B}_2^{*})^{-1}\mb{z}_1.
\]
Therefore
\[
\mb{S}^T (\mb{K_X}-\mb{B}_2^{*}) {\mb{t}}_1=\mb{0}.
\]
We have from (\ref{eq:Delta}), (\ref{eq:EigDecom1}), and (\ref{eq:EigDecom2}) that
\begin{align*}
\frac{\mu}{2} (\mb{K_X}-\mb{B}_2^{*})^{-1}=\mb{\Delta}^{*}+\mb{M}_1^{*}=\mb{S}\mb{S}^T+\sum_{i=1}^p \alpha_i {\mb{a}}_i{\mb{a}}_i^T.
\end{align*}
On post-multiplying this by $(\mb{K_X}-\mb{B}_2^{*}){\mb{t}}_1$, we obtain
\begin{align*}
\frac{\mu}{2}{\mb{t}}_1&=\mb{S}\mb{S}^T(\mb{K_X}-\mb{B}_2^{*}){\mb{t}}_1+\sum_{i=1}^p \alpha_i {\mb{a}}_i{\mb{a}}_i^T(\mb{K_X}-\mb{B}_2^{*}){\mb{t}}_1\\
&=\sum_{i=1}^p \alpha_i {\mb{a}}_i \bigr ({\mb{a}}_i^T(\mb{K_X}-\mb{B}_2^{*}){\mb{t}}_1 \bigr).
\end{align*}
This proves that $${\mb{t}}_1 \in \textrm{span}\{{\mb{a}}_i\}_{i=1}^p.$$ We next show that $${\mb{t}}_1 \notin \textrm{span}\{\mb{s}_i\}_{i=1}^r.$$ Suppose otherwise that $${\mb{t}}_1 \in \textrm{span}\{\mb{s}_i\}_{i=1}^r.$$ Then there exist real numbers $\{c_i\}_{i=1}^r$ such that
\[
{\mb{t}}_1 = \sum_{i=1}^r c_i \mb{s}_i.
\]
Since $\mb{S}^T (\mb{K_X}-\mb{B}_2^{*}) {\mb{t}}_1  = \mb{0}$, we have
\[
\mb{s}_i^T (\mb{K_X}-\mb{B}_2^{*}) {\mb{t}}_1 = 0 \hspace{0.1 in} \textrm{for all}\hspace{0.1 in} i \in \{1,2,\dots,r\}.
\]
On multiplying this by $c_i$ and then summing over all $i$ in $\{1,2,\dots,r\}$, we obtain
\[
{\mb{t}}_1^T (\mb{K_X}-\mb{B}_2^{*}) {\mb{t}}_1 = 0,
\]
which is a contradiction because $\mb{K_X}-\mb{B}_2^{*}$ is positive definite. We therefore have that $${\mb{t}}_1 \notin \textrm{span}\{\mb{s}_i\}_{i=1}^r.$$ We have shown so far that there exists ${\mb{t}}_1 \in \textrm{span}\{\mb{a}_i\}_{i=1}^p$ such that the rank of $[\mb{S}, {\mb{t}}_1]$ is $r+1$ and $$\mb{S}^T (\mb{K_X}-\mb{B}_2^{*}) {\mb{t}}_1 = \mb{0}.$$ Let us now assume that there exists
\[
{\mb{T}}_j \triangleq [{\mb{t}}_{1}, {\mb{t}}_{2}, \dots, {\mb{t}}_{j}],
\]
where $${\mb{t}}_{1}, {\mb{t}}_{2}, \dots, {\mb{t}}_{j} \in \textrm{span}\{{\mb{a}}_i\}_{i=1}^p$$ and $1 \le j < m-r$ such that the rank of $[\mb{S}, {\mb{T}}_j]$ is $r+j$, $$\mb{S}^T (\mb{K_X}-\mb{B}_2^{*}) {\mb{T}}_j = \mb{0},$$ and $${\mb{t}}_{k}^T (\mb{K_X}-\mb{B}_2^{*}) {\mb{t}}_l = {0}$$ for all $k \neq l$ in $\{1,2,\dots,j\}.$ Then there exists $\mb{z}_{j+1} \neq \mb{0}$ such that $$[\mb{S}, {\mb{T}}_j]^T \mb{z}_{j+1}=\mb{0}.$$ Let us define
\[
{\mb{t}}_{j+1} \triangleq (\mb{K_X}-\mb{B}_2^{*})^{-1}\mb{z}_{j+1}.
\]
We therefore have that
\[
[\mb{S}, {\mb{T}}_j]^T (\mb{K_X}-\mb{B}_2^{*}) {\mb{t}}_{j+1} = \mb{0}.
\]
It can be shown as before that
\begin{align*}
{\mb{t}}_{j+1} &\in \textrm{span}\{{\mb{a}}_i\}_{i=1}^p
\end{align*}
and
\begin{align*}
{\mb{t}}_{j+1} &\notin \textrm{span} \left\{\{\mb{s}_i\}_{i=1}^r,\{{\mb{t}}_k\}_{k=1}^j \right\}.
\end{align*}
Hence, the rank of $[\mb{S}, {\mb{T}}_{j+1}],$ where
\[
{\mb{T}}_{j+1} \triangleq [{\mb{T}}_j, {\mb{t}}_{j+1}],
\]
is $r+j+1$, $$\mb{S}^T (\mb{K_X}-\mb{B}_2^{*}) {\mb{T}}_{j+1} = \mb{0},$$ and $${\mb{t}}_{k}^T (\mb{K_X}-\mb{B}_2^{*}) {\mb{t}}_l = {0},$$ for all $k \neq l$ in $\{1,2,\dots,j+1\}.$ It now follows from the mathematical induction that there exist $${\mb{t}}_1,{\mb{t}}_2,\dots,{\mb{t}}_{m-r} \in \textrm{span}\{{\mb{a}}_i\}_{i=1}^p$$ such that if we define
\[
{\mb{T}} \triangleq [{\mb{t}}_1,{\mb{t}}_2,\dots,{\mb{t}}_{m-r}],
\]
then $[\mb{S}, {\mb{T}}]$ is invertible,
\begin{align*}
\mb{S}^T (\mb{K_X}-\mb{B}_2^{*}) {\mb{T}} &=  \mb{0}, \hspace{0.1in} \textrm{and}\\
{\mb{T}}^T(\mb{K_X}-\mb{B}_2^{*}){\mb{T}} &= \mb{G},
\end{align*}
where $$\mb{G}\triangleq \textrm{Diag}  \Bigr\{\bigr({\mb{t}}_1^T(\mb{K_X}-\mb{B}_2^{*}){\mb{t}}_1\bigr), \bigr({\mb{t}}_2^T(\mb{K_X}-\mb{B}_2^{*}){\mb{t}}_2 \bigr), \dots, \bigr({\mb{t}}_{m-r}^T(\mb{K_X}-\mb{B}_2^{*}){\mb{t}}_{m-r}\bigr)\Bigr\}.$$ Since $\mb{B}_1^{*}{\mb{T}}=\mb{0}$ from (\ref{eq:KKT3}) and $(\mb{K_X}-\mb{B}_1^{*}-\mb{B}_2^{*}) \mb{S}= \mb{DS}$ from (\ref{eq:ActiveDistConst}), we immediately have that
\begin{align*}
\mb{S}^T (\mb{K_X}-\mb{B}_2^{*}) {\mb{T}} &= \mb{S}^T (\mb{K_X}-\mb{B}_1^{*}-\mb{B}_2^{*}) {\mb{T}} = \mb{S}^T \mb{D} {\mb{T}}=\mb{0}, \hspace{0.1in} \textrm{and}\\
{\mb{T}}^T(\mb{K_X}-\mb{B}_2^{*}){\mb{T}} &= {\mb{T}}^T(\mb{K_X}-\mb{B}_1^{*}- \mb{B}_2^{*}){\mb{T}}=\mb{G}.
\end{align*}
This completes the proof of parts (a) through (c) of the theorem.

For parts (d) through (f), we have from (\ref{eq:Delta}), (\ref{eq:EigDecom1}), and (\ref{eq:EigDecom3}) that
\begin{align*}
\frac{1}{2} (\mb{K_Y}-\mb{B}_2^{*})^{-1}=\mb{\Delta}^{*}+\mb{M}_2^{*}=\mb{S}\mb{S}^T+\sum_{i=1}^q \beta_i \mb{b}_i\mb{b}_i^T.
\end{align*}
Similar to the previous case, we can find $${\mb{w}}_1,{\mb{w}}_2,\dots,{\mb{w}}_{m-r} \in \textrm{span}\{\mb{b}_i\}_{i=1}^q$$ such that if we define
\[
{\mb{W}} \triangleq [{\mb{w}}_1,{\mb{w}}_2,\dots,{\mb{w}}_{m-r}],
\]
then $[\mb{S}, {\mb{W}}]$ is invertible,
\begin{align*}
\mb{S}^T (\mb{K_Y}-\mb{B}_2^{*}) {\mb{W}} &=  \mb{0}, \hspace{0.1in} \textrm{and}\\
{\mb{W}}^T(\mb{K_Y}-\mb{B}_2^{*}){\mb{W}} &= \mb{H},
\end{align*}
where $$\mb{H}\triangleq \textrm{Diag}  \Bigr\{\bigr({\mb{w}}_1^T(\mb{K_Y}-\mb{B}_2^{*}){\mb{w}}_1\bigr), \bigr({\mb{w}}_2^T(\mb{K_Y}-\mb{B}_2^{*}){\mb{w}}_2 \bigr), \dots, \bigr({\mb{w}}_{m-r}^T(\mb{K_Y}-\mb{B}_2^{*}){\mb{w}}_{m-r}\bigr)\Bigr\}.$$ Since $\mb{B}_2^{*}{\mb{W}}=\mb{0}$ from (\ref{eq:KKT3}), we conclude
\begin{align*}
\mb{S}^T \mb{K_Y} {\mb{W}} &= \mb{S}^T (\mb{K_Y}-\mb{B}_2^{*}) {\mb{W}}=\mb{0}, \hspace{0.1in} \textrm{and}\\
{\mb{W}}^T\mb{K_Y}{\mb{W}} &= {\mb{W}}^T(\mb{K_Y}- \mb{B}_2^{*}){\mb{W}}=\mb{H}.
\end{align*}
This completes the proof of parts (d) through (f) of the theorem.
\end{proof}

We have the following corollary of Theorem \ref{thm:PropOptGaussSol}.
\begin{Cor} \label{cor:CorrThm2}
If $r<m= r+p$, then we can set $$\mb{t}_i = \sqrt{\alpha_i}\mb{a}_i$$ for all $i$ in $\{1,2,\dots,p\}$. Similarly, if $r<m=r+q$, then we can set $$\mb{w}_i = \sqrt{\beta_i}\mb{b}_i$$ for all $i$ in $\{1,2,\dots,q\}$.
\end{Cor}
\begin{proof}
Let $r<m=r+p$ and let us set $$\mb{t}_i = \sqrt{\alpha_i}\mb{a}_i$$ for all $i$ in $\{1,2,\dots,p\}$ in the definition of $\mb{T}$. We have from (\ref{eq:Delta}), (\ref{eq:EigDecom1}), and (\ref{eq:EigDecom2}) that
\begin{align}
\frac{\mu}{2} (\mb{K_X}-\mb{B}_2^{*})^{-1} = \sum_{i=1}^r \lambda_i \mb{s}_i \mb{s}_i^T + \sum_{i=1}^p \alpha_i {\mb{a}}_i {\mb{a}}_i^T. \label{eq:Corr1}
\end{align}
Now, on post-multiplying (\ref{eq:Corr1}) by $(\mb{K_X}-\mb{B}_2^{*}) \mb{s}_1$, we obtain
\begin{align*}
\frac{\mu}{2} \mb{s}_1 &= \sum_{i=1}^{r} \lambda_i \mb{s}_i \left (\mb{s}_i^{T} (\mb{K_X}-\mb{B}_2^{*}) \mb{s}_1 \right )+\sum_{i=1}^{p} \alpha_i {\mb{a}}_i \left ({\mb{a}}_i^{T} (\mb{K_X}-\mb{B}_2^{*}) \mb{s}_1 \right ),
\end{align*}
which can be re-written as
\begin{align}
\mb{s}_1 \left (\frac{\mu}{2} - \lambda_1 \left (\mb{s}_1^{T}(\mb{K_X}-\mb{B}_2^{*}) \mb{s}_1 \right ) \right ) - \sum_{i=2}^{r} \lambda_i \mb{s}_i \left (\mb{s}_i^{T}(\mb{K_X}-\mb{B}_2^{*}) \mb{s}_1 \right ) &=\sum_{i=1}^{p} \alpha_i {\mb{a}}_i \left ({\mb{a}}_i^{T} (\mb{K_X}-\mb{B}_2^{*}) \mb{s}_1 \right ). \label{eq:Corr2}
\end{align}
Since $[\mb{S,T}]$ is invertible from (\ref{eq:Corr1}), its columns are linearly independent. Hence, the coefficients of all vectors in (\ref{eq:Corr2}) must be zero. Therefore,
\begin{align*}
\lambda_1 \mb{s}_1^{T} (\mb{K_X}-\mb{B}_2^{*}) \mb{s}_1 &= \frac{\mu}{2}  , \\
\mb{s}_i^{T} (\mb{K_X}-\mb{B}_2^{*}) \mb{s}_1 &= 0, \hspace{0.1in}\forall i \in \{2,\dots,r\}, \hspace{0.1in} \textrm{and} \\
{\mb{a}}_i^{T} (\mb{K_X}-\mb{B}_2^{*}) \mb{s}_1 &= 0, \hspace{0.1in}\forall i \in \{1,\dots,p\}.
\end{align*}
Likewise, on post-multiplying (\ref{eq:Corr1}) by $(\mb{K_X}-\mb{B}_2^{*}) \mb{s}_2, \dots, (\mb{K_X}-\mb{B}_2^{*}) \mb{s}_r, (\mb{K_X}-\mb{B}_2^{*}) {\mb{a}}_1, \dots, (\mb{K_X}-\mb{B}_2^{*}) {\mb{a}}_p$ and then equating all coefficients to zero, we obtain similar equations. In summary,
\begin{align*}
\lambda_i \mb{s}_i^{T}(\mb{K_X}-\mb{B}_2^{*}) \mb{s}_i &= \frac{\mu}{2}, \hspace{0.1in}\forall i \in \{1,\dots,r\},\\
\alpha_i {\mb{a}}_i^{T} (\mb{K_X}-\mb{B}_2^{*}) {\mb{a}}_i &=\frac{\mu}{2} , \hspace{0.1in}\forall i \in \{1,\dots,p\},\\
\mb{s}_i^{T} (\mb{K_X}-\mb{B}_2^{*}) \mb{s}_j &= 0, \hspace{0.1in}\forall i,j \in \{1,\dots,r\}, i \neq j, \\
{\mb{a}}_i^{T} (\mb{K_X}-\mb{B}_2^{*}) {\mb{a}}_j &= 0, \hspace{0.1in}\forall i,j \in \{1,\dots,p\}, i \neq j, \hspace{0.1in} \textrm{and}\\
\mb{s}_i^{T} (\mb{K_X}-\mb{B}_2^{*}) {\mb{a}}_j &= 0, \hspace{0.1in}\forall i \in \{1,\dots,r\}, \forall j \in \{1,\dots,p\}.
\end{align*}
Hence,
\begin{align}
[\mb{S},\mb{T}]^T (\mb{K_X}-\mb{B}_2^{*}) [\mb{S},\mb{T}]= \frac{\mu}{2} \mb{I}_m. \label{eq:Corr3}
\end{align}
Parts (a) through (c) of Theorem \ref{thm:PropOptGaussSol} follow immediately from (\ref{eq:KKT3}), (\ref{eq:KKT4}), and (\ref{eq:Corr3}) because $\mb{M}_1^{*} = \mb{TT}^T$ in this case.

The proof for the case when $r < m = r+q$ is exactly similar. It starts with the following from (\ref{eq:Delta}), (\ref{eq:EigDecom1}), and (\ref{eq:EigDecom3})
\begin{align*}
\frac{1}{2} (\mb{K_Y}-\mb{B}_2^{*})^{-1}=\mb{\Delta}^{*}+\mb{M}_2^{*}=\sum_{i=1}^r \lambda_i \mb{s}_i \mb{s}_i^T +\sum_{i=1}^q \beta_i \mb{b}_i\mb{b}_i^T.
\end{align*}
\end{proof}
In summary, the key properties of the optimal Gaussian solution are as follows. If $\mb{\Delta}^{*}$ (and hence $\mb{S}$) is not invertible, then there exist two matrices $\mb{T}$ and $\mb{W}$ such that
%\[
%\mb{T} \triangleq \left\{
%\begin{array}{l l}
%  [\hspace{0.05in}] & \quad \mbox{if $r=m$}\\
%  \mb{A} & \quad \mbox{if $r<m=r+p$}\\
%  \bar{\mb{A}} & \quad \mbox{if $r<m<r+p$}\\ \end{array} \right.
%\]
%and
%\[
%\mb{W} \triangleq \left\{
%\begin{array}{l l}
%  [\hspace{0.05in}] & \quad \mbox{if $r=m$}\\
%  \mb{B} & \quad \mbox{if $r<m=r+q$}\\
%  \bar{\mb{B}} & \quad \mbox{if $r<m<r+q,$}\\ \end{array} \right.
%\]
their columns respectively are in $\textrm{span}\{{\mb{a}}_i\}_{i=1}^p$ and $\textrm{span}\{\mb{b}_i\}_{i=1}^q$, $[\mb{S,T}]$ and $[\mb{S,W}]$ are invertible, $\mb{S}$ and $\mb{T}$ are cross $(\mb{K_X}-\mb{B}_2^{*})${-orthogonal}, and $\mb{S}$ and $\mb{W}$ are cross $(\mb{K_Y}-\mb{B}_2^{*})${-orthogonal}. We shall exploit these properties in the next section to prove the optimality of an optimization problem, which is central to prove our main result.

\section{Converse Ingredients}\label{sec:ConvIngre}
Let us define an optimization problem as
\begin{align*}
\left (P\right) \hspace {0.2in}\triangleq \hspace {0.5in}\min_{\mb{U,V}} \hspace {0.1in} & \mu I(\mb{X};\mb{U}|\mb{V})+I(\mb{Y};\mb{V}) \nonumber\\
\textrm{subject to} \hspace{0.1in} &\mb{K}_{\mb{X}|\mb{U,V}} \preccurlyeq \mb{D} \hspace{0.1in} \textrm{and} \nonumber\\
&\mb{X}\leftrightarrow \mb{Y}\leftrightarrow \mb{V},
\end{align*}
where $\mb{X}, \mb{Y}, \mb{D},$ and $\mu$ are defined as before. We refer to this problem as the \emph{main optimization problem} and denote it by $\left (P\right)$. We have the following theorem.
\begin{Thm} \label{thm:MainOptProb}
A Gaussian $(\mb{U,V})$ is an optimal solution of the main optimization problem $\left (P\right)$.
\end{Thm}
We prove this theorem in the remainder of the section. The proof for $\mu$ in $[0,1]$ is easy. In this case, the objective of $\left (P\right)$ can be lower bounded as
\begin{align}
\mu I(\mb{X};\mb{U}|\mb{V})+I(\mb{Y};\mb{V}) &= \mu I(\mb{X};\mb{U,V})-\mu I(\mb{X};\mb{V})+I(\mb{Y};\mb{V})\nonumber\\
&= \mu I(\mb{X};\mb{U})+\mu I(\mb{X};\mb{V}|\mb{U})+\mu [I(\mb{Y};\mb{V})- I(\mb{X};\mb{V})] + (1-\mu) I(\mb{Y};\mb{V})\nonumber\\
&\ge \mu I(\mb{X};\mb{U}) \label{eq:ConvIngr1}\\
&= \mu h(\mb{X}) - \mu h(\mb{X}|\mb{U}) \nonumber\\
&\ge \frac{\mu}{2} \log \frac{|\mb{K_X}|}{|\mb{K_{X|U}}|},\label{eq:ConvIngr2}
\end{align}
where
\begin{enumerate}
\item[(\ref{eq:ConvIngr1})] follows because of the facts that $$I(\mb{Y};\mb{V}) \ge 0$$ and $$I(\mb{X};\mb{V} | \mb{U}) \ge 0,$$ and we have
\[
I(\mb{Y};\mb{V})- I(\mb{X};\mb{V}) \ge 0
\]
because of the data processing inequality \cite[Theorem 2.8.1]{Cover} and the Markov chain $\mb{X}\leftrightarrow \mb{Y}\leftrightarrow \mb{V},$ and
\item[(\ref{eq:ConvIngr2})] follows because the Gaussian distribution maximizes the differential entropy for a given covariance matrix \cite[Theorem 8.6.5]{Cover}, i.e.,
 \[
 h(\mb{X}|\mb{U}) \le \frac{1}{2} \log \left( \left (2 \pi e \right)^m \left |\mb{K}_{\mb{X}|\mb{U}} \right | \right).
 \]
\end{enumerate}
Inequalities (\ref{eq:ConvIngr1}) and (\ref{eq:ConvIngr2}) become equalities if we choose a Gaussian $(\mb{U,V})$ such that $\mb{V}$ is independent of $(\mb{X},\mb{Y}, \mb{U})$. Because of the distortion constraint in $(P)$,  the conditional covariance of $\mb{X}$ given $(\mb{U,V})$ should satisfy
\begin{align*}
\mb{0} \preccurlyeq \mb{K}_{\mb{X}|\mb{U,V}}=\mb{K}_{\mb{X}|\mb{U}} \preccurlyeq \mb{D}.
\end{align*}
Since conditioning reduces covariance in a positive semidefinite sense, we also have
\begin{align*}
\mb{K}_{\mb{X}|\mb{U}} \preccurlyeq \mb{K_X}.
\end{align*}
Hence, if $\mu$ is in $[0,1]$, then a Gaussian $(\mb{U,V})$ is an optimal solution of the \emph{main optimization problem} $\left (P\right)$ and the optimal value is
\begin{align}
v\left (P\right) &=\hspace{0.25in}\min_{\mb{K}_{\mb{X}|\mb{U}}} \hspace {0.2in} \frac{\mu}{2} \log \frac{ \left |\mb{K_X} \right |}{ \left|\mb{K}_{\mb{X}|\mb{U}}\right|} \nonumber\\
&\hspace{0.3in}\textrm{subject to} \hspace{0.12in} \mb{K}_{\mb{X}} \succcurlyeq \mb{K}_{\mb{X}|\mb{U}}  \succcurlyeq \mb{0}\hspace{0.1in} \textrm{and}\nonumber\\
&\hspace{0.95in}\mb{D} \succcurlyeq \mb{K}_{\mb{X}|\mb{U}} \nonumber\\
&=v\left (P_{pt-pt} \right).\label{eq:OptVal1}
\end{align}
We therefore assume that $\mu > 1$ in the rest of the section.

Let us first restrict the solution space of $\left (P\right)$ to Gaussian distributions. This results in an optimization problem $\left (P_{G1} \right)$, or equivalently $\left (P_{G2} \right)$, defined in Section \ref{sec:PropOptGaussSol}. For convenience, we shall work with the $\left (P_{G2} \right)$ formulation. First note that since restricting the solution space to Gaussian distributions can only increase the optimal value of the \emph{main optimization problem} $\left (P\right)$, we immediately have
\begin{align}
v\left (P_{G1} \right)=v\left (P_{G2} \right) \ge v\left (P\right). \label{eq:ConvIngr4}
\end{align}
So, it suffices to prove the reverse inequality
\begin{align*}
v\left (P_{G2} \right) \le v\left (P\right).
\end{align*}

Let $(\mb{B}^{*}_1,\mb{B}^{*}_2)$ be an optimal solution to $\left (P_{G2} \right)$. %Since $\left (P_{G1} \right)$ and $\left (P_{G2} \right)$ are equivalent, there exists an optimal solution $\left(\mb{K}_{\mb{Y}|\mb{V}^{*}},\mb{K}_{\mb{X}|\mb{U}^{*},\mb{V}^{*}}\right)$ to $\left (P_{G1} \right)$ satisfying
%\begin{align}
%\mb{K}_{\mb{Y}|\mb{V}^{*}} &= \mb{K}_{\mb{Y}}-\mb{B}^{*}_2 \hspace{0.1in} \textrm{and} \label{eq:ConvIngr5}\\
%\mb{K}_{\mb{X}|\mb{U}^{*},\mb{V}^{*}} &= \mb{K}_{\mb{X}}-\mb{B}^{*}_1-\mb{B}^{*}_2, \label{eq:ConvIngr6}
%\end{align}
%where $\mb{U}^{*}$ and $\mb{V}^{*}$ are optimal Gaussian random vectors.
As discussed in Section \ref{sec:PropOptGaussSol}, $(\mb{B}^{*}_1,\mb{B}^{*}_2)$ gives three matrices $\mb{S}, \mb{T},$ and $\mb{W}$ that satisfy the properties in Theorem \ref{thm:PropOptGaussSol}. Using these properties, the optimal value of $\left (P_{G2} \right)$ can be expressed as
\begin{align}
v\left (P_{G2} \right) &= \frac{\mu}{2} \log \frac{\left |\mb{K_X}-\mb{B}^{*}_2 \right |}{ \left |\mb{K_X}-\mb{B}^{*}_1-\mb{B}^{*}_2 \right |}+\frac{1}{2} \log \frac{\left |\mb{K}_{\mb{Y}} \right |}{ \left |\mb{K_Y}-\mb{B}^{*}_2 \right |} \nonumber\\
&= \frac{\mu}{2} \log \frac{\left |[\mb{S,T}]^T\left(\mb{K_X}-\mb{B}^{*}_2\right)[\mb{S,T}] \right |}{ \left |[\mb{S,T}]^T\left(\mb{K_X}-\mb{B}^{*}_1-\mb{B}^{*}_2\right)[\mb{S,T}] \right |}+\frac{1}{2} \log \frac{\left |[\mb{S,W}]^T\mb{K}_{\mb{Y}}[\mb{S,W}] \right |}{ \left |[\mb{S,W}]^T\left(\mb{K_Y}-\mb{B}^{*}_2 \right)[\mb{S,W}] \right |} \label{eq:ConvIngr7}\\
&= \frac{\mu}{2} \log \frac{\left |\mat{\mb{S}^T \left(\mb{K_X}-\mb{B}^{*}_2\right) \mb{S} & \mb{0} \\ \mb{0} & \mb{T}^T \left(\mb{K_X}-\mb{B}^{*}_2\right) \mb{T}} \right |}{ \left | \mat{\mb{S}^T \left(\mb{K_X}-\mb{B}^{*}_1-\mb{B}^{*}_2\right) \mb{S} & \mb{0} \\ \mb{0} & \mb{T}^T \left(\mb{K_X}-\mb{B}^{*}_1-\mb{B}^{*}_2\right) \mb{T}} \right |} \nonumber\\
&\hspace{0.5in}+\frac{1}{2} \log \frac{\left |\mat{\mb{S}^T \mb{K}_{\mb{Y}} \mb{S} & \mb{0} \\ \mb{0} & \mb{W}^T \mb{K}_{\mb{Y}} \mb{W}} \right |}{ \left | \mat{\mb{S}^T \left(\mb{K_Y}-\mb{B}^{*}_2 \right) \mb{S} & \mb{0} \\ \mb{0} & \mb{W}^T \left(\mb{K_Y}-\mb{B}^{*}_2 \right) \mb{W}} \right |} \label{eq:ConvIngr8}\\
&= \frac{\mu}{2} \log \frac{\left |\mb{S}^T \left(\mb{K_X}-\mb{B}^{*}_2\right) \mb{S} \right| \left|\mb{T}^T \left(\mb{K_X}-\mb{B}^{*}_2\right) \mb{T}\right| }{ \left | \mb{S}^T \left(\mb{K_X}-\mb{B}^{*}_1-\mb{B}^{*}_2\right) \mb{S} \right| \left|\mb{T}^T \left(\mb{K_X}-\mb{B}^{*}_1-\mb{B}^{*}_2\right) \mb{T} \right |} \nonumber\\
&\hspace{0.5in}+\frac{1}{2} \log \frac{\left |\mb{S}^T \mb{K}_{\mb{Y}} \mb{S} \right| \left| \mb{W}^T \mb{K}_{\mb{Y}} \mb{W} \right |}{ \left | \mb{S}^T \left(\mb{K_Y}-\mb{B}^{*}_2 \right) \mb{S} \right| \left| \mb{W}^T \left(\mb{K_Y}-\mb{B}^{*}_2 \right) \mb{W} \right |} \nonumber\\
&= \frac{\mu}{2} \log \frac{\left |\mb{S}^T \left(\mb{K_X}-\mb{B}^{*}_2\right) \mb{S} \right| }{ \left | \mb{S}^T \mb{D} \mb{S} \right| } +\frac{1}{2} \log \frac{\left |\mb{S}^T \mb{K}_{\mb{Y}} \mb{S} \right|}{ \left | \mb{S}^T \left(\mb{K_Y}-\mb{B}^{*}_2 \right) \mb{S} \right| }, \label{eq:ConvIngr9}
\end{align}
where
\begin{enumerate}
\item[(\ref{eq:ConvIngr7})] follows because $[\mb{S,T}]$ and $[\mb{S,W}]$ are invertible,
\item[(\ref{eq:ConvIngr8})] follows because $\mb{S}$ and $\mb{T}$ are cross $\bigr\{(\mb{K_X}-\mb{B}^{*}_2),(\mb{K_X}-\mb{B}^{*}_1-\mb{B}^{*}_2)\bigr\}$-orthogonal, and $\mb{S}$ and $\mb{W}$ are cross $\bigr\{\mb{K_Y},(\mb{K_Y}-\mb{B}^{*}_2)\bigr\}$-orthogonal, and
\item[(\ref{eq:ConvIngr9})] follows from (\ref{eq:ActiveDistConst}) and the facts that
\begin{align*}
\mb{T}^T \left(\mb{K_X}-\mb{B}^{*}_2\right) \mb{T}  &=\mb{T}^T \left(\mb{K_X}-\mb{B}^{*}_1-\mb{B}^{*}_2\right) \mb{T} \hspace{0.1in} \textrm{and} \\
\mb{W}^T \mb{K_Y} \mb{W} &= \mb{W}^T \left(\mb{K_Y}-\mb{B}^{*}_2 \right) \mb{W}.
\end{align*}
\end{enumerate}

\subsection{Distortion Projection}
The special structure to the optimal Gaussian solution of $\left (P_{G2} \right)$ suggests a way to lower bound $\left (P\right)$ by projecting the sources $\mb{X}$ and $\mb{Y}$ on $\mb{S}$ and imposing the distortion constraint on the subspace spanned by the columns of $\mb{S}$. Note that the distortion constraint is tight on this subspace for the optimal Gaussian solution. We refer to this method of lower bounding $\left (P\right)$ as \emph{distortion projection}. Let us define
\begin{align*}
\tilde {\mb{X}} &\triangleq \mb{S}^T\mb{X},\\
\tilde {\mb{Y}} &\triangleq \mb{S}^T\mb{Y},\\
\tilde {\mb{D}} &\triangleq \mb{S}^T\mb{DS},\\
\tilde {\mb{B}}^{*}_1 &\triangleq \mb{S}^T{\mb{B}}^{*}_1 \mb{S},\\
\tilde {\mb{B}}^{*}_2 &\triangleq \mb{S}^T{\mb{B}}^{*}_2 \mb{S},\\
\tilde {\mb{M}}_1^{*} &\triangleq \left(\mb{S}^T\left(\mb{K_X}-\mb{B}^{*}_2\right)\mb{S}\right)^{-1}\mb{S}^T \left(\mb{K_X}-\mb{B}^{*}_2\right)\mb{M}_1^{*}\left(\mb{K_X}-\mb{B}^{*}_2\right)\mb{S}\left(\mb{S}^T\left(\mb{K_X}-\mb{B}^{*}_2\right)\mb{S}\right)^{-1}, \hspace{0.1in} \textrm{and}\\
\tilde {\mb{M}}_2^{*} &\triangleq \left(\mb{S}^T\left(\mb{K_Y}-\mb{B}^{*}_2\right)\mb{S}\right)^{-1}\mb{S}^T \left(\mb{K_Y}-\mb{B}^{*}_2\right)\mb{M}_2^{*}\left(\mb{K_Y}-\mb{B}^{*}_2\right)\mb{S}\left(\mb{S}^T\left(\mb{K_Y}-\mb{B}^{*}_2\right)\mb{S}\right)^{-1}.
\end{align*}
Since $\mb{S}$ has full column rank, we immediately have that
\begin{align*}
\mb{K}_{\tilde {\mb{X}}},\mb{K}_{\tilde {\mb{Y}}},\tilde {\mb{D}}  &\succ \mb{0},\\
\tilde {\mb{B}}^{*}_1, \tilde {\mb{B}}^{*}_2 &\succcurlyeq \mb{0}, \hspace{0.1in} \textrm{and}\\
\tilde {\mb{M}}^{*}_1, \tilde {\mb{M}}^{*}_2 &\succcurlyeq \mb{0}.
\end{align*}
The \emph{projected optimization problem} $(\tilde P)$ is now defined as
\begin{align*}
(\tilde P) \hspace {0.2in}\triangleq \hspace {0.5in}\min_{\mb{U,V}} \hspace {0.1in} & \mu I(\tilde {\mb{X}};\mb{U}|\mb{V})+I(\tilde {\mb{Y}};\mb{V}) \nonumber\\
\textrm{subject to} \hspace{0.1in} &\mb{K}_{\tilde{\mb{X}}|\mb{U,V}} \preccurlyeq \tilde{ \mb{D}} \hspace{0.1in} \textrm{and}\nonumber\\
&\tilde{\mb{X}}\leftrightarrow \tilde{\mb{Y}}\leftrightarrow \mb{V}.
\end{align*}
We next show that the \emph{main optimization problem} $\left (P\right)$ is lower bounded by the \emph{projected optimization problem} $(\tilde P)$. Since $[\mb{S,T}]$ and $[\mb{S,W}]$ are invertible and mutual information is nonnegative, we have
\begin{align}
\mu I(\mb{X};\mb{U}|\mb{V})+I(\mb{Y};\mb{V}) &= \mu I \left(\mb{S}^T \mb{X}, \mb{T}^T \mb{X};\mb{U}|\mb{V} \right)+I\left(\mb{S}^T\mb{Y},\mb{W}^T\mb{Y};\mb{V}\right) \nonumber\\
&= \mu I \left (\mb{S}^T \mb{X};\mb{U}|\mb{V} \right) + \mu I \left(\mb{T}^T \mb{X};\mb{U}|\mb{V},\mb{S}^T \mb{X} \right) +I\left (\mb{S}^T\mb{Y};\mb{V}\right)+I\left(\mb{W}^T\mb{Y};\mb{V}|\mb{S}^T \mb{Y}\right) \nonumber\\
&\ge \mu I \bigr(\tilde{ \mb{X}};\mb{U}|\mb{V} \bigr)+I\bigr (\tilde{\mb{Y}};\mb{V}\bigr).\label{eq:DistProj1}
\end{align}
Consider any $(\mb{U,V})$ feasible for $\left (P\right)$. Then
\begin{align}
&\mb{D} \succcurlyeq \mb{K}_{\mb{X}|\mb{U,V}} \hspace{0.1in} \textrm{and} \label{eq:DistProj2}\\
&\mb{X} \leftrightarrow \mb{Y} \leftrightarrow \mb{V} \label{eq:DistProj3}
\end{align}
Now (\ref{eq:DistProj2}) implies
\begin{align}
\tilde {\mb{D}} &= \mb{S}^T \mb{D}\mb{S} \succcurlyeq \mb{S}^T\mb{K}_{\mb{X}|\mb{U,V}} \mb{S} = \mb{K}_{\tilde {\mb{X}}|\mb{U,V}}, \label{eq:DistProj9}
\end{align}
and (\ref{eq:DistProj3}) yields
\begin{align}
0 &= I\left(\mb{X};\mb{V}|\mb{Y}\right)\nonumber\\
&= I\left(\mb{S}^T\mb{X};\mb{V}|\mb{Y}\right)+I\left(\mb{T}^T\mb{X};\mb{V}|\mb{Y},\mb{S}^T\mb{X}\right)\nonumber\\
&\ge I\left(\mb{S}^T\mb{X};\mb{V}|\mb{Y}\right) \label{eq:DistProj4}\\
&=I\left(\mb{S}^T\mb{X};\mb{V}|\mb{S}^T\mb{Y},\mb{W}^T\mb{Y}\right) \label{eq:DistProj5}\\
&=h\left(\mb{S}^T\mb{X}|\mb{S}^T\mb{Y},\mb{W}^T\mb{Y}\right)-h\left(\mb{S}^T\mb{X}|\mb{V},\mb{S}^T\mb{Y},\mb{W}^T\mb{Y}\right)\nonumber\\
&\ge h\left(\mb{S}^T\mb{X}|\mb{S}^T\mb{Y}\right)-h\left(\mb{S}^T\mb{X}|\mb{V},\mb{S}^T\mb{Y}\right)\label{eq:DistProj6}\\
&= I\left(\mb{S}^T\mb{X};\mb{V}|\mb{S}^T\mb{Y}\right)\nonumber\\
&= I\bigr(\tilde {\mb{X}};\mb{V}|\tilde{\mb{Y}}\bigr)\nonumber\\
&\ge 0, \label{eq:DistProj7}
\end{align}
where
\begin{enumerate}
\item[(\ref{eq:DistProj4})] and (\ref{eq:DistProj7}) follows because mutual information is nonnegative,
\item[(\ref{eq:DistProj5})] follows because $[\mb{S,W}]$ is invertible, and
\item[(\ref{eq:DistProj6})] follows because conditioning reduces entropy and we have from Theorem \ref{thm:PropOptGaussSol} that $\mb{W}^T\mb{Y}$ is independent of $\mb{S}^T\mb{Y}$, which implies that $\mb{W}^T\mb{Y}$ is also independent of $\mb{S}^T\mb{X}$ because $\mb{X}=\mb{Y}+\mb{N}$ by assumption.
\end{enumerate}
Now (\ref{eq:DistProj7}) is equivalent to
\begin{align*}
&\tilde{\mb{X}} \leftrightarrow \tilde{\mb{Y}} \leftrightarrow \mb{V},
\end{align*}
which together with (\ref{eq:DistProj9}) implies that $(\mb{U,V})$ is feasible for $(\tilde P)$. Hence, the feasible set of $\left (P\right)$ is contained in that of $(\tilde P)$. Moreover, (\ref{eq:DistProj1}) above implies that the objective of $\left (P\right)$ is no less than that of $(\tilde P)$. We therefore have that the \emph{projected optimization problem} $(\tilde P)$ lower bounds the \emph{main optimization problem} $\left (P\right)$, i.e.,
\begin{align}
v\left (P\right) &\ge v(\tilde P). \label{eq:DistProj8}
\end{align}
By restricting the solution space of $(\tilde P)$ to Gaussian distributions, we obtain its Gaussian version
\begin{align*}
(\tilde{P}_{G2}) \hspace {0.2in} \triangleq \hspace {0.2in}\min_{\tilde{\mb{B}}_1,\tilde{\mb{B}}_2} \hspace {0.2in} &\frac{\mu}{2} \log \frac{  |\mb{K}_{\tilde {\mb{X}}}- \tilde{\mb{B}}_2  |}{ |\mb{K}_{\tilde {\mb{X}}}- \tilde{\mb{B}}_1- \tilde{\mb{B}}_2|} + \frac{1}{2} \log \frac{  |\mb{K}_{\tilde {\mb{Y}}}  |}{ |\mb{K}_{\tilde {\mb{Y}}}-\tilde{\mb{B}}_2|}\\
\textrm{subject to} \hspace{0.1in} &\tilde{\mb{B}}_i \succcurlyeq \mb{0} \hspace {0.15 cm} \textrm{for all} \hspace {0.15 cm} i \in \{1,2\}, \hspace {0.15 cm} \textrm{and}\\
&\tilde{\mb{D}} \succcurlyeq \mb{K}_{\tilde {\mb{X}}}-\tilde{\mb{B}}_1-\tilde{\mb{B}}_2.
\end{align*}
It is easy to verify that the projected optimal Gaussian solution $(\tilde {\mb{B}}^{*}_1,\tilde {\mb{B}}^{*}_2)$ is feasible for $(\tilde{P}_{G2})$ and it meets the projected distortion constraint $\tilde{\mb{D}}$ with equality from (\ref{eq:ActiveDistConst}). We next show that $(\tilde {\mb{B}}^{*}_1,\tilde {\mb{B}}^{*}_2)$ is in fact optimal for $(\tilde P)$.

\emph{Remark 1:} If $r=m$, then there is no need for \emph{distortion projection} because $\mb{S}$ is invertible, and hence so is $\mb{\Delta}^{*}$. 
%It follows from Proposition \ref{prop:Zhang} that in this case the Gaussian achievable scheme is optimal \cite{Guo}.

\subsection{Source Enhancement}
In this subsection, we use the KKT conditions (\ref{eq:KKT1}) through (\ref{eq:KKT5}) satisfied by $({\mb{B}}^{*}_1,{\mb{B}}^{*}_2)$ to derive conditions that must be satisfied by $(\tilde {\mb{B}}^{*}_1,\tilde {\mb{B}}^{*}_2)$. These conditions are then used to define the \emph{enhanced optimization problem,} which lower bounds $(\tilde P)$. We show that the optimal solution to the \emph{enhanced optimization problem} is Gaussian, in particular $(\tilde {\mb{B}}^{*}_1,\tilde {\mb{B}}^{*}_2)$ is optimal for the problem. This will in turn prove that $(\tilde {\mb{B}}^{*}_1,\tilde {\mb{B}}^{*}_2)$ is optimal for $(\tilde P)$. This approach of lower bounding is referred to as the \emph{source enhancement} \cite{Guo} and is similar to the \emph{channel enhancement} idea of Weingarten \emph{et al.} \cite{Wein}.

We start with the following key lemma.
\begin{Lem} \label{lem:SrcEnh1}
For $\mb{K}_{\tilde{\mb{X}}}, \mb{K}_{\tilde{\mb{Y}}}, \tilde{\mb{D}}, \tilde {\mb{B}}^{*}_i$, and $\tilde {\mb{M}}_i^{*},$ where $i = 1, 2$, defined as above, the following hold
\begin{align}
\mb{I}_r =\frac{\mu}{2} \bigr(\mb{K}_{\tilde{\mb{X}}}-\tilde{\mb{B}}^{*}_2\bigr)^{-1} - \tilde {\mb{M}}_1^{*}&=\frac{1}{2} \bigr( \mb{K}_{\tilde{\mb{Y}}}-\tilde{\mb{B}}^{*}_2\bigr)^{-1} - \tilde {\mb{M}}_2^{*}, \label{eq:SrcEnh1}\\
\tilde{\mb{B}}^{*}_i\tilde {\mb{M}}_i^{*} &= \mb{0} \hspace{0.1in} \textrm{for all} \hspace{0.1in} i \in \{1, 2\}, \hspace{0.1in} \textrm{and} \label{eq:SrcEnh2}\\
\mb{K}_{\tilde{\mb{X}}}-\tilde{\mb{B}}^{*}_1-\tilde{\mb{B}}^{*}_2&=\tilde{\mb{D}}.\label{eq:SrcEnh3}
\end{align}
\end{Lem}
\begin{proof}
See Appendix D.
\end{proof}

%Note that conditions (\ref{eq:SrcEnh1}) through (\ref{eq:SrcEnh3}) are $r$-dimensional version of Zhang's sufficiency conditions, which were stated in the proof of Proposition \ref{prop:Zhang}. 
Let $\mb{K}_{\hat {\mb{X}}}$ and $\mb{K}_{\hat {\mb{Y}}}$ be two real symmetric matrices satisfying
\begin{align}
\frac{\mu}{2} \bigr( \mb{K}_{\tilde{\mb{X}}}-\tilde{\mb{B}}^{*}_2\bigr)^{-1} - \tilde {\mb{M}}_1^{*} &= \frac{\mu}{2} \bigr( \mb{K}_{\hat{\mb{X}}}-\tilde{\mb{B}}^{*}_2\bigr)^{-1} \hspace{0.1in} \textrm{and} \label{eq:SrcEnh4}\\
\frac{1}{2} \bigr( \mb{K}_{\tilde{\mb{Y}}}-\tilde{\mb{B}}^{*}_2\bigr)^{-1} - \tilde {\mb{M}}_2^{*} &= \frac{1}{2} \bigr( \mb{K}_{\hat{\mb{Y}}}-\tilde{\mb{B}}^{*}_2\bigr)^{-1}. \label{eq:SrcEnh5}
\end{align}
We now have the following lemma, which is similar to \cite[Lemmas 11, 12]{Wein}.
\begin{Lem}\label{lem:SrcEnh2}
For $\mb{K}_{\tilde {\mb{X}}}, \mb{K}_{\tilde {\mb{Y}}}, \mb{K}_{\hat {\mb{X}}}, \mb{K}_{\hat {\mb{Y}}}, \tilde{\mb{B}}^{*}_i, \tilde{\mb{M}}^{*}_i,$ $i = 1, 2$, defined as above, and $\mu > 1$, the following hold
\begin{align}
\mb{K}_{\hat {\mb{X}}} - \tilde{\mb{B}}^{*}_2 &= \frac{\mu}{2} \mb{I}_r, \label{eq:SrcEnh6}\\
\mb{K}_{\hat {\mb{Y}}} - \tilde{\mb{B}}^{*}_2 &= \frac{1}{2} \mb{I}_r, \label{eq:SrcEnh7}\\
\mb{K}_{\hat {\mb{X}}} &\succ \mb{K}_{\hat {\mb{Y}}} \succcurlyeq   \mb{K}_{\tilde {\mb{Y}}} \succ \mb{0}, \label{eq:SrcEnh8}\\
\mb{K}_{\hat {\mb{X}}} &\succcurlyeq   \mb{K}_{\tilde {\mb{X}}} \succ \mb{0},\label{eq:SrcEnh9}\\
\frac{|\mb{K}_{\tilde{\mb{Y}}}|}{|\mb{K}_{\tilde{\mb{Y}}}-\tilde{\mb{B}}^{*}_2|} &= \frac{|\mb{K}_{\hat{\mb{Y}}}|}{|\mb{K}_{\hat{\mb{Y}}}-\tilde{\mb{B}}^{*}_2|}, \hspace{0.1in} \textrm{and} \label{eq:SrcEnh10}\\
\frac{|\mb{K}_{\tilde{\mb{X}}}-\tilde{\mb{B}}^{*}_2|}{|\mb{K}_{\tilde{\mb{X}}}-\tilde{\mb{B}}^{*}_1-\tilde{\mb{B}}^{*}_2|} &= \frac{|\mb{K}_{\hat{\mb{X}}}-\tilde{\mb{B}}^{*}_2|}{|\mb{K}_{\hat{\mb{X}}}-\tilde{\mb{B}}^{*}_1-\tilde{\mb{B}}^{*}_2|}. \label{eq:SrcEnh11}
\end{align}
\end{Lem}
\begin{proof}
See Appendix E.
\end{proof}
%\begin{Lem} Let
%\[
%\mb{K}_{\hat {\mb{Y}}} \triangleq \left( \mb{K}_{\tilde{\mb{Y}}}^{-1}-2 \tilde{\mb{M}}_2^{*}\right)^{-1}.
%\]
%Then the following hold
%\begin{align}
%\mb{K}_{\hat {\mb{Y}}} &\succcurlyeq   \mb{K}_{\tilde {\mb{Y}}}\\
%\frac{1}{2} \left( \mb{K}_{\tilde{\mb{Y}}}-\tilde{\mb{B}}^{*}_2\right)^{-1} - \tilde {\mb{M}}_2^{*} &= \frac{1}{2} \left( \mb{K}_{\hat{\mb{Y}}}-\tilde{\mb{B}}^{*}_2\right)^{-1}\\
%\frac{\left|\mb{K}_{\tilde{\mb{Y}}}\right|}{\left|\mb{K}_{\tilde{\mb{Y}}}-\tilde{\mb{B}}^{*}_2\right|} &= \frac{\left|\mb{K}_{\hat{\mb{Y}}}\right|}{\left|\mb{K}_{\hat{\mb{Y}}}-\tilde{\mb{B}}^{*}_2\right|}.
%\end{align}
%\end{Lem}
%\begin{proof}
%See Appendix C.
%\end{proof}
%\begin{Lem} Let
%\[
%\mb{K}_{\hat {\mb{X}}} \triangleq \left[ \left(\mb{K}_{\tilde{\mb{X}}}-\tilde{\mb{B}}^{*}_2\right)^{-1}-\frac{2}{\mu} \tilde{\mb{M}}_1^{*}\right]^{-1}+\tilde{\mb{B}}^{*}_2.
%\]
%Then the following hold
%\begin{align}
%\mb{K}_{\hat {\mb{X}}} &\succcurlyeq   \mb{K}_{\tilde {\mb{X}}}\\
%\frac{\mu}{2} \left( \mb{K}_{\tilde{\mb{X}}}-\tilde{\mb{B}}^{*}_2\right)^{-1} - \tilde {\mb{M}}_1^{*} &= \frac{\mu}{2} \left( \mb{K}_{\hat{\mb{X}}}-\tilde{\mb{B}}^{*}_2\right)^{-1}\\
%\frac{\left|\mb{K}_{\tilde{\mb{X}}}-\tilde{\mb{B}}^{*}_2\right|}{\left|\mb{K}_{\tilde{\mb{X}}}-\tilde{\mb{B}}^{*}_1-\tilde{\mb{B}}^{*}_2\right|} &= \frac{\left|\mb{K}_{\hat{\mb{X}}}-\tilde{\mb{B}}^{*}_2\right|}{\left|\mb{K}_{\hat{\mb{X}}}-\tilde{\mb{B}}^{*}_1-\tilde{\mb{B}}^{*}_2\right|}.
%\end{align}
%\end{Lem}
%\begin{proof}
%See Appendix D.
%\end{proof}

Let $\hat{\mb{X}}$ and $\hat{\mb{Y}}$ be two zero-mean Gaussian random vectors with covariance matrices $\mb{K}_{\hat{\mb{X}}}$ and $\mb{K}_{\hat{\mb{Y}}}$, respectively. Since $\mb{K}_{\hat{\mb{X}}} \succ \mb{K}_{\hat{\mb{Y}}}$ from (\ref{eq:SrcEnh8}), we can write
\[
\hat{\mb{X}} = \hat{\mb{Y}}+\hat{\mb{N}},
\]
where $\hat{\mb{N}}$ is a zero-mean Gaussian random vector with the covariance matrix $$\mb{K}_{\hat{\mb{N}}}=\mb{K}_{\hat{\mb{X}}}-\mb{K}_{\hat{\mb{Y}}}=\frac{\mu -1}{2} \mb{I}_r$$ and is independent of $\hat{\mb{Y}}$. Similarly, we can use (\ref{eq:SrcEnh8}) and (\ref{eq:SrcEnh9}) to relate $\hat{\mb{X}}$ and $\hat{\mb{Y}}$ with $\tilde{\mb{X}}$ and $\tilde{\mb{Y}}$, respectively, and write
\begin{align*}
\hat{\mb{X}} &= \tilde{\mb{X}}+\mb{N}_1 \hspace{0.1in} \textrm{and} \\
\hat{\mb{Y}} &= \tilde{\mb{Y}}+\mb{N}_2,
\end{align*}
where ${\mb{N}}_1$ and ${\mb{N}}_2$ are two zero-mean Gaussian random vectors with covariance matrices
\begin{align*}
\mb{K}_{\mb{N}_1}&=\mb{K}_{\hat{\mb{X}}}-\mb{K}_{\tilde{\mb{X}}} \hspace{0.1in} \textrm{and}\\
\mb{K}_{\mb{N}_2}&=\mb{K}_{\hat{\mb{Y}}}-\mb{K}_{\tilde{\mb{Y}}},
\end{align*}
respectively, and they are independent of $\tilde{\mb{X}}$ and $\tilde{\mb{Y}}$. Using (\ref{eq:SrcEnh3}), we define
\begin{align}
\hat{ \mb{D}} &\triangleq \tilde{ \mb{D}}+\mb{K}_{{\mb{N}}_1}=\mb{K}_{\hat{\mb{X}}}-\tilde{\mb{B}}^{*}_1-\tilde{\mb{B}}^{*}_2. \label{eq:SrcEnh12}
\end{align}
The \emph{enhanced optimization problem} $(\hat P)$ is now defined as
\begin{align*}
(\hat P) \hspace {0.2in}\triangleq \hspace {0.5in}\min_{\mb{U,V}} \hspace {0.1in} & \mu I(\hat {\mb{X}};\mb{U}|\mb{V})+I(\hat {\mb{Y}};\mb{V}) \nonumber\\
\textrm{subject to} \hspace{0.1in} &\mb{K}_{\hat{\mb{X}}|\mb{U,V}} \preccurlyeq \hat{ \mb{D}}\hspace{0.1in} \textrm{and} \nonumber\\
&\hat{\mb{X}}\leftrightarrow \hat{\mb{Y}}\leftrightarrow \mb{V}.
\end{align*}
We next show that $(\hat P)$ lower bounds $(\tilde P)$. Consider any $(\mb{U,V})$ feasible for $(\tilde P)$. Without loss of optimality, we can assume that the joint distribution between $\tilde {\mb{X}}, \tilde {\mb{Y}}, \mb{U},$ and $\mb{V}$ is
\[
\tilde p \triangleq p_{\tilde {\mb{X}}, \tilde {\mb{Y}}} p_{\mb{U} | \tilde {\mb{X}}, \mb{V}} p_{\mb{V} | \tilde {\mb{Y}}}.
\]
Now, $\tilde p$ induces two conditional distributions as follows
\begin{align*}
p_{\mb{V} | \hat {\mb{Y}}} &= \int_{\tilde {\mb{Y}}} p_{\mb{V} | \tilde {\mb{Y}}} p_{\tilde {\mb{Y}}|\hat {\mb{Y}}} \\
p_{\mb{U} | \hat {\mb{X}}, \mb{V}} &= \int_{\tilde {\mb{X}}} p_{\mb{U} | \tilde {\mb{X}}, \mb{V}} p_{\tilde {\mb{X}}|\hat {\mb{X}}, \mb{V}},
\end{align*}
where
\begin{align*}
p_{\tilde {\mb{X}}|\hat {\mb{X}}, \mb{V}} = \frac{p_{\tilde {\mb{X}}, \hat {\mb{X}}} p_{{\mb{V}}|\tilde {\mb{X}}}}{\int_{\tilde {\mb{X}}} p_{\tilde {\mb{X}}, \hat {\mb{X}}} p_{{\mb{V}}|\tilde {\mb{X}}}}.
\end{align*}
Then
\[
\hat p \triangleq p_{\hat {\mb{X}}, \hat {\mb{Y}}} p_{\mb{U} | \hat {\mb{X}}, \mb{V}} p_{\mb{V} | \hat {\mb{Y}}}
\]
is a joint distribution between $\hat {\mb{X}}, \hat {\mb{Y}}, \mb{U},$ and $\mb{V}$. It is clear that $\hat p$ satisfies the Markov condition
\begin{align}
\hat{\mb{X}} \leftrightarrow \hat{\mb{Y}} \leftrightarrow \mb{V}. \label{eq:SrcEnh13}
\end{align}
Moreover, (\ref{eq:SrcEnh12}) and the distortion constraint in the definition of $(\tilde P)$ yield
\begin{align}
\mb{K}_{\hat{\mb{X}}|\mb{U,V}}&=\mb{K}_{\tilde{\mb{X}}|\mb{U,V}}+\mb{K}_{{\mb{N}}_1} \preccurlyeq \tilde{\mb{D}}+\mb{K}_{{\mb{N}}_1}= \hat{\mb{D}}. \label{eq:SrcEnh14}
\end{align}
We next use the chain rule of mutual information to obtain
\begin{align*}
I(\tilde{\mb{X}},\hat{\mb{X}};\mb{U}|\mb{V})&=I(\hat{\mb{X}};\mb{U}|\mb{V})+I(\tilde{\mb{X}};\mb{U}|\mb{V},\hat{\mb{X}})\\
&=I(\tilde{\mb{X}};\mb{U}|\mb{V})+I(\hat{\mb{X}};\mb{U}|\mb{V},\tilde{\mb{X}})\\
&=I(\tilde{\mb{X}};\mb{U}|\mb{V})
\end{align*}
and
\begin{align*}
I(\tilde{\mb{Y}},\hat{\mb{Y}};\mb{V})&=I(\hat{\mb{Y}};\mb{V})+I(\tilde{\mb{Y}};\mb{V}|\hat{\mb{Y}})\\
&=I(\tilde{\mb{Y}};\mb{V})+I(\hat{\mb{Y}};\mb{V}|\tilde{\mb{Y}})\\
&=I(\tilde{\mb{Y}};\mb{V}).
\end{align*}
Since mutual information is nonnegative, these imply that
\begin{align}
I(\tilde{\mb{X}};\mb{U}|\mb{V}) \ge I(\hat{\mb{X}};\mb{U}|\mb{V}) \label{eq:SrcEnh15}
\end{align}
and
\begin{align}
I(\tilde{\mb{Y}};\mb{V}) \ge I(\hat{\mb{Y}};\mb{V}) \label{eq:SrcEnh16}
\end{align}
Now (\ref{eq:SrcEnh13}) and (\ref{eq:SrcEnh14}) together imply that the distribution $\hat p$, and hence $(\mb{U,V})$, is feasible for $(\hat P)$. Therefore, the feasible set of $(\tilde P)$ is contained in that of $(\hat P)$. Moreover, (\ref{eq:SrcEnh15}) and (\ref{eq:SrcEnh16}) assert that the objective value of $(\hat P)$ is no more than that of $(\tilde P)$. We therefore conclude that the \emph{enhanced optimization problem}  $(\hat P)$ lower bounds the \emph{projected optimization problem}  $(\tilde P)$, i.e.,
\begin{align}
v(\tilde P) \ge v(\hat P). \label{eq:SrcEnh17}
\end{align}

\emph{Remark 2:} If $r<m=r+p$, then there is no need to enhance the source $\tilde{\mb{X}}$ and the distortion $\tilde{\mb{D}}$ because $\mb{M}_1^{*}=\mb{TT}^T$ from Corollary \ref{cor:CorrThm2}, and hence $\tilde{\mb{M}}_1^{*}=\mb{0}$. Similarly, if $r<m=r+q$, then there is no need to enhance the source $\tilde{\mb{Y}}$ because $\mb{M}_2^{*}=\mb{WW}^T$ from Corollary \ref{cor:CorrThm2} again, and hence $\tilde{\mb{M}}_2^{*}=\mb{0}$. Finally, if $r<m=r+p=r+q$, then there is no need for \emph{source enhancement}.

\subsection{Oohama's Approach}
We now apply Oohama's approach \cite{Oohama} to prove that $(\tilde {\mb{B}}^{*}_1,\tilde {\mb{B}}^{*}_2)$ is optimal for $(\hat P).$ The objective of $(\hat P)$ can be decomposed as
\begin{align}
\mu I \bigr (\hat {\mb{X}};\mb{U}|\mb{V}\bigr )+I \bigr (\hat {\mb{Y}};\mb{V}\bigr ) = \mu I \bigr (\hat{\mb{X}};\mb{U,V}\bigr ) -\bigr [\mu I \bigr (\hat{\mb{X}};\mb{V}\bigr ) - I \bigr (\hat{\mb{Y}};\mb{V}\bigr ) \bigr]. \label{eq:Oohama1}
\end{align}
We next define two subproblems that are used to lower bound the \emph{enhanced optimization problem} $(\hat P)$. The first subproblem $(\hat P_1)$ minimizes the first mutual information in the right-hand-side of (\ref{eq:Oohama1}) subject to the distortion constraint in $(\hat P)$ and the second subproblem $(\hat P_2)$ maximizes the expression within the parenthesis in the right-hand-side of (\ref{eq:Oohama1}) subject to the Markov condition in $(\hat P)$. In other words, $(\hat P_1)$ is defined as
\begin{align*}
(\hat P_1) \hspace {0.2in}\triangleq \hspace {0.5in}\min_{\mb{U,V}} \hspace {0.1in} & \mu I \bigr (\hat{\mb{X}};\mb{U,V}\bigr ) \\
\textrm{subject to} \hspace{0.1in} &\mb{K}_{\hat{\mb{X}}|\mb{U,V}} \preccurlyeq \hat{ \mb{D}},
\end{align*}
and $(\hat P_2)$ is defined as
\begin{align*}
(\hat P_2) \hspace {0.2in}\triangleq \hspace {0.5in}\max_{\mb{V}} \hspace {0.1in} &\mu I \bigr (\hat{\mb{X}};\mb{V}\bigr ) - I \bigr (\hat{\mb{Y}};\mb{V}\bigr ) \\
\textrm{subject to} \hspace{0.1in} &\hat{\mb{X}}\leftrightarrow \hat{\mb{Y}}\leftrightarrow \mb{V}.
\end{align*}
It is clear from the decomposition in (\ref{eq:Oohama1}) and from the definitions of $(\hat P), (\hat P_1),$ and $(\hat P_2)$ that $(\hat P_1)$ and $(\hat P_2)$ lower bound $(\hat P)$, i.e.,
\begin{align}
v (\hat P) &\ge v(\hat P_1 ) - v(\hat P_2). \label{eq:Oohama2}
\end{align}
We now give two lemmas about the optimal solutions to subproblems $(\hat P_1)$ and $(\hat P_2)$.
\begin{Lem} \label{lem:Oohama1}
A Gaussian $(\mb{U,V})$ with the conditional covariance matrix $$\mb{K}_{\hat{\mb{X}}|\mb{U,V}}=\mb{K}_{\hat{\mb{X}}}-\tilde{\mb{B}}^{*}_1-\tilde{\mb{B}}^{*}_2=\hat{\mb{D}}$$ is optimal for the subproblem $(\hat P_1)$, and the optimal value is
\begin{align}
v(\hat P_1) &= \frac{\mu}{2} \log \frac{\bigr | \mb{K}_{\hat{\mb{X}}} \bigr |}{\bigr |\hat {\mb{D}}\bigr|}. \label{eq:Oohama3}
\end{align}
\end{Lem}
\begin{proof}
See Appendix F.
\end{proof}
\begin{Lem} \label{lem:Oohama2}
A Gaussian $\mb{V}$ with the conditional covariance matrix $$\mb{K}_{\hat{\mb{Y}}|\mb{V}}=\mb{K}_{\hat{\mb{Y}}}-\tilde{\mb{B}}^{*}_2$$ is optimal for the subproblem $(\hat P_2)$, and the optimal value is
\begin{align}
v(\hat P_2) &= \frac{\mu}{2} \log \frac {\left | \mb{K}_{\hat{\mb{X}}} \right |}{\bigr | \mb{K}_{\hat{\mb{X}}}-\tilde{\mb{B}}^{*}_2 \bigr |}-\frac{1}{2} \log \frac {\left | \mb{K}_{\hat{\mb{Y}}} \right |}{\bigr | \mb{K}_{\hat{\mb{Y}}}-\tilde{\mb{B}}^{*}_2 \bigr |}. \label{eq:Oohama4}
\end{align}
\end{Lem}
\begin{proof}
See Appendix G.
\end{proof}
Substituting (\ref{eq:Oohama3}) and (\ref{eq:Oohama4}) into (\ref{eq:Oohama2}), we obtain
\begin{align}
v(\hat P) &\ge \frac{\mu}{2} \log \frac{\bigr | \mb{K}_{\hat{\mb{X}}}-\tilde{\mb{B}}^{*}_2 \bigr |}{\bigr |\hat {\mb{D}}\bigr|}+\frac{1}{2} \log \frac {\bigr | \mb{K}_{\hat{\mb{Y}}} \bigr |}{\bigr | \mb{K}_{\hat{\mb{Y}}}-\tilde{\mb{B}}^{*}_2 \bigr |} \nonumber\\
&=\frac{\mu}{2} \log \frac{\bigr | \mb{K}_{\tilde{\mb{X}}}-\tilde{\mb{B}}^{*}_2 \bigr |}{\bigr |\tilde {\mb{D}}\bigr|}+\frac{1}{2} \log \frac {\bigr | \mb{K}_{\tilde{\mb{Y}}} \bigr |}{\bigr | \mb{K}_{\tilde{\mb{Y}}}-\tilde{\mb{B}}^{*}_2 \bigr |}\label{eq:Oohama5}\\
&= v\bigr (P_{G2} \bigr), \label{eq:Oohama6}
\end{align}
where
\begin{enumerate}
\item[(\ref{eq:Oohama5})] follows from (\ref{eq:SrcEnh3}), (\ref{eq:SrcEnh10}), (\ref{eq:SrcEnh11}), and (\ref{eq:SrcEnh12}), and
\item[(\ref{eq:Oohama6})] follows from (\ref{eq:ConvIngr9}).
\end{enumerate}
We conclude from (\ref{eq:DistProj8}), (\ref{eq:SrcEnh17}), and (\ref{eq:Oohama6}) that
\begin{align*}
v\left (P\right) &\ge v\left (P_{G2} \right).
\end{align*}
It now follows from this and (\ref{eq:ConvIngr4}) that
\begin{align}
v\left (P\right) &= v\left (P_{G1} \right)=v\left (P_{G2} \right), \label{eq:Oohama7}
\end{align}
which proves that a Gaussian $(\mb{U,V})$ is optimal for the \emph{main optimization problem} $\left (P\right)$. This completes the proof of Theorem \ref{thm:MainOptProb}.

\section{Converse Proof of Theorem \ref{thm:MainThm}}\label{sec:ConvProof}
Liu and Viswanath gave a single-letter outer bound to the rate region in \cite{Liu}. We shall use a similar outer bound that is reminiscent of the Berger-Tung outer bound~\cite{Berger,Tung}.
\begin{Lem} \label{lem:SinLetterOuterBound}
If the rate-distortion vector $(R_1, R_2, \mb{D})$ is achievable, then there exist random vectors $\mb{U}$ and $\mb{V}$ such that
\begin{align*}
R_1 &\ge I(\mb{X};\mb{U}|\mb{V}), \\
R_2 &\ge I(\mb{Y};\mb{V}), \\
\mb{D} &\succcurlyeq \mb{K}_{\mb{X}|\mb{U,V}}, \hspace{0.1in} \textrm{and} \\
\mb{X} &\leftrightarrow \mb{Y} \leftrightarrow \mb{V}.
\end{align*}
\end{Lem}
The proof of the lemma is similar to \cite[Lemma 2]{Rahman1} and is omitted. We are now ready to prove the converse of Theorem \ref{thm:MainThm}. If $(R_1,R_2, \mb{D})$ is achievable, then
\begin{align}
\mu R_1 + R_2 &\ge v\left (P\right) \label{eq:Conv1}\\
&= \left\{
\begin{array}{l l}
  v\left (P_{pt-pt} \right) & \quad \mbox{if $0 \le \mu \le 1$} \\
  v\left (P_{G1} \right) & \quad \mbox{if $\mu > 1$}\\ \end{array} \right.  \label{eq:Conv2}\\
  &= \mathcal{R}^{*}(\mb{D},\mu),  \label{eq:Conv3}
\end{align}
where
\begin{enumerate}
\item[(\ref{eq:Conv1})] follows from Lemma \ref{lem:SinLetterOuterBound}, and
\item[(\ref{eq:Conv2})] follows from (\ref{eq:OptVal1}) and (\ref{eq:Oohama7}).
\end{enumerate}
And if $(R_1,R_2,\mb{D}) \in \overline {\mathcal{RD}}$, then (\ref{eq:Conv3}) again holds because $\mathcal{R}^{*}(\mb{D},\mu)$ is continuous in $\mb{D}$. So, (\ref{eq:Conv3}) is a lower bound for any $(R_1, R_2)$ in the rate region $\mathcal{R}(\mb{D})$. Hence,
\begin{align*}
\mathcal{R}(\mb{D},\mu) &= \inf_{(R_1,R_2) \in \mathcal{R}(\mb{D}) } \mu R_1 + R_2\\
&\ge \mathcal{R}^{*}(\mb{D},\mu).
\end{align*}
This completes the converse proof of Theorem \ref{thm:MainThm}.

\emph{Remark 3:} It follows from Theorem \ref{thm:MainThm} and Lemma \ref{lem:BTRegion} that one
can add the constraints
\begin{align*}
& \mb{U} \leftrightarrow \mb{X} \leftrightarrow \mb{Y} \leftrightarrow
        \mb{V} \hspace{0.1in} \textrm{and} \\
& (\mb{U},\mb{V},\mb{X},\mb{Y}) \ \text{are jointly Gaussian}
\end{align*}
to the optimization problem
\begin{align*}
\left (P\right) \hspace {0.2in}\triangleq \hspace {0.5in}\min_{\mb{U,V}} \hspace {0.1in} & \mu I(\mb{X};\mb{U}|\mb{V})+I(\mb{Y};\mb{V}) \nonumber\\
\textrm{subject to} \hspace{0.1in} &\mb{K}_{\mb{X}|\mb{U,V}} \preccurlyeq \mb{D} \hspace{0.1in} \textrm{and} \nonumber\\
&\mb{X}\leftrightarrow \mb{Y}\leftrightarrow \mb{V},
\end{align*}
without changing its optimal value.
%Since the optimal solution to $(P)$ is Gaussian, we can impose the long Markov condition
%\[
%\mb{U} \leftrightarrow \mb{X} \leftrightarrow \mb{Y} \leftrightarrow \mb{V}
%\]
%in its definition without loss of optimality. Recall that we need the long Markov condition (C1) in the Gaussian achievable scheme.

%An important technical point worth pointing out is that in the {Gaussian achievable scheme}, described in Section 3.2, the long Markov chain condition (C1) must hold. It is therefore required to show that in the definition of the \emph{main optimization problem} $\left (P\right)$, we can impose the long Markov condition without loss of optimality. It is sufficient to show the following.
%\begin{Lem}
%The optimal solution $(\mb{U,V})$ to $\left (P\right)$ satisfies the long Markov chain
%\[
%U \leftrightarrow \mb{X} \leftrightarrow \mb{Y} \leftrightarrow \mb{V}.
%\]
%\end{Lem}
%\begin{proof}
%See Appendix I.
%\end{proof}

\section{Solution for the General Case}\label{sec:Extension}
In this section, we lift the assumptions on $\mb{K_X}, \mb{K_Y},$ and $\mb{D}$ and allow them to be any positive semidefinite matrices. We shall show that the Gaussian achievable scheme is optimal for this general problem. For this section, we denote the rate region of the problem by $\mathcal{R}\left(\mb{K_X}, \mb{K_Y}, \mb{D}\right)$. Note that $\mb{K_X}$ and $\mb{K_Y}$ completely specify the joint distribution of $\mb{X}$ and $\mb{Y}$ because we continue to assume that $\mb{X} = \mb{Y} + \mb{N}$. Similarly, $\mathcal{R}_G\left(\mb{K_X}, \mb{K_Y}, \mb{D}\right)$ is used to denote the rate region achieved by the Gaussian achievable scheme. We use $\mathcal{R}\left(\mb{K_X}, \mb{K_Y}, \mb{D}, \mu\right)$ and $\mathcal{R}_G\left(\mb{K_X}, \mb{K_Y}, \mb{D}, \mu\right)$ to denote the two minimum weighted sum-rates. Likewise, we denote the set $\mathcal{S}$ defined in Section \ref{subsec:scheme} by $\mathcal{S}(\mb{K}_{ {\mb{X}}}, \mb{K}_{{\mb{Y}}},\mb{D})$. We use similar notation later in the section. We start with the following extension.

\begin{Thm} \label{thm:Extension1}
If $\mb{K_X}$ and $\mb{D}$ are positive definite, and $\mb{K_Y}$ is positive semidefinite, then
\[
\mathcal{R}\left(\mb{K_X}, \mb{K_Y}, \mb{D}, \mu\right)=\mathcal{R}_G\left(\mb{K_X}, \mb{K_Y}, \mb{D}, \mu\right).
\]
\end{Thm}
\begin{proof}
It suffices to prove that
\[
\mathcal{R}\left(\mb{K_X}, \mb{K_Y}, \mb{D}, \mu\right)\ge\mathcal{R}_G\left(\mb{K_X}, \mb{K_Y}, \mb{D}, \mu\right).
\]
If $\mb{K_Y}$ is positive definite (hence nonsingular), then the result follows from Theorem \ref{thm:MainThm}. We therefore assume that $\mb{K_Y}$ is singular and has a rank $p < m$.  The eigen decomposition of $\mb{K_Y}$ is
\[
\mb{K_Y} = \mb{Q \Sigma}\mb{Q}^T,
\]
where $\mb{Q}$ is an orthogonal matrix and
\[
\mb{\Sigma} = \textrm{Diag} (\alpha_1,\dots,\alpha_p,0,\dots,0).
\]
Let us partition $\mb{Q}$ as
\[
\mb{Q} = [\mb{Q}_1, \mb{Q}_2 ],
\]
where $\mb{Q}_1$ is an $m \times p$ matrix. Let us define
\begin{align*}
\mb{Q}^T \mb{K_N Q} &\triangleq \mat{\mb{E} & \mb{F}^T \\ \mb{F} & \mb{G}},
\end{align*}
where $\mb{E}$, $\mb{F},$ and $\mb{G}$ are submatrices of dimensions $p \times p$, $ (m-p) \times p,$ and $(m-p) \times (m-p)$, respectively. Since $\mb{Q}_2^T\mb{K_YQ_2} = \mb{0}$ and $\mb{X}=\mb{Y}+\mb{N},$ we have that
\[
\mb{G}=\mb{Q}_2^T\mb{K_NQ_2}=\mb{Q}_2^T\mb{K_XQ_2} \succ \mb{0},
\]
i.e., $\mb{G}$ is positive definite. Using this, we define
\begin{align*}
\mb{A} &\triangleq  \mat{\mb{I}_p & -\mb{F}^T \mb{G}^{-1} \\ \mb{0} & \mb{I}_{m-p}} \mb{Q}^T.
\end{align*}
$\mb{A} $ defines a transformed problem in which the transformed sources are
\begin{align*}
\bar {\mb{X}} &\triangleq \mb{A} \mb{X} \hspace{0.05in} \textrm{and} \\
\bar {\mb{Y}} &\triangleq \mb{A} \mb{Y},
\end{align*}
which satisfy
\[
\bar {\mb{X}} = \bar {\mb{Y}} + \bar {\mb{N}},
\]
where $\bar {\mb{N}} \triangleq \mb{A} \mb{N},$ and the transformed distortion matrix is
\begin{align*}
\bar {\mb{D}} &\triangleq\mb{A} \mb{D} \mb{A}^T.
\end{align*}
The covariance matrix of the transformed source $\bar {\mb{Y}}$ is
\begin{align*}
\mb{K}_{\bar {\mb{Y}}} &= \mb{A} \mb{K}_{\mb{Y}} \mb{A}^T = \mb{\Sigma}= \mat{\mb{\Sigma}_1 & \mb{0} \\ \mb{0} & \mb{0}},
\end{align*}
where
\begin{align*}
\mb{\Sigma}_1 &\triangleq \textrm{Diag} (\alpha_1,\dots,\alpha_p),
\end{align*}
and the covariance matrix of $\bar {\mb{N}}$ is
\begin{align*}
\mb{K}_{\bar {\mb{N}}} &= \mb{A} \mb{K}_{\mb{N}} \mb{A}^T \nonumber\\
&= \mat{\mb{I}_p & -\mb{F}^T \mb{G}^{-1} \\ \mb{0} & \mb{I}_{m-p}} \mat{\mb{E} & \mb{F}^T \\ \mb{F} & \mb{G}} \mat{\mb{I}_p & \mb{0} \\-\mb{G}^{-1}\mb{F} & \mb{I}_{m-p}} \nonumber\\
&=\mat{\mb{E}-\mb{F}^T\mb{G}^{-1}\mb{F} & \mb{0} \\ \mb{0} & \mb{G}}.
\end{align*}
Using these, the covariance matrix of the transformed source $\bar {\mb{X}}$ can be expressed as
\begin{align*}
\mb{K}_{\bar {\mb{X}}} &= \mb{K}_{\bar {\mb{Y}}}+\mb{K}_{\bar {\mb{N}}} \nonumber\\
&=\mat{\mb{\Sigma}_1+\mb{E}-\mb{F}^T\mb{G}^{-1}\mb{F} & \mb{0} \\ \mb{0} & \mb{G}}.
\end{align*}
Since $\mb{A}$ is invertible, the above transformation is information lossless, and hence the transformed problem is equivalent to the original problem. Therefore,
\begin{align*}
\mathcal{R}\left(\mb{K_X}, \mb{K_Y}, \mb{D}, \mu\right)&=\mathcal{R}\left(\mb{K}_{\bar {\mb{X}}}, \mb{K}_{\bar {\mb{Y}}}, \bar {\mb{D}}, \mu\right) \hspace{0.05in} \textrm{and} \\
\mathcal{R}_G\left(\mb{K_X}, \mb{K_Y}, \mb{D}, \mu\right)&=\mathcal{R}_G\left(\mb{K}_{\bar {\mb{X}}}, \mb{K}_{\bar {\mb{Y}}}, \bar {\mb{D}}, \mu\right).
\end{align*}
So, it is sufficient to prove that
\[
\mathcal{R}\left(\mb{K}_{\bar {\mb{X}}}, \mb{K}_{\bar {\mb{Y}}}, \bar {\mb{D}}, \mu\right) \ge \mathcal{R}_G\left(\mb{K}_{\bar {\mb{X}}}, \mb{K}_{\bar {\mb{Y}}}, \bar {\mb{D}}, \mu\right).
\]

Let us define the following matrices
\begin{align*}
\mb{K}_{\bar {\mb{N}}_1^{(n)}} &\triangleq \mat{\mb{0} & \mb{0} \\ \mb{0} & \frac{1}{n}\mb{G}} \hspace{0.05in} \textrm{and}\\
\mb{K}_{\bar {\mb{N}}_2^{(n)}} &\triangleq \mat{\mb{E}-\mb{F}^T\mb{G}^{-1}\mb{F} & \mb{0} \\ \mb{0} & \left(1-\frac{1}{n}\right) \mb{G}},
\end{align*}
where $n$ is a positive integer. It is clear that these matrices are positive semidefinite and they satisfy
\[
\mb{K}_{\bar {\mb{N}}}=\mb{K}_{\bar {\mb{N}}_1^{(n)}}+\mb{K}_{\bar {\mb{N}}_2^{(n)}}.
\]
Let ${\bar {\mb{N}}_1}^{(n)}$ and ${\bar {\mb{N}}_2}^{(n)}$ be zero-mean vector Gaussian sources with covariance matrices $\mb{K}_{\bar {\mb{N}}_1^{(n)}}$ and $\mb{K}_{\bar {\mb{N}}_2^{(n)}},$ respectively. In addition, suppose they are independent of each other and all other vector Gaussian sources. We can then write
\[
{\bar {\mb{X}}}={\bar {\mb{Y}}}+{\bar {\mb{N}}_1}^{(n)}+{\bar {\mb{N}}_2}^{(n)}.
\]
Let us consider a new problem in which encoder 1 has access to ${\bar {\mb{X}}},$ encoder 2 has access to $\left({\bar {\mb{Y}}},{\bar {\mb{N}}_1}^{(n)}\right),$ and the distortion constraint on ${\bar {\mb{X}}}$ is $\bar {\mb{D}}.$ This problem is clearly a relaxation to the original problem because encoder 2 has access to more information about ${\bar {\mb{X}}}$ than the original problem. In other words, any feasible scheme for the original problem is also feasible for this new problem. Now since there is no distortion constraint on $\bar {\mb{Y}}$ and the sufficient statistic of ${\bar {\mb{X}}}$ in $\left({\bar {\mb{Y}}},{\bar {\mb{N}}_1}^{(n)}\right)$ is ${\bar {\mb{Y}}}+{\bar {\mb{N}}_1}^{(n)}$, this new problem is equivalent to the problem in which encoder 2, instead of $\left({\bar {\mb{Y}}},{\bar {\mb{N}}_1}^{(n)}\right)$, has access to the sum ${\bar {\mb{Y}}}+{\bar {\mb{N}}_1}^{(n)}.$ Let us denote this sum by $\bar {\mb{Y}}^{(n)}$, i.e.,
\begin{align*}
\bar {\mb{Y}}^{(n)} &\triangleq  {\bar {\mb{Y}}}+{\bar {\mb{N}}_1}^{(n)},
\end{align*}
which has a positive definite covariance matrix
\[
\mb{K}_{\bar {\mb{Y}}^{(n)}}=\mb{K}_{\bar {\mb{Y}}}+\mb{K}_{\bar {\mb{N}}_1^{(n)}}=\mat{\mb{\Sigma}_1 & \mb{0} \\ \mb{0} & \frac{1}{n}\mb{G}}.
\]
It follows that
\begin{align*}
\mathcal{R}\left(\mb{K}_{\bar {\mb{X}}}, \mb{K}_{\bar {\mb{Y}}^{(n)}}, \bar {\mb{D}}, \mu\right) \le \mathcal{R}\left(\mb{K}_{\bar {\mb{X}}}, \mb{K}_{\bar {\mb{Y}}}, \bar {\mb{D}}, \mu\right).
\end{align*}
Since this is true for all $n$ and $\mathcal{R}\left(\mb{K}_{\bar {\mb{X}}}, \mb{K}_{\bar {\mb{Y}}^{(n)}}, \bar {\mb{D}}, \mu\right)$ is nondecreasing in $n$, we obtain
\begin{align}
\lim_{n \to \infty} \mathcal{R}\left(\mb{K}_{\bar {\mb{X}}}, \mb{K}_{\bar {\mb{Y}}^{(n)}}, \bar {\mb{D}}, \mu\right) \le \mathcal{R}\left(\mb{K}_{\bar {\mb{X}}}, \mb{K}_{\bar {\mb{Y}}}, \bar {\mb{D}}, \mu\right). \label{eq:Extension1}
\end{align}
Since $\mb{K}_{\bar {\mb{X}}}, \mb{K}_{\bar {\mb{Y}}^{(n)}},$ and $\bar {\mb{D}}$ are positive definite, the conclusion of Theorem \ref{thm:MainThm} holds for this sequence of relaxed problems, i.e., for each $n$
\begin{align*}
\mathcal{R}\left(\mb{K}_{\bar {\mb{X}}}, \mb{K}_{\bar {\mb{Y}}^{(n)}}, \bar {\mb{D}}, \mu\right) = \mathcal{R}_G\left(\mb{K}_{\bar {\mb{X}}}, \mb{K}_{\bar {\mb{Y}}^{(n)}}, \bar {\mb{D}}, \mu\right).
\end{align*}
This and (\ref{eq:Extension1}) together imply that
\begin{align}
\lim_{n \to \infty} \mathcal{R}_G\left(\mb{K}_{\bar {\mb{X}}}, \mb{K}_{\bar {\mb{Y}}^{(n)}}, \bar {\mb{D}}, \mu\right) \le \mathcal{R}\left(\mb{K}_{\bar {\mb{X}}}, \mb{K}_{\bar {\mb{Y}}}, \bar {\mb{D}}, \mu\right). \label{eq:Extension2}
\end{align}
Now for each $n$, there exists $\left(\mb{U}^{(n)},\mb{V}^{(n)} \right)$ in $\mathcal{S}\left(\mb{K}_{\bar {\mb{X}}}, \mb{K}_{\bar {\mb{Y}}^{(n)}},\bar {\mb{D}}\right)$ such that
\begin{align}
\mathcal{R}_G\left(\mb{K}_{\bar {\mb{X}}}, \mb{K}_{\bar {\mb{Y}}^{(n)}}, \bar {\mb{D}}, \mu\right) = \mu I \left({\bar {\mb{X}}};\mb{U}^{(n)}|\mb{V}^{(n)} \right) + I \left(\bar {\mb{Y}}^{(n)};\mb{V}^{(n)} \right). \label{eq:Extension3}
\end{align}
Since ${\bar {\mb{X}}},\bar {\mb{Y}}^{(n)},\mb{U}^{(n)},$ and $\mb{V}^{(n)}$ are jointly Gaussian, we can without loss of generality parameterize them by positive semidefinite matrices $\mb{B}_1$ and $\mb{B}_2$ as in the definition $(P_{G2})$. These matrices lie in a compact set because they satisfy the KKT conditions that are continuous, and they are bounded as  $\mb{B}_1+\mb{B}_2 \prec \mb{K}_{\bar {\mb{X}}}$. Therefore, there exists a subsequence of $\mb{K}_{\bar {\mb{Y}}^{(n)}}$ along which $\left(\mb{U}^{(n)},\mb{V}^{(n)}\right)$ converges to $(\mb{U},\mb{V})$ in $\mathcal{S}\left(\mb{K}_{\bar {\mb{X}}}, \mb{K}_{\bar {\mb{Y}}},\bar {\mb{D}}\right)$. Since the right-hand-side of (\ref{eq:Extension3}) is continuous in $\left(\bar {\mb{Y}}^{(n)},\mb{U}^{(n)},\mb{V}^{(n)}\right)$, this implies
\begin{align}
\lim_{n \to \infty} \mathcal{R}_G\left(\mb{K}_{\bar {\mb{X}}}, \mb{K}_{\bar {\mb{Y}}^{(n)}}, \bar {\mb{D}}, \mu\right) &= \mu I \left({\bar {\mb{X}}};\mb{U}|\mb{V} \right) + I \left(\bar {\mb{Y}};\mb{V}\right) \nonumber\\
&\ge \mathcal{R}_G\left(\mb{K}_{\bar {\mb{X}}}, \mb{K}_{\bar {\mb{Y}}}, \bar {\mb{D}}, \mu\right). \label{eq:Extension4}
\end{align}
It now follows from (\ref{eq:Extension2}) and (\ref{eq:Extension4}) that
\begin{align*}
\mathcal{R}\left(\mb{K}_{\bar {\mb{X}}}, \mb{K}_{\bar {\mb{Y}}}, \bar {\mb{D}}, \mu\right) \ge \mathcal{R}_G\left(\mb{K}_{\bar {\mb{X}}}, \mb{K}_{\bar {\mb{Y}}}, \bar {\mb{D}}, \mu\right).
\end{align*}
This proves Theorem \ref{thm:Extension1}.
\end{proof}

We next use Theorem \ref{thm:Extension1} to prove our result for the most general case of the problem.
\begin{Thm} \label{thm:Extension2}
For any positive semidefinite $\mb{K_X}, \mb{K_Y},$ and $\mb{D}$, we have
\[
\mathcal{R}\left(\mb{K_X}, \mb{K_Y}, \mb{D}, \mu\right)=\mathcal{R}_G\left(\mb{K_X}, \mb{K_Y}, \mb{D}, \mu\right).
\]
\end{Thm}
\begin{proof}
Let us suppose that the rank of $\mb{K_X}$ is $p \le m$.  Since $\mb{K_X}$ is positive semidefinite, its eigen decomposition is
\[
\mb{K_X} = \mb{Q \Sigma}\mb{Q}^T,
\]
where $\mb{Q}$ is an orthogonal matrix and
\[
\mb{\Sigma} = \textrm{Diag} (\alpha_1,\dots,\alpha_p,0,\dots,0).
\]
Let us partition $\mb{Q} $ as
\[
\mb{Q} \triangleq [\mb{Q}_1, \mb{Q}_2 ],
\]
where $\mb{Q}_1$ is an $m \times p$ matrix. Since $\mb{Q}_2^T\mb{K_XQ_2} = \mb{0}$ and $\mb{X}=\mb{Y}+\mb{N},$ we have
\[
\mb{Q}_2^T\mb{K_YQ_2}=\mb{Q}_2^T\mb{K_NQ_2} = \mb{0},
\]
which implies that
\begin{align*}
\mb{Q}^T\mb{K_YQ} &= \mat{\mb{Q}_1^T \mb{K_Y} \mb{Q}_1 & \mb{0} \\ \mb{0} & \mb{0}} \hspace{0.1in} \textrm{and}\\
\mb{Q}^T\mb{K_NQ} &= \mat{\mb{Q}_1^T \mb{K_N} \mb{Q}_1 & \mb{0} \\ \mb{0} & \mb{0}}.
\end{align*}
Let us define
\begin{align*}
\mb{Q}^T \mb{D Q} &\triangleq \mat{\mb{E} & \mb{F}^T \\ \mb{F} & \mb{G}},
\end{align*}
where $\mb{E}$, $\mb{F},$ and $\mb{G}$ are submatrices of dimensions $p \times p$, $ (m-p) \times p,$ and $(m-p) \times (m-p)$, respectively. We need the following lemma.
\begin{Lem}  \label{lem:pseudo} \cite[Appendix A.5.5, p. 651]{Boyd} $\mb{Q}^T \mb{D Q} \succcurlyeq \mb{0}$ if and only if
\begin{align*}
\mb{G} &\succcurlyeq \mb{0}, \\
\mb{E}-\mb{F}^T\mb{G}^{+}\mb{F} &\succcurlyeq \mb{0}, \hspace{0.05in} \textrm{and}\\
(\mb{I}_{m-p}-\mb{GG}^{+})\mb{F} &= \mb{0},
\end{align*}
where $\mb{G}^{+}$ is the pseudo-inverse or Moore-Penrose inverse of $\mb{G}$ \cite[Appendix A.5.4, p. 649]{Boyd}.
\end{Lem}

Let
\begin{align*}
\mb{T} &\triangleq \mat{\mb{T}_1 \\ \mb{T}_2}\triangleq \mat{\mb{I}_p & -\mb{F}^T \mb{G}^{+} \\ \mb{0} & \mb{I}_{m-p}} \mb{Q}^T,
\end{align*}
where $\mb{T}_1$ is a $p \times m$ matrix. Using this, we obtain a transformed problem in which the transformed sources are
\begin{align*}
\bar {\mb{X}} &\triangleq \mat{\mb{X}_1\\  \mb{X}_2} \triangleq \mat{\mb{T}_1 \mb{X}\\  \mb{T}_2 \mb{X}}= \mb{T} \mb{X} \hspace{0.1in} \textrm{and}\\
\bar {\mb{Y}} &\triangleq \mat{\mb{Y}_1\\  \mb{Y}_2} \triangleq \mat{\mb{T}_1 \mb{Y}\\  \mb{T}_2 \mb{Y}}= \mb{T} \mb{Y}.
\end{align*}
Using Lemma \ref{lem:pseudo}, we obtain the transformed distortion matrix
\begin{align}
\bar {\mb{D}} &\triangleq\mb{T} \mb{D} \mb{T}^T \nonumber\\
&= \mat{\mb{I}_p & -\mb{F}^T \mb{G}^{+} \\ \mb{0} & \mb{I}_{m-p}} \mb{Q}^T \mb{D} \mb{Q} \mat{\mb{I}_p & \mb{0} \\-\mb{G}^{+}\mb{F} & \mb{I}_{m-p}} \nonumber\\
&= \mat{\mb{I}_p & -\mb{F}^T \mb{G}^{+} \\ \mb{0} & \mb{I}_{m-p}}  \mat{\mb{E} & \mb{F}^T \\ \mb{F} & \mb{G}} \mat{\mb{I}_p & \mb{0} \\-\mb{G}^{+}\mb{F} & \mb{I}_{m-p}} \nonumber\\
&= \mat{\mb{E}-\mb{F}^T \mb{G}^{+}\mb{F} & \mb{0} \\\mb{0} & \mb{G}} \nonumber\\
&= \mat{ \mb{D}_1 & \mb{0} \\\mb{0} & \mb{D}_2}, \label{eq:Extension5}
\end{align}
where
\begin{align*}
\mb{D}_1&\triangleq \mb{E}-\mb{F}^T \mb{G}^{+}\mb{F} \hspace{0.1in} \textrm{and}\\
\mb{D}_2&\triangleq\mb{G}.
\end{align*}
The covariance matrix of the transformed source $\bar {\mb{X}}$ is
\begin{align*}
\mb{K}_{\bar {\mb{X}}} &= \mb{T} \mb{K}_{\mb{X}} \mb{T}^T \nonumber\\
&= \mat{\mb{I}_p & -\mb{F}^T \mb{G}^{+} \\ \mb{0} & \mb{I}_{m-p}} \mb{Q}^T\mb{K_XQ} \mat{\mb{I}_p & \mb{0} \\-\mb{G}^{+}\mb{F} & \mb{I}_{m-p}} \nonumber\\
&= \mat{\mb{I}_p & -\mb{F}^T \mb{G}^{+} \\ \mb{0} & \mb{I}_{m-p}} \mat{\mb{\Sigma}_1 & \mb{0} \\\mb{0} & \mb{0}} \mat{\mb{I}_p & \mb{0} \\-\mb{G}^{+}\mb{F} & \mb{I}_{m-p}} \nonumber\\
&= \mat{\mb{\Sigma}_1 & \mb{0} \\\mb{0} & \mb{0}},
\end{align*}
where
\begin{align*}
\mb{\Sigma}_1 &\triangleq \textrm{Diag} (\alpha_1,\dots,\alpha_p),
\end{align*}
and the covariance matrix of the transformed source $\bar {\mb{Y}}$ is
\begin{align*}
\mb{K}_{\bar {\mb{Y}}} &= \mb{T} \mb{K}_{\mb{Y}} \mb{T}^T \nonumber\\
&= \mat{\mb{I}_p & -\mb{F}^T \mb{G}^{+} \\ \mb{0} & \mb{I}_{m-p}} \mb{Q}^T\mb{K_YQ} \mat{\mb{I}_p & \mb{0} \\-\mb{G}^{+}\mb{F} & \mb{I}_{m-p}} \nonumber\\
&= \mat{\mb{I}_p & -\mb{F}^T \mb{G}^{+} \\ \mb{0} & \mb{I}_{m-p}} \mat{\mb{Q}^T_1\mb{K_Y} \mb{Q}_1 & \mb{0} \\\mb{0} & \mb{0}} \mat{\mb{I}_p & \mb{0} \\-\mb{G}^{+}\mb{F} & \mb{I}_{m-p}} \nonumber\\
&= \mat{\mb{Q}^T_1\mb{K_Y} \mb{Q}_1 & \mb{0} \\\mb{0} & \mb{0}}.
\end{align*}
It follows that $\mb{X}_2$ and $\mb{Y}_2$ are deterministic, i.e.,
\[
\mb{X}_2 = \mb{Y}_2 = \mb{0}.
\]
Since $\mb{T}$ is invertible, the distortion constraint is equivalent to
\begin{align}
\mb{T} \mb{D} \mb{T}^T&\succcurlyeq \frac{1}{n} \sum_{i=1}^n E\left [ \left (\bar {\mb{X}}^n(i) - \hat {\bar{\mb{X}}}^n(i) \right  ) \left (\bar {\mb{X}}^n(i) - \hat {\bar{\mb{X}}}^n(i) \right )^T \right ] \nonumber\\
&=\frac{1}{n} \sum_{i=1}^n E\left [ \mat{{\mb{X}}^n_1(i) - \hat {\mb{X}}^n_1(i) \\ \mb{0}} \mat{{\mb{X}}^n_1(i) - \hat {\mb{X}}^n_1(i) \\ \mb{0}}^T \right] \nonumber\\
&=\mat{\frac{1}{n} \sum_{i=1}^n E\left [ \left ({\mb{X}}^n_1(i) - \hat {\mb{X}}^n_1(i) \right) \left ({\mb{X}}^n_1(i) - \hat {\mb{X}}^n_1(i) \right)^T \right] & \mb{0} \\ \mb{0} & \mb{0}}.\label{eq:Extension6}
\end{align}
Since $\mb{D}_1$ and $\mb{D}_2$ are positive semidefinite from Lemma \ref{lem:pseudo}, (\ref{eq:Extension5}) and (\ref{eq:Extension6}) imply that the distortion constraint is equivalent to
\[
\mb{D}_1 \succcurlyeq \frac{1}{n} \sum_{i=1}^n E\left [ \left ({\mb{X}}^n_1(i) - \hat {\mb{X}}^n_1(i) \right) \left ({\mb{X}}^n_1(i) - \hat {\mb{X}}^n_1(i) \right)^T \right].
\]

Since $\mb{T}$ is invertible, the above transformation is information lossless, and hence the transformed problem is equivalent to the original problem. Moreover, the transformed problem is effectively $p$-dimensional with the sources $\mb{X}_1$ and $\mb{Y}_1$, and the distortion matrix $\mb{D}_1$ such that
\begin{align*}
\mb{K}_{\mb{X}_1} &=\mb{\Sigma}_1 \succ \mb{0} \hspace{0.1in} \textrm{and}\\
\mb{X}_1&=\mb{Y}_1+\mb{N}_1,
\end{align*}
where $\mb{N}_1 \triangleq \mb{T}_1\mb{N}$. We therefore have that
\begin{align}
\mathcal{R}\left(\mb{K_X}, \mb{K_Y}, \mb{D}, \mu\right)&=\mathcal{R}\left(\mb{K}_{\mb{X}_1}, \mb{K}_{\mb{Y}_1}, \mb{D}_1, \mu\right) \hspace{0.1in} \textrm{and} \label{eq:Extension7}\\
\mathcal{R}_G\left(\mb{K_X}, \mb{K_Y}, \mb{D}, \mu\right)&=\mathcal{R}_G\left(\mb{K}_{\mb{X}_1}, \mb{K}_{\mb{Y}_1}, \mb{D}_1, \mu\right). \label{eq:Extension77}
\end{align}
Since $\mb{K}_{\mb{X}_1}$ is positive definite, if $\mb{D}_1$ is
singular, then the right-hand side of (\ref{eq:Extension7}) and
(\ref{eq:Extension77}) are both infinite, so the conclusion
trivially holds.
Otherwise, we have that
%Since $\mb{K}_{\mb{X}_1}$ is positive definite, $\mb{D}_1$ must also be positive definite because otherwise the rate region is
%\[
%\left\{ (R_1, R_2) : R_1 = +\infty \right \},
%\]
%i.e., encoder 1 requires infinite rate. Now
$\mb{K}_{\mb{X}_1}$ and ${\mb{D}}_1$ are positive definite and $\mb{K}_{{\mb{Y}}_1}$ is positive semidefinite. In that case Theorem \ref{thm:Extension1}
implies that
\begin{align*}
\mathcal{R}\left(\mb{K}_{ {\mb{X}}_1} , \mb{K}_{ {\mb{Y}}_1}, {\mb{D}}_1, \mu\right) = \mathcal{R}_G\left(\mb{K}_{ {\mb{X}}_1} , \mb{K}_{ {\mb{Y}}_1}, {\mb{D}}_1, \mu\right).
\end{align*}
This together with (\ref{eq:Extension7}) and (\ref{eq:Extension77}) establishes the desired equality
\begin{align*}
\mathcal{R}\left(\mb{K}_{{\mb{X}}} , \mb{K}_{{\mb{Y}}},{\mb{D}}, \mu\right) = \mathcal{R}_G\left(\mb{K}_{{\mb{X}}} , \mb{K}_{{\mb{Y}}},{\mb{D}}, \mu\right).
\end{align*}
Theorem \ref{thm:Extension2} is thus proved.
\end{proof}

\hspace{0.1in}\\
\Large
{\textbf{Appendix A: \hspace{0.1in} \textbf{Proof of Equality in Lemma \ref{lem:BTRegion}} }\newline
\normalsize \\
Suppose $\mu$ is in $[0, 1]$. Then for any $(\mb{U,V})$ in $\mathcal{S}$, we have
\begin{align}
\mu I(\mb{X};\mb{U|V})+I(\mb{Y};\mb{V}) &= \mu I(\mb{X};\mb{U,V})-\mu I(\mb{X},\mb{V})+I(\mb{Y};\mb{V})\nonumber\\
&= \mu I(\mb{X};\mb{U})+\mu I(\mb{X};\mb{V}|\mb{U})+\mu [I(\mb{Y},\mb{V})- I(\mb{X};\mb{V})] + (1-\mu) I(\mb{Y};\mb{V})\nonumber\\
&\ge \mu I(\mb{X};\mb{U}) \label{eq:AppenA1}\\
&= \frac{\mu}{2} \log \frac{|\mb{K_X}|}{|\mb{K}_{\mb{X}|\mb{U}}|}\nonumber
\end{align}
where (\ref{eq:AppenA1}) follows because of the facts that $$I(\mb{Y};\mb{V}) \ge 0$$ and $$I(\mb{X};\mb{V}|\mb{U}) \ge 0,$$ and we have
\[
I(\mb{Y},\mb{V})- I(\mb{X};\mb{V}) \ge 0
\]
because of the data processing inequality \cite[Theorem 2.8.1]{Cover} and the Markov chain $\mb{X}\leftrightarrow \mb{Y}\leftrightarrow \mb{V}.$ The inequality (\ref{eq:AppenA1}) is achieved by any $(\mb{U,V})$ in $\mathcal{S}$ such that $\mb{V}$ is independent of $(\mb{X},\mb{Y},\mb{U})$, and the conditional covariance of $\mb{X}$ given $(\mb{U,V})$ satisfies
\begin{align*}
\mb{0} \preccurlyeq \mb{K}_{\mb{X}|\mb{U,V}}=\mb{K}_{\mb{X}|\mb{U}} \preccurlyeq \mb{D}.
\end{align*}
Since conditioning reduces covariance in a positive semidefinite sense, we have an additional constraint
$$\mb{K}_{\mb{X}|\mb{U}} \preccurlyeq \mb{K_X}.$$
We therefore have the following
\begin{align*}
\mathcal{R}_G(\mb{D}, \mu) &= \min_{(R_1,R_2) \in \mathcal{R}_G(\mb{D})} \mu R_1+R_2 \\
&=\min_{(\mb{U,V}) \in \mathcal{S}} \mu I(\mb{X};\mb{U|V}) +I(\mb{Y};\mb{V}) \\
&=\hspace{0.25in}\min_{\mb{K}_{\mb{X}|\mb{U}}} \hspace {0.2in} \frac{\mu}{2} \log \frac{ \left |\mb{K_X} \right |}{ \left|\mb{K}_{\mb{X}|\mb{U}}\right|} \nonumber\\
&\hspace{0.3in}\textrm{subject to} \hspace{0.12in} \mb{K}_{\mb{X}} \succcurlyeq \mb{K}_{\mb{X}|\mb{U}}  \succcurlyeq \mb{0}\hspace{0.1in} \textrm{and}\nonumber\\
&\hspace{0.95in}\mb{D} \succcurlyeq \mb{K}_{\mb{X}|\mb{U}}\\
&=v\left (P_{pt-pt} \right).
\end{align*}
Suppose now that $\mu > 1$. Then any $(\mb{U,V})$ in $\mathcal{S}$ can be characterized by positive semidefinite conditional covariance matrices $\mb{K}_{\mb{Y}|\mb{V}}$ and $\mb{K}_{\mb{X}|\mb{U,V}}$ such that
\begin{align*}
&\mb{K}_{\mb{Y}} \succcurlyeq \mb{K}_{\mb{Y}|\mb{V}}  \succcurlyeq \mb{0},\\
&\mb{K}_{\mb{Y}|\mb{V}}+\mb{K_N} \succcurlyeq \mb{K}_{\mb{X}|\mb{U,V}} \succcurlyeq \mb{0},\\
&\mb{D} \succcurlyeq \mb{K}_{\mb{X}|\mb{U,V}},
\end{align*}
and
\begin{align*}
I(\mb{X};\mb{U|V})&=\frac{1}{2} \log \frac{|\mb{K}_{\mb{Y}|\mb{V}}+\mb{K_N}|}{|\mb{K}_{\mb{X}|\mb{U,V}}|}, \\
I(\mb{Y};\mb{V})&=\frac{1}{2} \log \frac{|\mb{K}_{\mb{Y}} |}{|\mb{K}_{\mb{Y}|\mb{V}}|}.
\end{align*}
%Conditions (3) and (4) imply that there exist two positive semidefinite matrices $\mb{B}_1$ and $\mb{B}_2$ such that
%\begin{align*}
%\mb{K}_{\mb{Y}|V} &= \mb{K_Y}-\mb{B}_2,
%\end{align*}
%and
%\begin{align*}
%\mb{K}_{\mb{X}|\mb{U,V}} &= \mb{K}_{\mb{Y}|V}+\mb{K_N} - \mb{B}_1 \\
%&= \mb{K_Y}-\mb{B}_2+\mb{K_N} - \mb{B}_1\\
%&=\mb{K_X} - \mb{B}_1 - \mb{B}_2.
%\end{align*}
%The condition (\ref{eq:KKT1}) then imposes a constraint that
%\begin{align*}
%\mb{K_X} - \mb{B}_1 - \mb{B}_2 &\preccurlyeq \mb{D}.
%\end{align*}
In this case, we have
\begin{align*}
\mathcal{R}_G(\mb{D}, \mu) &= \min_{(R_1,R_2) \in \mathcal{R}_G(\mb{D})} \mu R_1+R_2 \\
&=\min_{(\mb{U,V}) \in \mathcal{S}} \mu I(\mb{X};\mb{U|V}) +I(\mb{Y};\mb{V}) \\
&=\min_{\mb{K}_{\mb{Y}|\mb{V}},\mb{K}_{\mb{X}|\mb{U,V}}} \hspace {0.2in} \frac{\mu}{2} \log \frac{ \left |\mb{K}_{\mb{Y}|\mb{V}}+\mb{K_N} \right |}{ \left|\mb{K}_{\mb{X}|\mb{U,V}}\right|} + \frac{1}{2} \log \frac{ \left |\mb{K_Y} \right |}{ \left|\mb{K}_{\mb{Y}|\mb{V}}\right|}\nonumber\\
&\hspace{0.3in}\textrm{subject to} \hspace{0.3in} \mb{K}_{\mb{Y}} \succcurlyeq \mb{K}_{\mb{Y}|\mb{V}}  \succcurlyeq \mb{0},\\
&\hspace{1.13in}\mb{K}_{\mb{Y}|\mb{V}}+\mb{K_N} \succcurlyeq \mb{K}_{\mb{X}|\mb{U,V}} \succcurlyeq \mb{0},\hspace{0.1in} \textrm{and}\nonumber\\
&\hspace{1.13in}\mb{D} \succcurlyeq \mb{K}_{\mb{X}|\mb{U,V}}\\
&=v\left (P_{G1} \right).
\end{align*}
\\
\Large
{\textbf{Appendix B:\hspace{0.1in} \textbf{Proof of Lemma \ref{lem:ExistenceKKT}} }\newline
\normalsize \\
We will be using several results and terms from Bertsekas \emph{et al.} \cite{Bertsekas}. The book contains all of the background that these results need. The proof of the lemma is partially similar to that of Lemma 5 in \cite{Wein}. Let us first introduce some notation used in the proof. We use  vec$(\mb{A}_1,\mb{A}_2)$ to denote the column vector created by the concatenation of the columns of $m \times m$ matrices $\mb{A}_1$ and $\mb{A}_2.$ If  $\mb{a} = \textrm{vec}(\mb{A}_1,\mb{A}_2)$, then we use the notation mat$(\mb{a})$ to denote the inverse operation to get back the pair $(\mb{A}_1,\mb{A}_2),$ i.e.,
$$\textrm{mat}(\mb{a}) = (\mb{A}_1,\mb{A}_2).$$
The set of all column vectors created by the concatenation of the columns of $m \times m$ symmetric matrices $\mb{A}_1$ and $\mb{A}_2$ is denoted by $\mathcal{A}$, i.e.,
$$\mathcal{A} \triangleq \{ \textrm{vec}(\mb{A}_1,\mb{A}_2) : \mb{A}_i = \mb{A}_i^T \hspace{0.05in} \textrm{for all} \hspace{0.05in} i  \in \{1,2\}\}.$$
ri$(\mathcal{B})$ is used to denote the relative interior of the set $\mathcal{B}$. The sum of the two vector sets $\mathcal{V}_1$ and $\mathcal{V}_1$ is denoted by $\mathcal{V}_1+\mathcal{V}_2$ and is defined as $$\mathcal{V}_1+\mathcal{V}_2 \triangleq \{\mb{v}_1+\mb{v}_2 :  \mb{v}_i \in \mathcal{V}_i \hspace{0.05in} \textrm{for all}  \hspace{0.05in} i \in \{1,2\} \}.$$
We also need the following facts from linear algebra.
\begin{Lem} \label{lem:LinearAlgFact}
\begin{enumerate}
\item[(a)] If $\mb{E}$ is an $m \times n$  matrix and $\mb{F}$ is an $n \times m$ matrix, then $\textrm{\emph{Tr}} (\mb{EF})=\textrm{\emph{Tr}} (\mb{FE})$.
\item[(b)] If $\mb{E}$ and $\mb{F}$ are positive semidefinite, then $\mb{EF}=\mb{0}$ if and only if $\textrm{\emph{Tr}} (\mb{EF})={0}$.
\end{enumerate}
\end{Lem}
\begin{proof}
Part (a) immediately follows from the definition of $\textrm{Tr}(\cdot)$ function. Part (b) can be proved using the eigen decompositions of $\mb{E}$ and $\mb{F}$.
\end{proof}

We can express the problem $(P_{G2})$ as
\begin{align*}
\min_{\mb{b}} \hspace {0.2in} &h(\mb{b})\nonumber\\
\hspace{0.3in}\textrm{subject to} \hspace{0.2in} &\mb{b} \in \mathcal{B},
\end{align*}
where $\mb{b} \triangleq \textrm{vec}(\mb{B}_1,\mb{B}_2)$, $$h(\mb{b}) \triangleq \frac{\mu}{2} \log \frac{ \left |\mb{K_X}- \mb{B}_2 \right |}{ \left|\mb{K_X}- \mb{B_1}- \mb{B}_2\right|} + \frac{1}{2} \log \frac{ \left |\mb{K_Y} \right |}{ \left|\mb{K_Y}-\mb{B}_2\right|},$$ and
the feasible set $\mathcal{B}$ is written as
$$\mathcal{B} \triangleq \mathcal{B}_1 \cap \mathcal{B}_2 \cap \mathcal{B}_{12},$$ where for $i \in \{1,2\}$
$$\mathcal{B}_i \triangleq \{\textrm{vec}(\mb{B}_1,\mb{B}_2): \mb{B}_i \succcurlyeq \mb{0}\} \cap \mathcal{A}$$
and
$$\mathcal{B}_{12} \triangleq \{\textrm{vec}(\mb{B}_1,\mb{B}_2): \mb{B}_1+ \mb{B}_2 \succcurlyeq \mb{K_X}-\mb{D}\} \cap \mathcal{A}.$$ Since $h(\cdot)$ is continuously differentiable, it follows from \cite[Proposition 4.7.1, p. 255]{Bertsekas} that any local minima $\mb{b}^*$ must satisfy
\begin{align}
-\nabla h(\mb{b}^*) \in T_{\mathcal{B}} (\mb{b}^*)^*, \label{eq:nabla}
\end{align}
where $\nabla h(\mb{b}^*)$ is the gradient of $h(\cdot)$ at $\mb{b}^*,$ and $T_{\mathcal{B}} (\mb{b}^*)^*$ is the polar cone of the tangent cone $T_{\mathcal{B}} (\mb{b}^*)$ of $\mathcal{B}$ at $\mb{b}^*$. Now since $\mathcal{B}_i$ for all $i \in \{1,2\}$ and $\mathcal{B}_{12}$ are nonempty convex sets and $\textrm{ri}(\mathcal{B}_1) \cap \textrm{ri}(\mathcal{B}_2) \cap \textrm{ri}(\mathcal{B}_{12})$ is nonempty, it follows from \cite[Problem 4.23, p. 267]{Bertsekas} and \cite[Proposition 4.6.3, p. 254]{Bertsekas} that
\begin{align}
T_{\mathcal{B}} (\mb{b}^*)^* = T_{\mathcal{B}_1} (\mb{b}^*)^*+T_{\mathcal{B}_2} (\mb{b}^*)^*+T_{\mathcal{B}_{12}} (\mb{b}^*)^*. \label{eq:Tangent}
\end{align}
We next show that
\begin{align}
-\nabla h(\mb{b}^*) \in T_{\mathcal{B}_1} (\mb{b}^*)^* \cap \mathcal{A} + T_{\mathcal{B}_2} (\mb{b}^*)^* \cap \mathcal{A} + T_{\mathcal{B}_{12}} (\mb{b}^*)^* \cap \mathcal{A}. \label{eq:nabla1}
\end{align}
Note that $-\nabla h(\mb{b}^*)$ is a column concatenation of two $m \times m$ symmetric matrices. This together with (\ref{eq:nabla}) and (\ref{eq:Tangent}) yields
\begin{align}
-\nabla h(\mb{b}^*) = \mb{z}_1 + \mb{z}_2 + \mb{z}_{12} \in \mathcal{A}, \label{eq:nabla50}
\end{align}
where for $i \in \{1,2\}$
\begin{align*}
\mb{z}_i &\in  T_{\mathcal{B}_i} (\mb{b}^*)^*  \hspace{0.05in} \textrm{and} \\
\mb{z}_{12} &\in  T_{\mathcal{B}_{12}} (\mb{b}^*)^*.
\end{align*}
Let us now define
\begin{align*}
(\mb{K}_i,\mb{L}_i) &\triangleq \textrm{mat}(\mb{z}_i), \forall i \in \{1,2\} \hspace{0.05in}\textrm{and}\\
(\mb{K}_{12},\mb{L}_{12}) &\triangleq \textrm{mat}(\mb{z}_{12}).
\end{align*}
Using this, we define
\begin{align*}
\bar {\mb{z}}_i &\triangleq \textrm{vec}\left(\frac{1}{2}\left(\mb{K}_i+\mb{K}_i^T\right), \frac{1}{2}\left(\mb{L}_i+\mb{L}_i^T\right)\right), \forall i \in \{1,2\} \hspace{0.05in}\textrm{and}\\
\bar {\mb{z}}_{12} &\triangleq \textrm{vec}\left(\frac{1}{2}\left(\mb{K}_{12}+\mb{K}_{12}^T\right), \frac{1}{2}\left(\mb{L}_{12}+\mb{L}_{12}^T\right)\right).
\end{align*}
Since $\mathcal{B}_1$ is a nonempty convex set, it follows from \cite[Proposition 4.6.3, p. 254]{Bertsekas} that
\begin{align}
\mb{z}_1^T (\mb{b}-\mb{b}^*) \le 0, \hspace{0.1in} \forall \mb{b} \in \mathcal{B}_1. \label{eq:nabla52}
\end{align}
Consider any $\mb{b} \in \mathcal{B}_1$. Let
\[
(\mb{E}_1,\mb{F}_1) \triangleq \textrm{mat}(\mb{b}-\mb{b}^*).
\]
We now obtain
\begin{align}
\bar {\mb{z}}_1^T (\mb{b}-\mb{b}^*) &= \frac{1}{2} \textrm{Tr} \left(\left(\mb{K}_1+\mb{K}_1^T\right)\mb{E}_1 \right) + \frac{1}{2}\textrm{Tr} \left(\left(\mb{L}_1+\mb{L}_1^T\right)\mb{F}_1 \right) \nonumber\\
&= \textrm{Tr} \left(\mb{K}_1\mb{E}_1 \right) + \textrm{Tr} \left(\mb{L}_1\mb{F}_1 \right) \label{eq:nabla53}\\
&= {\mb{z}}_1^T (\mb{b}-\mb{b}^*) \nonumber\\
&\le 0, \label{eq:nabla54}
\end{align}
where
\begin{enumerate}
\item[(\ref{eq:nabla53})] follows because $\mb{E}_1$ and $\mb{F}_1$ are symmetric, and
\item[(\ref{eq:nabla54})] follows from (\ref{eq:nabla52}).
\end{enumerate}
By definition, $\bar {\mb{z}}_1 \in \mathcal{A}$. This and (\ref{eq:nabla54}) imply that
\begin{align}
\bar {\mb{z}}_1 \in T_{\mathcal{B}_1}(\mb{b}^*)^* \cap \mathcal{A}. \label{eq:nabla57}
\end{align}
We can similarly show that
\begin{align}
\bar {\mb{z}}_2 &\in T_{\mathcal{B}_2}(\mb{b}^*)^* \cap \mathcal{A} \hspace{0.05in} \textrm{and} \label{eq:nabla58}\\
\bar {\mb{z}}_{12} &\in T_{\mathcal{B}_{12}}(\mb{b}^*)^* \cap \mathcal{A}. \label{eq:nabla59}
\end{align}
Now
\begin{align}
&\bar {\mb{z}}_1 + \bar {\mb{z}}_2 + \bar {\mb{z}}_{12} \nonumber\\
&= \textrm{vec}\left(\frac{1}{2}\left(\mb{K}_1+\mb{K}_2+\mb{K}_{12}+\mb{K}_1^T+\mb{K}_2^T+\mb{K}_{12}^T\right), \frac{1}{2}\left(\mb{L}_1+\mb{L}_2+\mb{L}_{12}+\mb{L}_1^T+\mb{L}_2^T+\mb{L}_{12}^T\right)\right) \nonumber \\
&= \textrm{vec}\left( \left(\mb{K}_1+\mb{K}_2+\mb{K}_{12} \right), \left(\mb{L}_1+\mb{L}_2+\mb{L}_{12}\right)\right) \label{eq:nabla55} \\
&= {\mb{z}}_1 +  {\mb{z}}_2 + {\mb{z}}_{12} \nonumber \\
&= -\nabla h(\mb{b}^*), \label{eq:nabla56}
\end{align}
where
\begin{enumerate}
\item[(\ref{eq:nabla55})] follows because $\mb{K}_1+\mb{K}_2+\mb{K}_{12}$ and $\mb{L}_1+\mb{L}_2+\mb{L}_{12}$ are symmetric from (\ref{eq:nabla50}), and
\item[(\ref{eq:nabla56})] follows from the equality in (\ref{eq:nabla50}).
\end{enumerate}
This together with (\ref{eq:nabla57}) -- (\ref{eq:nabla59}) implies (\ref{eq:nabla1}). 

We now proceed to characterize the right-hand side of (\ref{eq:nabla1}). Consider any $\mb{z} \in T_{\mathcal{B}_1} (\mb{b}^*)^* \cap \mathcal{A}.$  It again follows from \cite[Proposition 4.6.3, p. 254]{Bertsekas} that
\begin{align}
\mb{z}^T (\mb{b}-\mb{b}^*) \le 0, \hspace{0.1in} \forall \mb{b} \in \mathcal{B}_1. \label{eq:Tangent1}
\end{align}
Let us define
\begin{align*}
(\mb{M}_1,\mb{M}_2) &\triangleq \textrm{mat}(\mb{z}), \\
(\mb{B}_1,\mb{B}_2) &\triangleq \textrm{mat}(\mb{b}), \hspace{0.05in}\textrm{and}\\
(\mb{B}_1^*,\mb{B}_2^*) &\triangleq \textrm{mat}(\mb{b}^*).
\end{align*}
Then (\ref{eq:Tangent1}) can be re-written as
\begin{align}
\sum_{i=1}^2\textrm{Tr}(\mb{M}_i(\mb{B}_i-\mb{B}_i^*)) \le 0, \hspace{0.1in} \forall \hspace{0.02in}  \textrm{vec}(\mb{B}_1,\mb{B}_2) \in \mathcal{B}_1. \label{eq:Tangent2}
\end{align}
We first show that $\mb{M}_2 = \mb{0}$. Let us pick $(\mb{B}_1,\mb{B}_2) = (\mb{B}_1^*,\mb{B}_2^*+\mb{M}_2)$. This means that
$$\textrm{Tr}(\mb{M}_2\mb{M}_2) \le 0,$$
which implies that $\mb{M}_2 = \mb{0}$ because $\mb{M}_2$ is symmetric. We next prove that $\mb{M}_1$ is negative semidefinite. Suppose there exists $\mb{w} \neq \mb{0}$ such that $\mb{w}^T \mb{M}_1 \mb{w} > 0.$ We then have
\[
0 < \mb{w}^T \mb{M}_1 \mb{w}  = \textrm{Tr}(\mb{w}^T \mb{M}_1 \mb{w}) = \textrm{Tr}( \mb{M}_1 \mb{w} \mb{w}^T),
\]
where the last equality follows from Lemma \ref{lem:LinearAlgFact}(a). But this contradicts (\ref{eq:Tangent2}) because $\textrm{vec}(\mb{B}_1^*+\mb{w} \mb{w}^T ,\mb{B}_2^*) \in \mathcal{B}_1,$ and hence $\mb{M}_1\preccurlyeq \mb{0}.$ We finally show that $\mb{M}_1\mb{B}_1^* = \mb{0}.$ Let $(\mb{B}_1,\mb{B}_2) = (\alpha \mb{B}_1^*,\mb{B}_2^*),$ where $\alpha > 1.$ Then (\ref{eq:Tangent2}) implies that
\[
\textrm{Tr}(\mb{M}_1\mb{B}_1^*) \le 0.
\]
Likewise, on picking $0 < \alpha < 1$, we obtain
\[
\textrm{Tr}(\mb{M}_1\mb{B}_1^*) \ge 0.
\]
Both together establish
\[
\textrm{Tr}(\mb{M}_1\mb{B}_1^*) = 0,
\]
which together with Lemma \ref{lem:LinearAlgFact}(b) implies that
\[
\mb{M}_1\mb{B}_1^*  = \mb{0}
\]
because $-\mb{M}_1$ and $\mb{B}_1^*$ are positive semidefinite. We therefore have that
\begin{align}
T_{\mathcal{B}_1} (\mb{b}^*)^* \cap \mathcal{A} \subseteq \{ \textrm{vec}(\mb{M}_1, \mb{0}) : \mb{M}_1 \preccurlyeq \mb{0} \hspace{0.05in}\textrm{and} \hspace{0.05in} \mb{M}_1\mb{B}_1^* = \mb{0}\}. \label{eq:Tangent4}
\end{align}
Similarly, we can show that
\begin{align}
T_{\mathcal{B}_2} (\mb{b}^*)^* \cap \mathcal{A} \subseteq \{ \textrm{vec}(\mb{0}, \mb{M}_2) : \mb{M}_2 \preccurlyeq \mb{0} \hspace{0.05in}\textrm{and} \hspace{0.05in} \mb{M}_2\mb{B}_2^* = \mb{0}\}. \label{eq:Tangent5}
\end{align}

Consider any $\mb{z} \in T_{\mathcal{B}_{12}} (\mb{b}^*)^* \cap \mathcal{A}.$ As before, we obtain
\begin{align}
\sum_{i=1}^2\textrm{Tr}(\mb{\Lambda}_i(\mb{B}_i-\mb{B}_i^*)) \le 0, \hspace{0.1in} \forall \hspace{0.02in}  \textrm{vec}(\mb{B}_1,\mb{B}_2) \in \mathcal{B}_{12}, \label{eq:Tangent3}
\end{align}
where
\begin{align*}
(\mb{\Lambda}_1,\mb{\Lambda}_2) &\triangleq \textrm{mat}(\mb{z}).
\end{align*}
On picking $(\mb{B}_1,\mb{B}_2) = (\mb{B}_1^*+\mb{\Lambda}_1,\mb{B}_2^*-\mb{\Lambda}_1),$ (\ref{eq:Tangent3}) yields
\[
\textrm{Tr}(\mb{\Lambda}_1\mb{\Lambda}_1) - \textrm{Tr}(\mb{\Lambda}_2\mb{\Lambda}_1) \le 0.
\]
Similarly, picking $(\mb{B}_1,\mb{B}_2) = (\mb{B}_1^*-\mb{\Lambda}_2,\mb{B}_2^*+\mb{\Lambda}_2)$ gives
\[
\textrm{Tr}(\mb{\Lambda}_2\mb{\Lambda}_2) - \textrm{Tr}(\mb{\Lambda}_1\mb{\Lambda}_2) \le 0.
\]
Both together imply that
\[
\textrm{Tr}((\mb{\Lambda}_1-\mb{\Lambda}_2) (\mb{\Lambda}_1-\mb{\Lambda}_2) ) \le 0,
\]
and therefore
\[
\mb{\Lambda}_1-\mb{\Lambda}_2 = \mb{0},
\]
because $\mb{\Lambda}_1$ and $\mb{\Lambda}_2$ are symmetric. Let us denote $\mb{\Lambda}_1$ and $\mb{\Lambda}_2$ by $\mb{\Lambda}$. As before, we can show that $\mb{\Lambda} \preccurlyeq \mb{0}$. We next prove that
\[
\textrm{Tr}(\mb{\Lambda}(\mb{B}_1^*+\mb{B}_2^*-\mb{K_X}+\mb{D})) = 0.
\]
Observe that $(\mb{B}_1,\mb{B}_2) = \bigr(\alpha (\mb{B}_1^*+\mb{B}_2^*-\mb{K_X}+\mb{D})+\mb{K_X}-\mb{D}-\mb{B}_2^*, \mb{B}_2^*\bigr),$ where $\alpha > 0,$ is a valid choice of $(\mb{B}_1,\mb{B}_2)$ in (\ref{eq:Tangent3}). For $\alpha > 1$, this implies
\[
\textrm{Tr}(\mb{\Lambda}(\mb{B}_1^*+\mb{B}_2^*-\mb{K_X}+\mb{D})) \le 0,
\]
and for $0 <\alpha < 1$, it gives
\[
\textrm{Tr}(\mb{\Lambda}(\mb{B}_1^*+\mb{B}_2^*-\mb{K_X}+\mb{D})) \ge 0.
\]
Therefore
\[
\textrm{Tr}(\mb{\Lambda}(\mb{B}_1^*+\mb{B}_2^*-\mb{K_X}+\mb{D})) = 0.
\]
This and Lemma \ref{lem:LinearAlgFact}(b) imply that
\[
\mb{\Lambda}(\mb{B}_1^*+\mb{B}_2^*-\mb{K_X}+\mb{D}) = \mb{0}.
\]
We thus have that
\begin{align}
T_{\mathcal{B}_{12}} (\mb{b}^*)^* \cap \mathcal{A} \subseteq \{ \textrm{vec}(\mb{\Lambda}, \mb{\Lambda}) | \mb{\Lambda} \preccurlyeq \mb{0} \hspace{0.05in}\textrm{and} \hspace{0.05in} \mb{\Lambda}(\mb{B}_1^*+\mb{B}_2^*-\mb{K_X}+\mb{D}) = \mb{0}\}. \label{eq:Tangent6}
\end{align}
It now follows from (\ref{eq:nabla1}), (\ref{eq:Tangent4}), (\ref{eq:Tangent5}), and (\ref{eq:Tangent6}) that $\nabla h(\mb{b}^*)$ can be written as
\begin{align*}
\nabla h(\mb{b}^*) = \textrm{vec}\left(\mb{M}_1+\mb{\Lambda},\mb{M}_2+\mb{\Lambda}\right)
\end{align*}
for some $\mb{M}_1, \mb{M}_2,$ and $\mb{\Lambda}$ such that
\begin{align*}
\mb{M}_i \mb{B}^{*}_i &= \mb{0}, \hspace{0.1in} \textrm{for all} \hspace{0.1in} i \in \{1,2\}\\
\mb{\Lambda} (\mb{B}^{*}_1+\mb{B}^{*}_2 -\mb{K_X}+ \mb{D}) &= \mb{0}, \hspace {0.1in} \textrm{and}\\
\mb{M}_1, \mb{M}_2,\mb{\Lambda}  &\succcurlyeq \mb{0}.
\end{align*}
Lemma \ref{lem:ExistenceKKT} now follows because
\begin{align*}
\nabla h(\mb{b}^*) = \textrm{vec}\left( \frac{\mu}{2}(\mb{K_X}-\mb{B}^{*}_1-\mb{B}^{*}_2)^{-1}, \frac{\mu}{2}(\mb{K_X}-\mb{B}^{*}_1-\mb{B}^{*}_2)^{-1} -\frac{\mu}{2}(\mb{K_X}-\mb{B}^{*}_2)^{-1}+\frac{1}{2}(\mb{K_Y}-\mb{B}^{*}_2)^{-1} \right).
\end{align*}
\\
\Large
{\textbf{Appendix C:\hspace{0.1in} \textbf{Proof of Lemma \ref{lem:PosSemiDefDelta}} }\newline
\normalsize \\
Using (\ref{eq:Delta}), we obtain
\begin{align*}
\mb{\Delta}^{*} &= \frac{\mu}{2}(\mb{K_X}-\mb{B}^{*}_2)^{-1}- \mb{M}_1^{*}\\
&= (\mb{K_X}-\mb{B}^{*}_2)^{-1} \left[\frac{\mu}{2}(\mb{K_X}-\mb{B}^{*}_2)- (\mb{K_X}-\mb{B}^{*}_2)\mb{M}_1^{*}(\mb{K_X}-\mb{B}^{*}_2)\right](\mb{K_X}-\mb{B}^{*}_2)^{-1}.
\end{align*}
It is hence sufficient to show that
\begin{align*}
\frac{\mu}{2}(\mb{K_X}-\mb{B}^{*}_2)- (\mb{K_X}-\mb{B}^{*}_2)\mb{M}_1^{*}(\mb{K_X}-\mb{B}^{*}_2)
\end{align*}
is positive semidefinite. On pre- and post-multiplying (\ref{eq:KKT1}) by $\mb{K_X}-\mb{B}^{*}_1-\mb{B}^{*}_2$, we obtain
\begin{align}
\frac{\mu}{2}(\mb{K_X}-\mb{B}^{*}_1-\mb{B}^{*}_2)- (\mb{K_X}-\mb{B}^{*}_1-\mb{B}^{*}_2)(\mb{M}_1^{*}+\mb{\Lambda}^{*})(\mb{K_X}-\mb{B}^{*}_1-\mb{B}^{*}_2)=\mb{0}.\label{eq:AppenC1}
\end{align}
Using (\ref{eq:KKT3}) and (\ref{eq:KKT4}), we have
\begin{align}
(\mb{K_X}-\mb{B}^{*}_1-\mb{B}^{*}_2)\mb{M}_1^{*}(\mb{K_X}-\mb{B}^{*}_1-\mb{B}^{*}_2)&=(\mb{K_X}-\mb{B}^{*}_2)\mb{M}_1^{*}(\mb{K_X}-\mb{B}^{*}_2) \hspace{0.1in} \textrm{and} \label{eq:AppenC2}\\
(\mb{K_X}-\mb{B}^{*}_1-\mb{B}^{*}_2)\mb{\Lambda}^{*}(\mb{K_X}-\mb{B}^{*}_1-\mb{B}^{*}_2)&=\mb{D}\mb{\Lambda}^{*}\mb{D}.\label{eq:AppenC3}
\end{align}
Now (\ref{eq:AppenC1}) through (\ref{eq:AppenC3}) together imply that
\begin{align*}
\frac{\mu}{2}(\mb{K_X}-\mb{B}^{*}_2)- (\mb{K_X}-\mb{B}^{*}_2)\mb{M}_1^{*}(\mb{K_X}-\mb{B}^{*}_2)=\frac{\mu}{2}\mb{B}^{*}_1+\mb{D}\mb{\Lambda}^{*}\mb{D},
\end{align*}
which is a positive semidefinite matrix.

We next show that $\mb{\Delta}^{*}$ is nonzero. Suppose otherwise that $$\mb{\Delta}^{*}=\mb{0}.$$ This together with (\ref{eq:Delta}) implies that
\begin{align*}
\mb{M}_1^{*}&= \frac{\mu}{2}(\mb{K_X}-\mb{B}^{*}_2)^{-1} \succ \mb{0} \hspace{0.1in} \textrm{and}\\
\mb{M}_2^{*}&=\frac{1}{2}(\mb{K_Y}-\mb{B}^{*}_2)^{-1} \succ \mb{0},
\end{align*}
i.e., $\mb{M}_1^{*}$ and $\mb{M}_2^{*}$ are positive definite. It now follows from (\ref{eq:KKT3}) that
\[
\mb{B}^{*}_1=\mb{B}^{*}_2=\mb{0},
\]
which is a contradiction because $(\mb{0},\mb{0})$ is not feasible for the optimization problem $\left (P_{G2} \right)$ by~(\ref{eq:assumption}).\\ \\
\Large
{\textbf{Appendix D: \hspace{0.1in} \textbf{Proof of Lemma \ref{lem:SrcEnh1}} }\newline
\normalsize \\
It is clear by definition that $\tilde {\mb{B}}^{*}_1, \tilde {\mb{B}}^{*}_2, \tilde {\mb{M}}_1^{*},$ and $\tilde {\mb{M}}_2^{*}$ are positive semidefinite matrices. To prove (\ref{eq:SrcEnh1}), we use the first equality in (\ref{eq:Delta}) and obtain
\begin{align}
\mb{SS}^T&=\mb{\Delta}^{*}\nonumber\\
&=\frac{\mu}{2} \left(\mb{K_X}-\mb{B}^{*}_2 \right)^{-1} - \mb{M}_1^{*}\nonumber\\
&=\frac{\mu}{2} [\mb{S,T}]\Bigr([\mb{S,T}]^T\left(\mb{K_X}-\mb{B}^{*}_2 \right)[\mb{S,T}]\Bigr)^{-1}[\mb{S,T}]^T - \mb{M}_1^{*} \label{eq:AppenD1}\\
&=\frac{\mu}{2} [\mb{S,T}]\mat{\mb{S}^T \left(\mb{K_X}-\mb{B}^{*}_2\right) \mb{S} & \mb{0} \\ \mb{0} & \mb{T}^T \left(\mb{K_X}-\mb{B}^{*}_2\right) \mb{T}}^{-1}[\mb{S,T}]^T - \mb{M}_1^{*} \label{eq:AppenD2}\\
&=\frac{\mu}{2} [\mb{S,T}]\mat{\left(\mb{S}^T \left(\mb{K_X}-\mb{B}^{*}_2\right) \mb{S}\right)^{-1} & \mb{0} \\ \mb{0} & \left(\mb{T}^T \left(\mb{K_X}-\mb{B}^{*}_2\right) \mb{T}\right)^{-1}}[\mb{S,T}]^T - \mb{M}_1^{*}\nonumber\\
&=\frac{\mu}{2} \mb{S}\left(\mb{S}^T \left(\mb{K_X}-\mb{B}^{*}_2\right)\mb{S}\right)^{-1} \mb{S}^T + \frac{\mu}{2} \mb{T} \left(\mb{T}^T \left(\mb{K_X}-\mb{B}^{*}_2\right) \mb{T}\right)^{-1}\mb{T}^T - \mb{M}_1^{*},\label{eq:AppenD3}
\end{align}
where
\begin{enumerate}
\item[(\ref{eq:AppenD1})] follows because $[\mb{S,T}]$ is invertible, and
\item[(\ref{eq:AppenD2})] follows because $\mb{S}$ and $\mb{T}$ are cross $\left(\mb{K_X}-\mb{B}^{*}_2\right)$-orthogonal.
\end{enumerate}
On pre- and post-multiplying (\ref{eq:AppenD3}) by $\mb{S}^T \left(\mb{K_X}-\mb{B}^{*}_2\right)$ and $\left(\mb{K_X}-\mb{B}^{*}_2\right)\mb{S}$, respectively, and again using the fact that $\mb{S}$ and $\mb{T}$ are cross $\left(\mb{K_X}-\mb{B}^{*}_2\right)$-orthogonal, we obtain
\begin{align*}
\left(\mb{S}^T \left(\mb{K_X}-\mb{B}^{*}_2\right)\mb{S}\right)\left(\mb{S}^T \left(\mb{K_X}-\mb{B}^{*}_2\right)\mb{S}\right) &=\frac{\mu}{2} \left(\mb{S}^T \left(\mb{K_X}-\mb{B}^{*}_2\right)\mb{S}\right)  - \mb{S}^T \left(\mb{K_X}-\mb{B}^{*}_2\right)\mb{M}_1^{*} \left(\mb{K_X}-\mb{B}^{*}_2\right)\mb{S},
\end{align*}
which is equivalent to
\begin{align}
\mb{I}_r &=\frac{\mu}{2} \left(\mb{S}^T \left(\mb{K_X}-\mb{B}^{*}_2\right)\mb{S}\right)^{-1} - \left(\mb{S}^T\left(\mb{K_X}-\mb{B}^{*}_2\right)\mb{S}\right)^{-1}\mb{S}^T \left(\mb{K_X}-\mb{B}^{*}_2\right)\mb{M}_1^{*}\left(\mb{K_X}-\mb{B}^{*}_2\right)\mb{S}\left(\mb{S}^T\left(\mb{K_X}-\mb{B}^{*}_2\right)\mb{S}\right)^{-1}.\label{eq:AppenD4}
\end{align}
Similarly, using the second equality in (\ref{eq:Delta}) together with the facts that $[\mb{S,W}]$ is invertible and $\mb{S}$ and $\mb{W}$ are cross $\left(\mb{K_Y}-\mb{B}^{*}_2\right)$-orthogonal, we obtain
\begin{align}
\mb{I}_r &=\frac{1}{2} \left(\mb{S}^T \left(\mb{K_Y}-\mb{B}^{*}_2\right)\mb{S}\right)^{-1} - \left(\mb{S}^T\left(\mb{K_Y}-\mb{B}^{*}_2\right)\mb{S}\right)^{-1}\mb{S}^T \left(\mb{K_Y}-\mb{B}^{*}_2\right)\mb{M}_2^{*}\left(\mb{K_Y}-\mb{B}^{*}_2\right)\mb{S}\left(\mb{S}^T\left(\mb{K_Y}-\mb{B}^{*}_2\right)\mb{S}\right)^{-1}.\label{eq:AppenD5}
\end{align}
Now (\ref{eq:AppenD4}) and (\ref{eq:AppenD5}) together can be written as
\begin{align}
\mb{I}_r &=\frac{\mu}{2} \bigr(\mb{K}_{\tilde{\mb{X}}}-\tilde{\mb{B}}^{*}_2\bigr)^{-1} - \tilde {\mb{M}}_1^{*}=\frac{1}{2} \bigr( \mb{K}_{\tilde{\mb{Y}}}-\tilde{\mb{B}}^{*}_2\bigr)^{-1} - \tilde {\mb{M}}_2^{*}. \label{eq:AppenD6}
\end{align}
This proves (\ref{eq:SrcEnh1}).

To prove (\ref{eq:SrcEnh2}), we have from (\ref{eq:KKT3}) and (\ref{eq:EigDecom2}) that
\[
\mb{B}^{*}_1 \mb{a}_i=\mb{0},
\]
for all $i$ in $\{1,2,\dots,p\}$. Since the columns of $\mb{T}$ are in $\textrm{span}\{\mb{a}_i\}_{i=1}^p$, we have
\[
\mb{B}^{*}_1 \mb{T} =\mb{0}.
\]
This and (\ref{eq:KKT3}) together imply
\[
\mb{B}^{*}_1 \left(\mb{M}^{*}_1- \frac{\mu}{2} \mb{T} \left(\mb{T}^T \left(\mb{K_X}-\mb{B}^{*}_2\right) \mb{T}\right)^{-1}\mb{T}^T \right) =\mb{0}.
\]
We now use (\ref{eq:AppenD3}) and obtain
\[
\mb{B}^{*}_1 \left(\frac{\mu}{2} \mb{S}\left(\mb{S}^T \left(\mb{K_X}-\mb{B}^{*}_2\right)\mb{S}\right)^{-1} \mb{S}^T  - \mb{SS}^T \right) =\mb{0},
\]
which can be re-written as
\[
\mb{B}^{*}_1 \mb{S}\left(\frac{\mu}{2} \bigr(\mb{K}_{\tilde{\mb{X}}}-\tilde{\mb{B}}^{*}_2\bigr)^{-1} - \mb{I}_r\right)\mb{S}^T =\mb{0}.
\]
Using the first equality in (\ref{eq:AppenD6}) yields
\[
\mb{B}^{*}_1 \mb{S} \tilde{\mb{M}}_1^{*} \mb{S}^T =\mb{0}.
\]
We next invoke Lemma \ref{lem:LinearAlgFact}(b) to obtain
\[
\textrm{Tr} \bigr(\mb{B}^{*}_1 \mb{S} \tilde{\mb{M}}_1^{*} \mb{S}^T \bigr) ={0}.
\]
Using Lemma \ref{lem:LinearAlgFact}(a) gives
\[
\textrm{Tr} \bigr( \mb{S}^T \mb{B}^{*}_1 \mb{S} \tilde{\mb{M}}_1^{*}\bigr) ={0},
\]
which is equivalent to
\[
\textrm{Tr} \bigr(\tilde{\mb{B}}_1^{*}\tilde{\mb{M}}_1^{*}\bigr)={0}.
\]
Since $\tilde{\mb{B}}_1^{*}$ and $\tilde{\mb{M}}_1^{*}$ are positive semidefinite, by invoking Lemma \ref{lem:LinearAlgFact}(b) again, we obtain
\[
\tilde{\mb{B}}_1^{*}\tilde{\mb{M}}_1^{*}=\mb{0}.
\]
The proof of
\[
\tilde{\mb{B}}_2^{*}\tilde{\mb{M}}_2^{*}=\mb{0}.
\]
is exactly similar. This proves (\ref{eq:SrcEnh2}). The proof of (\ref{eq:SrcEnh3}) is immediate from (\ref{eq:ActiveDistConst}). \\ \\
\Large
{\textbf{Appendix E: \hspace{0.1in} \textbf{Proof of Lemma \ref{lem:SrcEnh2}} }\newline
\normalsize \\
The proofs of (\ref{eq:SrcEnh6}) and (\ref{eq:SrcEnh7}) are easy. They follow from (\ref{eq:SrcEnh1}), (\ref{eq:SrcEnh4}), and (\ref{eq:SrcEnh5}). Since $\mu > 1$, (\ref{eq:SrcEnh6}) and (\ref{eq:SrcEnh7}) imply that $$\mb{K}_{\hat{\mb{X}}} \succ \mb{K}_{\hat{\mb{Y}}}.$$ $\mb{K}_{\tilde{\mb{X}}}$ and $\mb{K}_{\tilde{\mb{Y}}}$ are positive definite by definition. Since $\tilde {\mb{M}}_1^{*}$ and $\tilde {\mb{M}}_2^{*}$ are positive semidefinite,
\begin{align*}
\mb{K}_{\hat{\mb{X}}} &\succcurlyeq \mb{K}_{\tilde{\mb{X}}} \hspace{0.1in} \textrm{and} \\
\mb{K}_{\hat{\mb{Y}}} &\succcurlyeq \mb{K}_{\tilde{\mb{Y}}}
\end{align*}
follow from (\ref{eq:SrcEnh4}) and (\ref{eq:SrcEnh5}), respectively. This proves (\ref{eq:SrcEnh8}) and (\ref{eq:SrcEnh9}). To prove (\ref{eq:SrcEnh10}), we have
\begin{align}
\frac{\bigr|\mb{K}_{\tilde{\mb{Y}}}\bigr|}{\bigr|\mb{K}_{\tilde{\mb{Y}}}-\tilde{\mb{B}}^{*}_2\bigr|} &=\frac{\bigr|\mb{K}_{\tilde{\mb{Y}}}-\tilde{\mb{B}}^{*}_2+\tilde{\mb{B}}^{*}_2\bigr|}{\bigr|\mb{K}_{\tilde{\mb{Y}}}-\tilde{\mb{B}}^{*}_2\bigr|}\nonumber\\
&=\frac{\bigr|\mb{I}_r+\tilde{\mb{B}}^{*}_2\bigr(\mb{K}_{\tilde{\mb{Y}}}-\tilde{\mb{B}}^{*}_2\bigr)^{-1}\bigr|}{\bigr|\mb{I}_r\bigr|}\nonumber\\
&=\frac{\bigr|\mb{I}_r+\tilde{\mb{B}}^{*}_2\bigr[\bigr(\mb{K}_{\tilde{\mb{Y}}}-\tilde{\mb{B}}^{*}_2\bigr)^{-1}-2\tilde {\mb{M}}_2^{*}\bigr]\bigr|}{\bigr|\mb{I}_r\bigr|} \label{eq:AppenE1}\\
&=\frac{\bigr|\mb{I}_r+\tilde{\mb{B}}^{*}_2\bigr(\mb{K}_{\hat{\mb{Y}}}-\tilde{\mb{B}}^{*}_2\bigr)^{-1}\bigr|}{\bigr|\mb{I}_r\bigr|} \label{eq:AppenE2}\\
&= \frac{\bigr|\mb{K}_{\hat{\mb{Y}}}\bigr|}{\bigr|\mb{K}_{\hat{\mb{Y}}}-\tilde{\mb{B}}^{*}_2\bigr|},\nonumber
\end{align}
where
\begin{enumerate}
\item[(\ref{eq:AppenE1})] follows from (\ref{eq:SrcEnh2}), and
\item[(\ref{eq:AppenE2})] follows from (\ref{eq:SrcEnh5}).
\end{enumerate}
To prove (\ref{eq:SrcEnh11}), we proceed similarly and obtain
\begin{align}
\frac{\bigr|\mb{K}_{\tilde{\mb{X}}}-\tilde{\mb{B}}^{*}_2\bigr|}{\bigr|\mb{K}_{\tilde{\mb{X}}}-\tilde{\mb{B}}^{*}_1-\tilde{\mb{B}}^{*}_2\bigr|}
&=\frac{\bigr|\mb{I}_r\bigr|}{\bigr|\mb{I}_r - \tilde{\mb{B}}^{*}_1 \bigr(\mb{K}_{\tilde{\mb{X}}}-\tilde{\mb{B}}^{*}_2\bigr)^{-1}\bigr|} \nonumber\\
&=\frac{\bigr|\mb{I}_r\bigr|}{\bigr|\mb{I}_r - \tilde{\mb{B}}^{*}_1 \bigr[\bigr(\mb{K}_{\tilde{\mb{X}}}-\tilde{\mb{B}}^{*}_2\bigr)^{-1}-\frac{2}{\mu}\tilde {\mb{M}}_1^{*}\bigr]\bigr|} \label{eq:AppenE3}\\
&=\frac{\bigr|\mb{I}_r\bigr|}{\bigr|\mb{I}_r - \tilde{\mb{B}}^{*}_1 \bigr(\mb{K}_{\hat{\mb{X}}}-\tilde{\mb{B}}^{*}_2\bigr)^{-1}\bigr|}\label{eq:AppenE4}\\
&=\frac{\bigr|\mb{K}_{\hat{\mb{X}}}-\tilde{\mb{B}}^{*}_2\bigr|}{\bigr|\mb{K}_{\hat{\mb{X}}}-\tilde{\mb{B}}^{*}_1-\tilde{\mb{B}}^{*}_2\bigr|},\nonumber
\end{align}
where
\begin{enumerate}
\item[(\ref{eq:AppenE3})] follows from (\ref{eq:SrcEnh2}), and
\item[(\ref{eq:AppenE4})] follows from (\ref{eq:SrcEnh4}). \\
\end{enumerate}
\Large
{\textbf{Appendix F:}\hspace{0.1in} \textbf{Proof of Lemma \ref{lem:Oohama1}}}\newline
\normalsize \\
We have
\begin{align}
h\bigr (\hat{\mb{X}} | \mb{U,V} \bigr ) &\le \frac{1}{2} \log \bigr ( \left (2 \pi e \right)^r \bigr |\mb{K}_{\hat{\mb{X}}|\mb{U,V}} \bigr| \bigr ) \label{eq:AppenF1}\\
&\le \frac{1}{2} \log \bigr ( \left (2 \pi e \right)^r \bigr |\hat{\mb{D}} \bigr| \bigr ),\label{eq:AppenF2}
\end{align}
where
\begin{enumerate}
\item[(\ref{eq:AppenF1})] follows from the fact the Gaussian distribution maximizes the differential entropy for a given covariance matrix \cite[Theorem 8.6.5]{Cover}, and
\item[(\ref{eq:AppenF2})] follows from the distortion constraint in the definition of $(\hat P_1)$ and the concavity of $\log |\cdot|$ function.
\end{enumerate}
Inequalities (\ref{eq:AppenF1}) and (\ref{eq:AppenF2}) are equalities if $\hat {\mb{X}}, \mb{U}$, and $\mb{V}$ are jointly Gaussian with the conditional covariance matrix $\mb{K}_{\hat{\mb{X}}|\mb{U,V}}$ such that
\begin{align}
\mb{K}_{\hat{\mb{X}}|\mb{U,V}} = \hat {\mb{D}} = \mb{K}_{\hat{\mb{X}}}-\tilde{\mb{B}}_1^{*}-\tilde{\mb{B}}_1^{*}, \label{eq:AppenF3}
\end{align}
where the last equality follows from (\ref{eq:SrcEnh12}). We thus conclude that a Gaussian $(\mb{U,V})$ with the conditional covariance matrix satisfying (\ref{eq:AppenF3}) is optimal for the subproblem $(\hat P_1)$, and the optimal value is
\begin{align}
v(\hat P_1) &= \mu h\bigr (\hat{\mb{X}} \bigr ) - \frac{\mu}{2} \log \bigr ( \bigr (2 \pi e \bigr)^r \bigr |\hat{\mb{D}} \bigr| \bigr ) \nonumber\\
&= \frac{\mu}{2} \log \bigr ( \bigr (2 \pi e \bigr)^r \bigr |\mb{K}_{\hat{\mb{X}}}\bigr | \bigr ) - \frac{\mu}{2} \log \bigr ( \bigr (2 \pi e \bigr)^r \bigr |\hat{\mb{D}} \bigr| \bigr ) \nonumber\\
&= \frac{\mu}{2} \log \frac{\bigr | \mb{K}_{\hat{\mb{X}}} \bigr |}{\bigr | \hat{\mb{D}}  \bigr |} \nonumber.
\end{align}
\Large
{\textbf{Appendix G:}\hspace{0.1in} \textbf{Proof of Lemma \ref{lem:Oohama2}}}\newline
\normalsize \\
Since conditioned on $\mb{V}$, $\hat{\mb{Y}}$ and $\hat{\mb{N}}$ are independent, we use the vector EPI \cite[Theorem 17.7.3]{Cover} to obtain
\begin{align}
h(\hat{\mb{Y}}|\mb{V}) - \mu h(\hat{\mb{X}}|\mb{V}) &= h(\hat{\mb{Y}}|\mb{V}) - \mu h(\hat{\mb{Y}}+\hat{\mb{N}}|\mb{V}) \nonumber\\
&\le h(\hat{\mb{Y}}|\mb{V}) - \frac{\mu r}{2} \log \bigr (2^{\frac{2}{r} h(\hat{\mb{Y}}|\mb{V})} + 2^{\frac{2}{m} h(\hat{\mb{N}})} \bigr ).\label{eq:AppenG1}
\end{align}
The inequality (\ref{eq:AppenG1}) is equality if $\hat{\mb{Y}}$ and $\mb{V}$ are jointly Gaussian and the conditioned covariance matrix $$\mb{K}_{\hat {\mb{Y}}|\mb{V}} = a\mb{K}_{\hat {\mb{N}}},$$ for some constant $a > 0$. By following standard calculus arguments, we can show that for $\mu > 1$ the right-hand side of (\ref{eq:AppenG1}) is concave in $h(\hat{\mb{Y}}|\mb{V})$ and has a global maximum at
\begin{align}
h(\hat{\mb{Y}}|\mb{V}) = h(\hat{\mb{N}}) - \frac{r}{2} \log(\mu -1).\label{eq:AppenG2}
\end{align}
Let $\mb{V}_G$ and $\hat{\mb{Y}}$ be jointly Gaussian such that the conditional covariance matrix of $\hat{\mb{Y}}$ given $\mb{V}_G$ is
\begin{align*}
\mb{K}_{\hat{\mb{Y}}|\mb{V}_G}=\mb{K}_{\hat{\mb{Y}}}-\tilde{\mb{B}}_2^{*}.
\end{align*}
We next show that this $\mb{V}_G$ achieves equality in (\ref{eq:AppenG1}) and satisfies (\ref{eq:AppenG2}) simultaneously. We have from (\ref{eq:SrcEnh6}) and (\ref{eq:SrcEnh7}) that
\begin{align}
\mb{K}_{\hat{\mb{Y}}}-\tilde{\mb{B}}^{*}_2 = (\mu -1)^{-1} \mb{K}_{\hat {\mb{N}}}, \label{eq:AppenG3}
\end{align}
i.e., the conditional covariance matrix $\mb{K}_{\hat{\mb{Y}}|\mb{V}_G}$ is proportional to $\mb{K}_{\hat {\mb{N}}}$. Hence, (\ref{eq:AppenG1}) is satisfied with equality. Moreover, for this $\mb{V}_G$, (\ref{eq:AppenG2}) and (\ref{eq:AppenG3}) are equivalent. Therefore,
\begin{align*}
h(\hat{\mb{Y}}|\mb{V}) - \mu h(\hat{\mb{X}}|\mb{V}) \le \frac{1}{2} \log \bigr ( \left (2 \pi e \right)^r \bigr |\mb{K}_{\hat{\mb{Y}}}-\tilde{\mb{B}}_2^{*} \bigr| \bigr )- \frac{\mu}{2}\log \bigr ( \left (2 \pi e \right)^r \bigr |\mb{K}_{\hat{\mb{X}}}-\tilde{\mb{B}}_2^{*} \bigr| \bigr ).
\end{align*}
We thus conclude that $\mb{V}_G$ is optimal for $(\hat P_2)$ and the optimal value is
\begin{align*}
v(\hat P_2) &= \mu h \bigr (\hat {\mb{X}}\bigr) - h \bigr (\hat {\mb{Y}}\bigr) + \frac{1}{2} \log \bigr ( \left (2 \pi e \right)^r \bigr |\mb{K}_{\hat{\mb{Y}}}-\tilde{\mb{B}}_2^{*} \bigr| \bigr )- \frac{\mu}{2}\log \bigr ( \left (2 \pi e \right)^r \bigr |\mb{K}_{\hat{\mb{X}}}-\tilde{\mb{B}}_2^{*} \bigr| \bigr )\\
&= \frac{\mu}{2}\log \bigr ( \left (2 \pi e \right)^r \bigr |\mb{K}_{\hat{\mb{X}}}\bigr| \bigr )-\frac{1}{2} \log \bigr ( \left (2 \pi e \right)^r \bigr |\mb{K}_{\hat{\mb{Y}}} \bigr| \bigr )+ \frac{1}{2} \log \bigr ( \left (2 \pi e \right)^r \bigr |\mb{K}_{\hat{\mb{Y}}}-\tilde{\mb{B}}_2^{*} \bigr| \bigr )- \frac{\mu}{2}\log \bigr ( \left (2 \pi e \right)^r \bigr |\mb{K}_{\hat{\mb{X}}}-\tilde{\mb{B}}_2^{*} \bigr| \bigr )\\
&= \frac{\mu}{2}\log \frac{\bigr |\mb{K}_{\hat{\mb{X}}}\bigr|}{\bigr |\mb{K}_{\hat{\mb{X}}}-\tilde{\mb{B}}_2^{*} \bigr|} - \frac{1}{2} \log \frac{\bigr |\mb{K}_{\hat{\mb{Y}}}\bigr|}{\bigr |\mb{K}_{\hat{\mb{Y}}}-\tilde{\mb{B}}_2^{*} \bigr|}.
\end{align*}

\end{document}